\documentclass[10pt, reqno]{amsart}
\usepackage{fullpage}
\usepackage{amsmath, amsthm}
\usepackage{amscd}
\usepackage{enumerate}
\usepackage{float,amsmath,amssymb,mathrsfs,bm,multirow,graphics}
\usepackage[dvips]{graphicx}
\usepackage[percent]{overpic}
\usepackage[numbers,sort&compress]{natbib}
\usepackage{xcolor}
\usepackage{todonotes}
\usepackage{amsaddr}
\usepackage{mathtools}
\usepackage{subcaption}
\usepackage{mathtools}
\usepackage{stackengine} 
 
\usepackage{tikz}
\usepackage{tikz-cd}
\usepackage{tqft}
\usetikzlibrary{decorations.pathreplacing,decorations.markings}

\newtheorem{theorem}{Theorem}
\newtheorem{proposition}{Proposition}[section]

\newtheorem{corollary}[proposition]{Corollary}
\newtheorem{definition}[proposition]{Definition}
\newtheorem{assumption}[proposition]{Assumption}
\newtheorem{remark}[proposition]{Remark}

\newcommand\Cb{\mathbb{C}} 

\newcommand\Pb{\mathbb{P}} 

\DeclareMathOperator{\diag}{diag}

\newcommand{\im}{\text{\upshape Im} \,}

 \newcommand\ben{\begin{equation*}}
 \newcommand\ebn{\end{equation*}}
 \newcommand\beq{\begin{equation}}
 \newcommand\eeq{\end{equation}}
 \newcommand\lb{\left(}
  \newcommand\rb{\right)} 

\numberwithin{equation}{section}

\allowdisplaybreaks[1]

\usepackage[colorlinks=true]{hyperref}
\hypersetup{urlcolor=blue, citecolor=red, linkcolor=blue}

\begin{document}
\title[Confluent conformal blocks]{Confluent conformal blocks of the second kind}

\author{Jonatan Lenells and Julien Roussillon}
\address{Department of Mathematics, KTH Royal Institute of Technology, \\ 100 44 Stockholm, Sweden.}
\email{julienro@kth.se}
\email{jlenells@kth.se}

\begin{abstract} 
We construct confluent conformal blocks of the second kind of the Virasoro algebra. 
We also construct the Stokes transformations which map such blocks in one Stokes sector to another. 
In the BPZ limit, we verify explicitly that the constructed blocks and the associated Stokes transformations reduce to solutions of the confluent BPZ equation and its Stokes matrices, respectively.
Both the confluent conformal blocks and the Stokes transformations are constructed by taking suitable confluent limits of the crossing transformations of the four-point Virasoro conformal blocks.\end{abstract}

\maketitle


\tableofcontents

\section{Introduction}
Virasoro conformal blocks \cite{BPZ} are holomorphic building blocks of correlations functions in two-dimensional conformal field theories. The AGT correspondence \cite{AGT} triggered the study of  so-called irregular conformal blocks \cite{G09,GT,BMT}. Such blocks arise when singularities of the Virasoro conformal blocks merge, thus giving rise to irregular singularities. 
The present work studies one of the simplest classes of irregular conformal blocks, namely, the blocks that arise when two regular singularities of the four-point Virasoro conformal block merge into an irregular singularity of rank one. 

In the so-called BPZ limit \cite{BPZ}, the four-point Virasoro conformal blocks degenerate to solutions of a hypergeometric equation with three regular singular points in the complex plane. Since the equation is second order, the solution space is two-dimensional. For each regular singularity, it is possible to choose a basis for the solution space which diagonalizes the corresponding monodromy matrix. More precisely, assuming without loss of generality that the three regular singularities are located at $0$, $1$, and $\infty$, there are vectors $\boldsymbol{F}^{p}(z)= (F_+^{p}(z), F_-^{p}(z))$, $p = 0, 1, \infty$, such that $\{F_+^p, F_-^p\}$ forms a basis of solutions for each $p$ and 
\begin{align}\label{monodromyintro}
\boldsymbol{F}^{0}(e^{2i \pi} z) = M_0 \boldsymbol{F}^{0}(z), \quad
\boldsymbol{F}^{1}(1+e^{2i \pi} z) = M_1 \boldsymbol{F}^{1}(1+ z), \quad
\boldsymbol{F}^{\infty}(e^{-2i \pi} z^{-1}) = M_\infty \boldsymbol{F}^{\infty}(z^{-1}),
\end{align}
where the three monodromy matrices $M_p$, $p = 0, 1, \infty$, are diagonal. The solutions $F_\pm^{p}(z)$ can be expressed in terms of hypergeometric functions. The bases $\boldsymbol{F}^{p}$, $p = 0, 1, \infty$, are known as the $s$-channel, $t$-channel, and $u$-channel degenerate conformal blocks, respectively, see e.g. \cite{Ribault}. In a similar way, the general (i.e. nondegenerate) $s$-channel, $t$-channel, and $u$-channel conformal blocks form infinite-dimensional bases for the space of four-point conformal blocks. 
The purpose of this paper is to describe what happens to these bases as the regular singular point at $1$ tends to infinity and merges with the regular singular point at $\infty$ to form an irregular singularity. 
In this limit, the hypergeometric equation degenerates into the confluent hypergeometric equation; we therefore call it the {\it confluent} limit and the resulting conformal blocks {\it confluent conformal blocks}.

In order to describe our results, it is convenient to first consider the hypergeometric case.  
The solutions $F_\pm^{p}(z)$, $p = 0, 1, \infty$, of the hypergeometric BPZ equation are most easily constructed by means of the Frobenius method. To apply this method in the case of $p = 0$ for example, one substitutes the ansatz $z^\alpha(1 + a_1z + a_2z^2 + \cdots)$ into the equation and equates coefficients of powers of $z$. It follows that there are two possible values $\alpha_\pm$ of $\alpha$, and that the associated power series coefficients $a_j^\pm$ can be determined recursively ($\alpha_\pm$ are the two solutions of the indicial equation and our assumptions on the parameters will be such that $\alpha_+ \neq \alpha_-$). This yields two solutions $F_\pm^{0}(z) = z^{\alpha_\pm}(1 + \sum_{j=1}^\infty a_j^\pm z^j)$, which when combined into the vector $\boldsymbol{F}^{0} = (F_+^{0}(z), F_-^{0}(z))$ satisfy the desired relation in (\ref{monodromyintro}) with the diagonal monodromy matrix $M_0  = \text{diag}(e^{2i\pi \alpha_+}, e^{2i\pi \alpha_-})$. We emphasize that this construction relies on the fact that the two power series $\sum_{j=1}^\infty a_j^\pm z^j$ converge in a neighborhood of $z = 0$. 

In the confluent limit, the hypergeometric equation degenerates into the confluent hypergeometric equation, which has a regular singular point at $0$ and an irregular singular point at $\infty$. 
Since $0$ is a regular singular point, a basis of solutions $\boldsymbol{B}(t) = (B_+(t), B_-(t))$ which diagonalizes the monodromy matrix at $0$ can be constructed with the help of the Frobenius method, as in the nonconfluent case. The solutions $B_\pm(t)$ can be expressed in terms of the confluent hypergeometric function of the first kind (also known as Kummer's function) and we therefore call them {\it degenerate confluent conformal blocks of the first kind}.

Since the singular point at $\infty$ is irregular, it is not possible to construct solutions near $\infty$ in the same way. In fact, if one substitutes the ansatz $t^\alpha e^{\beta t} (1 + d_1t^{-1} + d_2t^{-2} + \cdots)$ into the confluent equation and equates coefficients of powers of $t$, one finds that there are two possible choices $(\alpha_+, \beta_+)$ and $(\alpha_-, \beta_-)$ of the parameters $(\alpha, \beta)$. Given one of these choices, the associated power series coefficients $d_j^\pm$ can be computed recursively. But in contrast to the regular case, the two power series $\sum_{j=1}^\infty d_j^\pm t^{-j}$, in general, do not converge anywhere in the complex $t$-plane. Thus it is not possible to construct solutions in this way. The best one can do is to find solutions whose asymptotic behavior as $t \to \infty$ is given by these power series, see e.g. \cite{Ramis}. These solutions can be expressed in terms of the confluent hypergeometric function of the second kind (also known as Tricomi's function) and we therefore call them {\it degenerate confluent conformal blocks of the second kind}. Since the asymptotic expansion of Tricomi's function is only valid in a certain sector $a < \arg t < b$ of the complex $t$-plane, an infinite sequence of bases of solutions $\boldsymbol{D}_n(t)$, $n \in \mathbb{Z}$, are needed to cover all values of $\arg t$. This is an example of the Stokes phenomenon. Thus each basis $\boldsymbol{D}_n(t)$ has the desired asymptotics only in a sector $\Omega_n$ of the complex $t$-plane:
$$\boldsymbol{D}_n(t) \sim \boldsymbol{D}_{\text{asymp}}(t), \qquad t\to \infty, ~ t \in \Omega_n,$$
where $\boldsymbol{D}_{\text{asymp}}(t)$ denotes the formal asymptotic series 
\begin{align}\label{Dasympintro}
\boldsymbol{D}_{\text{asymp}}(t) = \begin{pmatrix} t^{\alpha_+} e^{\beta_+ t} \sum_{j=1}^\infty d_j^+ t^{-j} \\ t^{\alpha_-} e^{\beta_- t} \sum_{j=1}^\infty d_j^- t^{-j} \end{pmatrix}.
\end{align}
With our conventions, the Stokes sectors $\Omega_n$ are given by (see Figure \ref{fig:figstokes})
\begin{align}\label{stokessector}
\Omega_n=\left \{-\frac{3\pi}{2}+\pi(n-1) < \arg t<\frac{\pi}{2}+\pi(n-1)\right \}, \qquad n \in \mathbb{Z}.
\end{align}
\begin{figure}[h!]
\subcaptionbox*{$\; \Omega_1$}{\begin{tikzpicture}
\fill[fill=gray!90] (0,0) circle (1cm);
\draw[line width=0.7mm, white] (0,0)--(0,1);
\filldraw (0,0) circle (1pt);
\end{tikzpicture}}\hspace{3cm}
\subcaptionbox*{$\; \Omega_2$}{\begin{tikzpicture}
\fill[fill=gray!90] (0,0) circle (1cm);
\draw[line width=0.7mm, white] (0,0)--(0,-1);
\filldraw (0,0) circle (1pt);
\end{tikzpicture}}\hspace{3cm}
\subcaptionbox*{$\; \Omega_3$}{\begin{tikzpicture}
\fill[fill=gray!90] (0,0) circle (1cm);
\draw[line width=0.7mm, white] (0,0)--(0,1);
\filldraw (0,0) circle (1pt);
\end{tikzpicture}}
\caption{The Stokes sectors $\Omega_n$, $n=1,2,3$, in the complex $t$-plane.}
\label{fig:figstokes}
\end{figure}
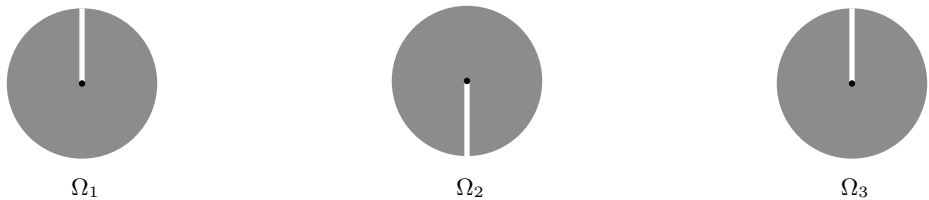

Of particular interest are the connection matrices $C_n$ and the Stokes matrices $S_n$ which, by definition,  are the unique $2\times 2$ matrices such that
\begin{align}\label{DnSnCnBPZ}
\boldsymbol{D}_n(t)=C_n \boldsymbol{B}(t), \quad
\boldsymbol{D}_{n+1}(t)=S_n \boldsymbol{D}_n(t), \qquad t \in \Omega_n, ~ n \in \mathbb{Z}.
\end{align}
These matrices can be obtained by taking appropriate confluent limits of the connection matrices which relate the three bases $\boldsymbol{F}^{p}(z)$ of the hypergeometric BPZ equation.

The goal of this paper is to show that the above description of the hypergeometric BPZ limit can be generalized to the setting of nondegenerate conformal blocks. 
Whereas the space of solutions of the confluent hypergeometric equation is two-dimensional, the space of confluent conformal blocks is infinite-dimensional. 
The infinite-dimensional analogs of the bases $\boldsymbol{B}(t)$ and $\boldsymbol{D}_n(t)$, $n \in \mathbb{Z}$, will be denoted by $\mathcal{B}(t) \equiv \mathcal B\lb \theta_*;\sigma;\substack{\theta_t \\ \theta_0};t\rb$ and $\mathcal{D}_{n}(t) \equiv \mathcal{D}_{n}\lb\substack{\theta_t\\ \theta_*};\nu;\theta_0;t\rb$, respectively. Here $\theta_0, \theta_t, \theta_*$ are parameters characterizing the conformal dimensions of the fields entering the correlation function, while $\sigma$ and $\nu$ are continuous indices labeling the infinite set of basis elements. 
Elements of the basis $\mathcal{B}(t)$ will be called {\it confluent conformal blocks of the first kind}. Up to prefactors of the form $t^\alpha e^{\beta t}$, these blocks can be represented as power series in $t$ which are conjectured to converge in the whole complex plane \cite{LNR,GT,Nagoya}.
Elements of the bases $\mathcal{D}_{n}(t)$ will be called {\it confluent conformal blocks of the second kind}. These blocks are characterized by the fact that they admit a particular asymptotic expansion $\mathcal D_{\text{asymp}} \lb\substack{\theta_t\\ \theta_*};\nu;\theta_0;t\rb \equiv t^{\alpha} e^{\beta t} \sum_{j=1}^\infty d_j t^{-j}$ in the Stokes sectors $\Omega_n$ as $t$ approaches the irregular singularity. The power series part of this expansion is believed to diverge everywhere in the complex $t$-plane and no closed formula is known for its coefficients. The formalism of irregular vertex operators developed in \cite{Nagoya} provides a recursive method to compute the coefficients of the series. A different but equivalent approach which relies on the computation of a term-by-term limit of the $u$-channel conformal blocks series expansion was proposed in \cite{GT,LNR}.

In this paper, we will take a different approach to the construction of the confluent conformal blocks of the second kind $\mathcal{D}_n(t)$. 
As mentioned above, the infinite-dimensional analogs of the bases $\boldsymbol{F}^{p}(z)$, $p = 0, 1, \infty$, are the $s$-channel, $t$-channel, and $u$-channel conformal blocks, which we denote by
$$\mathcal{F}^0(z) \equiv \mathcal F\lb\substack{\theta_{1}\;\quad \theta_{t}\\ \sigma_s \\ \theta_{\infty}\quad \theta_0};z\rb, \quad
\mathcal{F}^1(z) \equiv \mathcal F\lb\substack{\theta_{\infty}\;\quad 
          \theta_{t}\\ \sigma_t \\ \theta_0 \quad \theta_1};1- \frac{1}{z}\rb, \quad \text{and} \quad
\mathcal{F}^\infty(z) \equiv \mathcal F\lb\substack{\theta_{1}\;\quad \theta_{t}\\ \sigma_u \\ \theta_0 \quad \theta_\infty};\frac1z\rb,
$$
respectively. These blocks are related by crossing transformations which can be viewed as infinite-dimensional analogs of the connection matrices for the hypergeometric BPZ equation. By taking appropriate confluent limits of these transformations, we will show that the relations in (\ref{DnSnCnBPZ}) admit the following generalizations in the context of nondegenerate conformal blocks:
\begin{align}\label{DnSnCn}
\mathcal{D}_n(t)= \mathcal{C}_n \mathcal{B}(t), \quad
\mathcal{D}_{n+1}(t)=  \mathcal{S}_n \mathcal{D}_n(t), \qquad t \in \Omega_n, ~ n \in \mathbb{Z},
\end{align}
where $\mathcal{S}_n$ and $\mathcal{C}_n$ are integral operators. The first of these relations, established in Theorem \ref{mainth1}, provides a construction of the confluent conformal blocks of the second kind. Abusing notation and denoting the kernel by the same symbol as the operator, this relation can be written in more detail as
\begin{align*}
{\mathcal D}_{n}\lb\substack{\theta_t\\ \theta_*};\nu_n ;\theta_0;t\rb= \displaystyle \int_0^{+\infty} d\sigma_s~\mathcal{C}_{n}\left[\substack{\theta_t\vspace{0.08cm} \\ \theta_{*}\;\;\;\theta_0};\substack{\nu_n \vspace{0.15cm} \\  \sigma_s}\right]  \mathcal B\lb \theta_*;\sigma_s;\substack{\theta_t \\ \theta_0};t\rb, \qquad t \in \Omega_{n}, ~ n \in \mathbb{Z},
\end{align*}
where the integral kernel $\mathcal{C}_{n}$ will be computed explicitly, see equation (\ref{gnm}). The second relation in (\ref{DnSnCn}), established in Theorem \ref{mainth2},  provides an infinite-dimensional generalization of the famous Stokes phenomenon. It describes how the confluent conformal blocks of the second kind in different Stokes sectors are related to each other.
In more detail, this relation can be written as 
\begin{align*}
{\mathcal D}_{n+1}\lb\substack{\theta_t\\ \theta_*};\nu_{n+1};\theta_0;t\rb=\displaystyle \int_{-\infty}^{+\infty} d\nu_n ~ \mathcal{S}_{n}\left[\substack{\theta_t\vspace{0.08cm} \\ \theta_{*}\;\;\;\theta_0};\substack{\nu_{n+1}\vspace{0.15cm} \\  \nu_n}\right] {\mathcal D}_{n}\lb\substack{\theta_t\\ \theta_*};\nu_n;\theta_0;t\rb, \qquad t \in \Omega_n \cap \Omega_{n+1},  ~ n \in \mathbb{Z},
\end{align*}
where the kernels $\mathcal{S}_{n}$ will be computed explicitly by taking suitable confluent limits of the Virasoro fusion kernel, see equation (\ref{stokessn}). In analogy with the finite-dimensional case, we refer to these transformations as \textit{Stokes transformations} and to the corresponding kernels $\mathcal{S}_{n}$ as \textit{Stokes kernels}.

We will verify explicitly that the constructions of the connection and the Stokes transformations are consistent with the BPZ limit in the sense that (\ref{DnSnCn}) reduces to (\ref{DnSnCnBPZ}) in this limit and that the following diagram commutes:
\begin{center}
\begin{tikzcd}[row sep=0.5cm, column sep = 2.5cm]
\begin{tabular}{c} 4-point Virasoro \\ conformal blocks \end{tabular} \ar[r,"\text{confluent limit}"] \ar[d,"\text{BPZ limit}"] & \begin{tabular}{c} confluent \\ conformal blocks \end{tabular} \ar[d,"\text{BPZ limit}"] \\ \begin{tabular}{c} hypergeometric functions \end{tabular} \ar[r,"\text{confluent limit}"] & \begin{tabular}{c} confluent \\ hypergeometric functions \end{tabular}
\end{tikzcd}
\end{center}
In symbols, this diagram takes the following form:
\begin{center}
\begin{tikzcd}[row sep=0.5cm, column sep = 2.5cm]
\mathcal{F}^0(z), \mathcal{F}^1(z), \mathcal{F}^\infty(z) \ar[r,"\text{confluent limit}"] \ar[d,"\text{BPZ limit}"] & \text{$\mathcal{B}(t)$, $\mathcal{D}_n(t)$} \ar[d,"\text{BPZ limit}"] \\ 
\boldsymbol{F}^0(z), \boldsymbol{F}^1(z), \boldsymbol{F}^\infty(z) \ar[r,"\text{confluent limit}"] & \text{$\boldsymbol{B}(t)$, $\boldsymbol{D}_n(t)$} 
\end{tikzcd}
\end{center}

\subsection{Organization of the paper}
The hypergeometric BPZ equation and its confluent limit are studied in Section \ref{section2}. Section \ref{section3} reviews some properties of four-point Virasoro conformal blocks. In Section \ref{section4}, we recall previous results on confluent conformal blocks of the first kind. The statements of our main results are gathered in Section \ref{section5}, with proofs postponed to Section \ref{section6}. The BPZ limit is studied in Section \ref{section7} and conclusions are drawn in Section \ref{section8}. Finally, the appendix collects some results on two special functions which play a prominent role in the paper.

\section{Confluence of the hypergeometric BPZ equation}\label{section2}
We consider the hypergeometric BPZ equation and its confluent limit. The parameter $b$ will be used to parametrize the central charge $c$ of the conformal field theory according to
\begin{align}\label{cQdef}
c=1+6 Q^2, \quad \text{where} \quad Q=b+b^{-1}.
\end{align}
In addition to $b$, the BPZ equation also depends on four conformal dimensions $\Delta(\theta_0), \Delta(\theta_1), \Delta(\theta_\infty)$, and $\Delta(\theta_{\text{degen}})$, which we parametrize with the help of three parameters $\theta_0$, $\theta_1$, $\theta_\infty$ according to 
$$\Delta(x)=\frac{Q^2}{4}+x^2, \qquad \theta_{\text{degen}} = \frac{iQ}{2}+\frac{ib}{2}.$$
To begin with, we allow $b, \theta_0, \theta_1, \theta_\infty$ to be any nonzero complex numbers such that $\theta_1 \neq \theta_\infty$; later, in Section \ref{section2.3}, we will assume that $b > 0$ for simplicity.

\subsection{Hypergeometric BPZ equation}

The BPZ equation is given by
\begin{equation}\label{bpz}
\left(-\frac{1}{b^2}\partial^2_z+\frac{2z-1}{z(z-1)}\partial_z-\frac{\Delta(\theta_0)}{z^2}-\frac{\Delta(\theta_1)}{(z-1)^2}+\frac{\Delta(\theta_{\text{degen}})+\Delta(\theta_0)+\Delta(\theta_1)-\Delta(\theta_\infty)}{z(z-1)}\right)F(z)=0.
\end{equation}
It has three regular singular points at $z=0,1,\infty$. Defining the functions $F_\pm^{p}$, $p = 0, 1, \infty$, by  
\begin{align*}
F^{0}_\pm\lb z\rb= &\; z^{\frac{bQ}{2}\mp ib \theta_0}(1-z)^{\frac{bQ}{2}+ib\theta_1} {}_2F_1\left[\begin{matrix}\frac{1}{2}+ib(\mp \theta_0+\theta_1+\theta_{\infty}),\frac{1}{2}+ib(\mp \theta_0+\theta_1-\theta_{\infty})  \\ 1\mp 2ib \theta_0
\end{matrix};z\right],
	\\
F^{1}_\pm\lb z\rb=&\; \lb 1-\frac1z \rb^{\frac{bQ}{2}\pm i b \theta_1}z^{\frac12+b^2-ib\theta_\infty}{}_2F_1\left[\begin{matrix}\frac{1}{2}-ib(\theta_0\mp \theta_1-\theta_{\infty}),\frac{1}{2}+ib(\theta_0\pm \theta_1+\theta_{\infty})  \\ 1\pm 2ib \theta_1
\end{matrix};1-\frac{1}{z}\right],
	\\
F^{\infty}_\pm\lb z\rb=&\; \left(\frac{1}{z} \right)^{-\frac12-b^2\mp i b \theta_\infty}\left(1-\frac{1}{z}\right)^{\frac{bQ}{2}+ib\theta_1} {}_2F_1\left[\begin{matrix}\frac{1}{2}+ib(\theta_0+\theta_1\mp \theta_{\infty}),\frac{1}{2}- ib(\theta_0 - \theta_1\pm \theta_{\infty})  \\ 1\mp 2ib \theta_\infty
\end{matrix};\frac{1}{z}\right],
\end{align*}
straightforward computations show that 
\begin{align}\label{boldFpdef}
\boldsymbol{F}^{p}_\pm\lb z\rb=\begin{pmatrix} F^{p}_+\lb z\rb \\ F^{p}_-\lb z\rb\end{pmatrix}
\end{align}
is a basis of solutions of (\ref{bpz}) that diagonalizes the monodromy operator at $z=p$ for each $p = 0, 1, \infty$. In fact, it is easy to see that the monodromy relations (\ref{monodromyintro}) hold with the diagonal monodromy matrices
$$M_0 = \begin{pmatrix} e^{2\pi b\theta_0} & 0 \\ 0 & e^{-2\pi b\theta_0} \end{pmatrix}, \qquad M_1 = \begin{pmatrix} e^{-2\pi b\theta_1} & 0 \\ 0 & e^{2\pi b\theta_1} \end{pmatrix}, \qquad M_\infty = \begin{pmatrix} e^{2\pi b\theta_\infty} & 0 \\ 0 & e^{-2\pi b\theta_\infty} \end{pmatrix}.$$
The bases $\boldsymbol{F}^{p}\lb z\rb$ are related by the connection matrices $C_{pq}$ defined by
\begin{equation}\label{connections}
\boldsymbol{F}^{p}\lb z\rb= C_{pq} \boldsymbol{F}^{q}\lb z\rb, \qquad p,q \in \{0,1,\infty\}, ~ 0 < \arg{z}< \pi.
\end{equation}
A computation using the connection formulae for the Gauss hypergeometric function shows that
\begin{subequations}\label{connectiondef}
\begin{align}
& \label{connectioni0} \large C_{\infty0}=e^{\frac{i \pi b^2}{2}}\begin{pmatrix}\frac{e^{- \pi b(\theta_0+\theta_\infty)}\Gamma \left(2ib \theta_0 \right)\Gamma \left(1-2ib\theta_\infty \right)}{\Gamma \left(\frac{1}{2}+ib(\theta_0+\theta_1-\theta_\infty \right)\Gamma \left(\frac{1}{2}-ib(-\theta_0+\theta_1+\theta_\infty) \right)} & \frac{e^{\pi b(\theta_0-\theta_\infty)} \Gamma \left(-2ib \theta_0 \right)\Gamma \left(1-2ib \theta_\infty \right)}{\Gamma \left(\frac{1}{2}-ib(\theta_0-\theta_1+\theta_\infty \right)\Gamma \left(\frac{1}{2}-ib(\theta_0+\theta_1+\theta_\infty) \right)} \\ \frac{e^{\pi b(-\theta_0+\theta_\infty)}\Gamma \left(2ib \theta_0 \right)\Gamma \left(1+2ib \theta_\infty \right)}{\Gamma \left(\frac{1}{2}+ib(\theta_0-\theta_1+\theta_\infty \right)\Gamma \left(\frac{1}{2}+ib(\theta_0+\theta_1+\theta_\infty) \right)} & \frac{e^{\pi b(\theta_0+\theta_\infty)}\Gamma \left(-2ib\theta_0 \right)\Gamma \left(1+2ib \theta_\infty \right)}{\Gamma \left(\frac{1}{2}-ib(\theta_0+\theta_1-\theta_\infty \right)\Gamma \left(\frac{1}{2}+ib(-\theta_0+\theta_1+\theta_\infty) \right)} \end{pmatrix}, \\
& \large C_{10}=e^{\frac{i\pi b Q}2} \begin{pmatrix}  \frac{e^{-\pi b \theta_1}\Gamma\left(1+2ib \theta_1 \right)\Gamma\left(2 i b\theta_0 \right)}{\Gamma\left(\frac{1}{2}+ib(\theta_0+\theta_1-\theta_\infty) \right)\Gamma\left(\frac{1}{2}+ib(\theta_0+\theta_1+\theta_\infty) \right)} &  \frac{e^{-\pi b \theta_1}\Gamma\left(1+2ib \theta_1 \right)\Gamma\left(-2 ib\theta_0 \right)}{\Gamma\left(\frac{1}{2}-ib(\theta_0-\theta_1+\theta_\infty) \right)\Gamma\left(\frac{1}{2}+ib(-\theta_0+\theta_1+\theta_\infty) \right)} \\ \frac{e^{\pi b \theta_1}\Gamma\left(1-2ib\theta_1 \right)\Gamma\left(2ib \theta_0 \right)}{\Gamma\left(\frac{1}{2}+ib(\theta_0-\theta_1+\theta_\infty \right)\Gamma\left(\frac{1}{2}-ib(-\theta_0+\theta_1+\theta_\infty) \right)} & \frac{e^{\pi b \theta_1}\Gamma\left(1-2ib\theta_1 \right)\Gamma\left(-2ib \theta_0 \right)}{\Gamma\left(\frac{1}{2}-ib(\theta_0+\theta_1-\theta_\infty) \right)\Gamma\left(\frac{1}{2}-ib(\theta_1+\theta_0+\theta_\infty) \right)}  \end{pmatrix}, \label{connection10}  \\
& \large C_{\infty 1}=\begin{pmatrix}\frac{\Gamma \left(-2ib \theta_1 \right)\Gamma \left(1-2ib \theta_\infty \right)}{\Gamma \left(\frac{1}{2}-ib(-\theta_0+\theta_1+\theta_\infty \right)\Gamma \left(\frac{1}{2}-ib(\theta_0+\theta_1+\theta_\infty) \right)} & \frac{\Gamma \left(2 ib \theta_1 \right)\Gamma \left(1-2ib \theta_\infty \right)}{\Gamma \left(\frac{1}{2}+ib(\theta_0+\theta_1-\theta_\infty \right)\Gamma \left(\frac{1}{2}-ib(\theta_0-\theta_1+\theta_\infty) \right)} \\ \frac{\Gamma \left(-2ib\theta_1 \right)\Gamma \left(1+2ib \theta_\infty \right)}{\Gamma \left(\frac{1}{2}-ib(\theta_0+\theta_1-\theta_\infty \right)\Gamma \left(\frac{1}{2}+ib(\theta_0-\theta_1+\theta_\infty) \right)} & \frac{\Gamma \left(2ib \theta_1 \right)\Gamma \left(1+2ib \theta_\infty \right)}{\Gamma \left(\frac{1}{2}-ib(\theta_0-\theta_1-\theta_\infty \right)\Gamma \left(\frac{1}{2}+ib(\theta_0+\theta_1+\theta_\infty) \right)} \end{pmatrix}\label{connectioni1}.
\end{align} \end{subequations}

\subsection{Confluent BPZ equation}
Let us write 
\beq \label{confluence} \theta_1=\frac{\Lambda+\theta_*}{2}, \quad \theta_\infty=\frac{\Lambda-\theta_*}{2}, \quad z=\frac{t}{ib \Lambda}, \eeq
where $\Lambda$ and $\theta_*$ are new parameters and $t$ is a new complex variable. We are interested in the confluent limit $\Lambda \to \infty$ in which the two nonzero regular singular points of the BPZ equation (\ref{bpz}) merge at infinity. In this limit, equation (\ref{bpz}) reduces to the confluent BPZ equation given by
\begin{equation}\label{confluentbpz}
\left(\partial_t^2-\frac{b^2}{t} \partial_t-\frac{1}{4}+\frac{b^2 \Delta(\theta_0)}{t^2}-\frac{ib\theta_*}{t} \right)G(t)=0. 
\end{equation}
This equation has a regular singular point at $t=0$ and an irregular singular point of rank one at $t=\infty$.

\subsubsection{Degenerate confluent conformal blocks of the first kind}
A basis of solutions that diagonalizes the monodromy operator at $t=0$ is given by
\begin{equation}\label{whittaker1}
\boldsymbol{B}(t)=t^{\frac{b^2}{2}}\begin{pmatrix}M_{-ib\theta_*,-ib\theta_0}(t) \\ M_{-ib\theta_*,ib\theta_0}(t) \end{pmatrix},
\end{equation}
where $M_{k,\mu}(t)$ is the Whittaker function of the first kind; it is defined in terms of the confluent hypergeometric function of the first kind ${}_1F_1$ by
\beq \label{whitakkerM} M_{k,\mu}(t)=e^{-\frac{t}{2}}t^{\frac{1}{2}+\mu} {}_1F_1\left( \begin{matrix} \frac{1}{2}-k +\mu \\ 1+2\mu\end{matrix};t\right).
\eeq 

\subsubsection{Degenerate confluent conformal blocks of the second kind}
The procedure described in the introduction (see (\ref{Dasympintro})) leads to the following formal asymptotic series which forms a formal basis of solutions of (\ref{confluentbpz}) diagonalizing the monodromy matrix at infinity:
\begin{equation}\label{asymptotic}
\boldsymbol{D}_{\text{asymp}}(t) = t^{\frac{b^2}{2}}\begin{pmatrix}e^{-\frac{t}{2}}t^{-ib\theta_*} ~ _2F_0\left(\frac{1}{2}-ib(\theta_0-\theta_*), \frac{1}{2}+ib(\theta_0+\theta_*);-\frac{1}{t}\right) \\
e^{\frac{t}{2}}t^{ib\theta_*} ~ _2F_0\left(\frac{1}{2}-ib(\theta_0+\theta_*),\frac{1}{2}+ib(\theta_0-\theta_*); \frac1t\right)
 \end{pmatrix}.
\end{equation}
Here $_2F_0$ denotes the formal power series 
$${}_2F_0(\alpha, \beta; z) = \sum_{n=0}^\infty (\alpha)_n(\beta)_n \frac{z^n}{n!},$$
where the Pochhammer symbol $(q)_n$ is defined by $(q)_n = q(q+1) \cdots (q+n-1)$. 

Let $\Omega_n$ be the Stokes sectors defined in (\ref{stokessector}). The basis of solutions $\boldsymbol{D}_n(t)$ of \eqref{confluentbpz} which asymptotes to $\boldsymbol{D}_{\text{asymp}}(t)$ as $t \to \infty$ in the Stokes sector $\Omega_n$ can be expressed in terms of the Whittaker function of the second kind $W_{k,\mu}(t)$ defined by
\beq \label{whittakerw}
W_{k,\mu}(t)=e^{-\frac{t}{2}}t^{\mu+\frac{1}{2}} U\left(\mu-k+\frac{1}{2},1+2\mu,t \right),
\eeq
where $U(a,b,z)$ is the confluent hypergeometric function of the second kind:\
\beq
U(a,b,z)=\frac{\Gamma(1-b)}{\Gamma(a+1-b)}{}_1F_1\left( \begin{matrix} a \\ b \end{matrix} ; z\right)+\frac{\Gamma(b-1)}{\Gamma(a)} z^{1-b}{}_1F_1\left( \begin{matrix} a+1-b \\ 2-b\end{matrix};z\right)
\eeq
More precisely, using the asymptotic expansion
\begin{equation}
W_{k,\mu}(t) \sim e^{-\frac{t}{2}} t^k {}_2F_0\left(\frac{1}{2}+\mu-k, \frac{1}{2}-\mu -k,-\frac1t \right), \qquad | \arg t|<\frac{3\pi}{2},
\end{equation}
the bases $\boldsymbol{D}_n(t)$ are found to be
\begin{equation}\label{solutionsinf}
\boldsymbol{D}_{n}(t)=t^{\frac{b^2}{2}}\begin{pmatrix}e^{2\pi b \theta_* \lfloor \frac{n-1}2 \rfloor} W_{-ib\theta_*,-ib\theta_0}\lb e^{-2i \pi \lfloor \frac{n-1}2 \rfloor}t\rb \\ e^{-2\pi b \theta_*(\lfloor \frac{n}2 \rfloor-\frac{1}{2})}  W_{ib\theta_*,-ib\theta_0}\lb e^{-2i \pi(\lfloor \frac{n}2 \rfloor-\frac{1}{2})}t\rb\end{pmatrix}, \qquad t \in \Omega_{n}, ~ n \in \mathbb{Z},
\end{equation}
where $\lfloor x \rfloor$ denotes the greatest integer less than or equal to $x$.
It follows from \eqref{solutionsinf} that the bases $\boldsymbol{D}_{n}(t)$ satisfy the periodicity relation
\beq \label{cyclic}
\boldsymbol{D}_{n+2}\left(te^{2i \pi}\right)=e^{i \pi b^2} e^{2\pi b \theta_* \sigma_3}\boldsymbol{D}_n(t), \qquad n \in \mathbb{Z},
\eeq
where $\sigma_3 = \diag(1, -1)$ denotes the third Pauli matrix. 

\subsubsection{Stokes matrices}
The bases $\boldsymbol{D}_{n}(t)$ are connected by the Stokes matrices $S_n$ which are defined by
\beq \label{stokesrelation}
\boldsymbol{D}_{n+1}(t)=S_n \boldsymbol{D}_n(t), \qquad n \in \mathbb{Z}.
\eeq
Because of the periodicity relation \eqref{cyclic}, there are only two independent Stokes matrices $S_1$ and $S_2$. In fact, starting from the relation $\boldsymbol{D}_{n+3}(te^{2i \pi})=S_{n+2} \boldsymbol{D}_{n+2}(te^{2i \pi})$ and using \eqref{cyclic}, we obtain
\beq \label{otherstokes}
S_{n+2}=e^{2\pi b \theta_* \sigma_3} S_n e^{-2\pi b \theta_* \sigma_3}.
\eeq
Using the following analytic continuation formula for the Whittaker function of the second kind:
\beq \begin{split}
(-1)^m W_{k,\mu}(z e^{2i \pi m})=&-\frac{e^{2i \pi k}\sin(2\pi m \mu) + \sin(\mu \pi(2m-2))}{\operatorname{sin}{2\pi \mu}}W_{k,\mu}(z) \\
&-\frac{2i \pi\sin(2\pi m \mu) e^{i k \pi}}{\sin(2\pi \mu)\Gamma\left(\frac{1}{2}+\mu-k \right)\Gamma \left(\frac{1}{2}-\mu-k \right)}W_{-k,\mu}(ze^{i \pi}),
\end{split} \eeq
we infer that the Stokes matrices $S_1$ and $S_2$ are given by
\begin{equation}\label{stokes}
S_1= \begin{pmatrix}1 & 0 \\ -\frac{2i \pi}{\Gamma(\frac{1}{2}-ib(\theta_0+\theta_*))\Gamma(\frac{1}{2}+ib(\theta_0-\theta_*))} & 1 \end{pmatrix}, \qquad S_2=\begin{pmatrix} 1 &-\frac{2i \pi e^{2\pi b \theta_*}}{\Gamma(\frac{1}{2}-ib(\theta_0-\theta_*)) \Gamma(\frac{1}{2}+ib(\theta_0+\theta_*))}  \\ 0 & 1 \end{pmatrix}.\end{equation}
%
%

\subsubsection{Connection matrices}
The bases $\boldsymbol{B}(t)$ and $\boldsymbol{D}_n(t)$ are related by the connection matrices $C_n$ which are defined by
\begin{equation}\label{conni}
\boldsymbol{D}_n(t)=C_n \boldsymbol{B}(t), \qquad t \in \Omega_n, ~ n=1,2, \dots.
\end{equation}
Here and in what follows we have restricted ourselves to positive values of $n \in \mathbb{Z}$ for simplicity.

\begin{proposition}
The connection matrices $C_n$ are given by
\begin{equation}\label{knm}
C_n = \begin{pmatrix}
K_{++}^{(n)}(\theta_*,\theta_0) & K_{++}^{(n)}(\theta_*,-\theta_0) \\ K_{-+}^{(n)}(\theta_*,\theta_0) & K_{-+}^{(n)}(\theta_*,-\theta_0)
\end{pmatrix}, \qquad n =1,2,\dots,
\end{equation}
where
\begin{align}\nonumber
K_{++}^{(n)}(\theta_*,\theta_0)=&-e^{2i\pi \lb \lfloor \frac{n}2 \rfloor -\frac12 \rb \lb -\frac12+ib(\theta_0-\theta_*) \rb}\operatorname{csc}{\lb \pi b^2\rb}  
	\\  \nonumber
& \times \lb e^{(-1)^{n}ibQ \pi} \operatorname{cosh}{\lb\pi b(\theta_0-\theta_*)\rb}+ \operatorname{cosh}{\lb\pi b(\theta_0-\theta_*-ib) \rb}\rb \frac{\Gamma(2ib\theta_0)}{\Gamma\lb \frac12+ib(\theta_0+\theta_*) \rb}, 
	\\ \label{kelement} 
K_{-+}^{(n)}(\theta_*,\theta_0)=&\; e^{-2i\pi \lb \lfloor \frac{n}2 \rfloor-\frac12 \rb \lb \frac12-ib(\theta_0+\theta_*) \rb}\frac{\Gamma(2ib\theta_0)}{\Gamma\lb \frac12+ib(\theta_0-\theta_*) \rb}.
\end{align}
\end{proposition}
\begin{proof}
It is easy to verify that the first connection matrix is given by
\begin{align}\label{C1expression}
C_1=\left(
\begin{array}{cc}
 \frac{\Gamma (2 i b \theta_0)}{\Gamma \left(\frac{1}{2}+ib (\theta_0+\theta_*)\right)} & \frac{\Gamma (-2 i b \theta_0)}{\Gamma \left(\frac{1}{2}-i b (\theta_0-\theta_*)\right)} \\
 \frac{i e^{\pi  b (\theta_0+\theta_*)} \Gamma (2 i b \theta_0)}{\Gamma \left(\frac{1}{2}+ib(\theta_0-\theta_*)\right)} & \frac{i e^{\pi  b (\theta_*-\theta_0)} \Gamma (-2 i b \theta_0)}{\Gamma \left(\frac{1}{2}-i b (\theta_0+\theta_*)\right)} \\
\end{array}
\right).
\end{align}
Since
\begin{equation}\label{cj0}
C_n=\left( \displaystyle \prod_{k=1}^{n-1} S_{n-k} \right) C_1
 = S_{n-1}S_{n-2} \cdots S_1C_1, \qquad n=2,3,\dots,
\end{equation}
it follows from (\ref{stokes}) and (\ref{C1expression}) that \eqref{knm} holds for $n=1$ and $n=2$. Moreover, by \eqref{conni}, the connection matrices $C_n$ satisfy the periodicity relation
\beq \label{periodcn0}
C_{n+2}=e^{i\pi b^2} e^{2\pi b \theta_* \sigma_3} C_n e^{-i \pi b Q}e^{-2\pi b \theta_0 \sigma_3}.
\eeq
Proceeding by induction, it is therefore sufficient to prove that
\beq \begin{split}
& C_{2j+1} = e^{i\pi j b^2} e^{2\pi b j \theta_* \sigma_3} C_1 e^{-i \pi j b Q}e^{-2\pi b j \theta_0 \sigma_3}, \qquad j=1,2,\dots, \\
& C_{2j+2} = e^{i\pi j b^2} e^{2\pi b j \theta_* \sigma_3} C_2 e^{-i \pi j b Q}e^{-2\pi b j \theta_0 \sigma_3}, \qquad j=1,2,\dots,
\end{split} \eeq
and direct computations show that the matrices defined in \eqref{knm} obey these relations.
\end{proof}

\subsection{Confluence of the solutions}\label{section2.3}
In this subsection, we explain how to obtain the solutions of the confluent BPZ equation \eqref{confluentbpz} from the solutions of the hypergeometric BPZ equation \eqref{bpz} by taking an appropriate confluent limit. Let us first set $\theta_1 = \frac{\Lambda+\theta_*}{2}$ and $\theta_\infty=\frac{\Lambda-\theta_*}2$. We introduce renormalized versions $\boldsymbol{\tilde{F}}^{p}(z,\Lambda)$ of the bases $\boldsymbol{F}^{p}(z)$ defined in (\ref{boldFpdef}) as follows:
\begin{equation} \label{renorm} \left\{ \begin{split}
& \boldsymbol{\tilde{F}}^{0}\left(z,\Lambda\right)=N_0(\Lambda) \boldsymbol{F}^{0}\lb z\rb, \\
& \boldsymbol{\tilde{F}}^{1}\left(z,\Lambda\right)=N_\infty(\Lambda) \boldsymbol{F}^{1}\lb z\rb, \\
& \boldsymbol{\tilde{F}}^{\infty}\left(z,\Lambda \right)=N_\infty(\Lambda) \boldsymbol{F}^{\infty}\lb z\rb,
\end{split}\right.\quad \left\{\begin{split}
&N_0(\Lambda) =\begin{pmatrix} \left(i b\Lambda \right)^{\frac{bQ}{2}-ib\theta_0} & 0 \\ 0 & \left(i b\Lambda \right)^{\frac{bQ}{2}+ib\theta_0}  \end{pmatrix}, \\
& N_\infty(\Lambda)=\begin{pmatrix}Q_+(\Lambda) & 0 \\ 0 &Q_-(\Lambda) \end{pmatrix},\end{split}\right. \end{equation}
where 
\beq Q_\pm(\Lambda) = e^{-\frac{i \pi b Q}{2}}e^{\pm \frac{\pi b}{2}(\theta_*+\Lambda)}(ib\Lambda)^{\frac{b}{2}(b\mp2i \theta_*)}.\eeq
After performing the change of variables $z=\frac{t}{ib\Lambda}$, equation \eqref{bpz} has three regular singular points at $t= 0$, $t = ib\Lambda$, and $t = \infty$. The next proposition shows that the solution basis $\boldsymbol{B}(t)$ of the confluent BPZ equation defined in \eqref{whittaker1} is the confluent limit of the renormalized basis of solutions $\boldsymbol{\tilde{F}}^{0}(\frac{t}{ib\Lambda}, \Lambda)$ of the hypergeometric BPZ equation. 

\begin{proposition}
We have
\beq \label{limit0}
\lim_{|\Lambda| \to \infty} \boldsymbol{\tilde{F}}^{0}\left(\frac{t}{ib\Lambda},\Lambda\right) = \boldsymbol{B}(t).
\eeq
\end{proposition}

\begin{proof}
The renormalized function $\boldsymbol{\tilde{F}}^{0}\left(\frac{t}{ib\Lambda},\Lambda\right)$ explicitly reads
$$\boldsymbol{\tilde{F}}^{0}\left(\frac{t}{ib\Lambda},\Lambda\right)=\begin{pmatrix} \tilde{F}_+^{0}\left(\frac{t}{ib\Lambda},\Lambda\right) \\ \tilde{F}^{0}_-\left(\frac{t}{ib\Lambda},\Lambda\right)\end{pmatrix},$$
where
 $$\large \tilde{F}^{0}_\pm\left(\frac{t}{ib\Lambda},\Lambda \right)= \left(1-\frac{t}{ib \Lambda}  \right)^{\frac{bQ}{2}+\frac{ib}{2}(\theta_*+\Lambda)}t^{\frac{bQ}{2}\mp ib \theta_0}{}_2F_1\left[\begin{matrix}\frac{1}{2}+ib(\theta_*\mp \theta_0),\frac{1}{2}+ib(\Lambda \mp \theta_0)  \\ 1\mp 2ib \theta_0
\end{matrix};\frac{t}{ib\Lambda}\right].$$
Since
$$\lim_{ \beta \to \infty} {}_2F_1\lb \begin{matrix} \alpha, \beta \\ \gamma \end{matrix}; \frac{t}{\beta} \rb = {}_1F_1 \lb \begin{matrix}\alpha \\ \gamma \end{matrix};t \rb,$$
the limit \eqref{limit0} follows.
\end{proof}

We next explain how to recover the solution bases $\boldsymbol{D}_n$ of the confluent BPZ equation adapted to the irregular singular point at $\infty$. Actually, the $\boldsymbol{D}_n$ can be obtained as the confluent limit in two different ways: either starting from $\boldsymbol{\tilde{F}}^{\infty}$ or from $\boldsymbol{\tilde{F}}^{1}$. This is consistent with the fact that the two bases $\boldsymbol{\tilde{F}}^{\infty}$ and $\boldsymbol{\tilde{F}}^{1}$ diagonalize the monodromy matrices at the two singular points that merge at $\infty$ in the confluent limit. 

For the rest of this section we assume that $b,\Lambda>0$ for simplicity. Defining the renormalized connection matrices $\tilde{C}_{pq}$ by the relation $\large \boldsymbol{\tilde{F}}^{p} = \tilde{C}_{pq} \large \boldsymbol{\tilde{F}}^{q}$, we have
\beq \label{renormconnections} \tilde{C}_{1 0}(\Lambda) = N_\infty(\Lambda) C_{1 0} N_0^{-1}(\Lambda), \quad \tilde{C}_{\infty 1}(\Lambda)=N_\infty(\Lambda) C_{\infty 1}N_\infty^{-1}(\Lambda), \quad \tilde{C}_{\infty 0}(\Lambda)=N_\infty(\Lambda) C_{\infty 0} N_0^{-1}(\Lambda), \eeq
where $C_{\infty 0}, C_{1 0}$, and $C_{\infty 1}$ are the connection matrices defined in \eqref{connections}. The crucial point is that \eqref{connections} holds for $0 < \arg{z} < \pi$. To access all the Stokes sectors $\Omega_n$, it is therefore convenient to write  
\beq \label{zlimit}
z = \frac{t e^{-2i \pi(j-1)}}{ib\epsilon\Lambda},
\eeq
where $j\geq1$ is an integer and $\epsilon=\pm1$. It is straightforward to show that $0 < \arg{z} < \pi$ if and only if
\beq \label{argt}
\frac{\epsilon \pi}2 + 2\pi(j-1) < \arg{t} < \frac{\pi}2(2+\epsilon) +2\pi(j-1).
\eeq
Decomposing the Stokes sector $\Omega_n$ into the two halves $\Omega_n^-$ and $\Omega_n^+$ defined by
\beq \label{halfstokes}\begin{split}
& \Omega_n^- = \left\{ -\frac{3\pi}2+\pi(n-1) < \arg{t} < -\frac{\pi}2+\pi(n-1)\right\}, \\
& \Omega_n^+ = \left\{-\frac{\pi}2+\pi(n-1) < \arg{t} < \frac{\pi}2+\pi(n-1)\right\},
\end{split}\eeq
we infer from \eqref{argt} that $0 < \arg{z} < \pi$ if and only if
\beq \label{case1}
\begin{cases} t \in \Omega_{2j+1}^-=\Omega_{2j}^+, & \epsilon=+1, \\ t \in \Omega_{2j}^-=\Omega_{2j-1}^+, & \epsilon=-1.\end{cases} \eeq
Our next proposition utilizes \eqref{zlimit} and (\ref{case1}) to construct the solution basis $\boldsymbol{D}_n(t)$ everywhere in $\Omega_n$ from the renormalized solution bases $\tilde{\boldsymbol{F}}^\infty$ and $\tilde{\boldsymbol{F}}^1$.

\begin{proposition}\label{prop2p4}
The following limits hold for any integer $j\geq1$:
\begin{subequations}\label{confluentlimits}\begin{align}
& \label{confluencea} \lb e^{i\pi b^2} e^{2\pi b \theta_* \sigma_3}\rb^{j-1} \lim\limits_{\substack{\Lambda \to + \infty}}\tilde{\boldsymbol{ F}}^{\infty}\lb \frac{t e^{-2i\pi(j-1)}}{ib\epsilon\Lambda},\epsilon \Lambda\rb = \begin{cases} 
\boldsymbol{D}_{2j+1}(t), & t \in \Omega_{2j+1}^-,~ \epsilon=+1, \\ 
\boldsymbol{D}_{2j}(t), & t \in \Omega_{2j}^-,~ \epsilon=-1,
\end{cases} \\
& \label{confluenceb} \lb e^{i\pi b^2} e^{2\pi b \theta_* \sigma_3}\rb^{j-1} \lim\limits_{\substack{\Lambda \to + \infty}}\tilde{\boldsymbol{F}}^{1}\lb \frac{t e^{-2i\pi(j-1)}}{ib\epsilon\Lambda},\epsilon \Lambda\rb 
= \begin{cases} \boldsymbol{D}_{2j}(t), & t \in \Omega_{2j}^+,~ \epsilon=+1, \\ 
\boldsymbol{D}_{2j-1}(t), & t\in \Omega_{2j-1}^+,~ \epsilon=-1, 
\end{cases}
\end{align}\end{subequations} 
\end{proposition}

\begin{proof}
Let us first prove \eqref{confluencea}. We start from the connection formula
\begin{equation}\label{conninf0}
 \tilde{\boldsymbol{ F}}^{\infty}\lb \frac{t e^{-2i\pi(j-1)}}{ib\epsilon\Lambda},\epsilon \Lambda\rb  = \tilde{C}_{\infty 0}(\epsilon \Lambda)  \tilde{\boldsymbol{ F}}^{0}\lb \frac{t e^{-2i\pi(j-1)}}{ib\epsilon\Lambda}, \epsilon\Lambda\rb,
\end{equation}
where the renormalized connection matrix $\tilde{C}_{\infty 0}$ is given by \eqref{renormconnections}. The renormalized basis $\tilde{\boldsymbol{F}}^0$ satisfies
$$\tilde{\boldsymbol{ F}}^0\lb \frac{t e^{-2i\pi(j-1)}}{i\epsilon b\Lambda},\epsilon \Lambda\rb = \lb e^{-i \pi b Q}e^{-2\pi b \theta_0 \sigma_3} \rb^{j-1} \boldsymbol{\tilde{F}}^{0}\left(\frac{t}{ib \epsilon \Lambda},\epsilon\Lambda\right).$$
Thus the connection formula \eqref{conninf0} can be rewritten as
\begin{equation*} \begin{split}
\lb e^{i\pi b^2} e^{2\pi b \theta_* \sigma_3}\rb^{j-1} & \tilde{\boldsymbol{ F}}^{\infty}\lb \frac{t e^{-2i\pi(j-1)}}{ib\epsilon\Lambda},\epsilon \Lambda\rb  
	\\
& =  \lb e^{i\pi b^2} e^{2\pi b \theta_* \sigma_3}\rb^{j-1} \tilde{C}_{\infty 0}(\epsilon \Lambda) \lb e^{-i \pi b Q}e^{-2\pi b \theta_0 \sigma_3} \rb^{j-1} \boldsymbol{\tilde{F}}^{0}\left(\frac{t}{ib \epsilon \Lambda},\epsilon\Lambda\right).\end{split} \end{equation*}
Letting $\Lambda \to +\infty$ in this equation and using the limit \eqref{limit0} of $\boldsymbol{\tilde{F}}^{0}(\frac{t}{ib\epsilon \Lambda},\epsilon\Lambda)$, we find
\beq \label{firstlimit} \begin{split}
 \lb e^{i\pi b^2} e^{2\pi b \theta_* \sigma_3}\rb^{j-1} \lim\limits_{\substack{\Lambda \to +\infty}}& \tilde{\boldsymbol{ F}}^{\infty}\lb \frac{t e^{-2i\pi(j-1)}}{ib\epsilon\Lambda},\epsilon \Lambda\rb \\
= & \lb e^{i\pi b^2} e^{2\pi b \theta_* \sigma_3}\rb^{j-1} \lim\limits_{\Lambda \to +\infty} \lb \tilde{C}_{\infty 0}(\epsilon \Lambda)\rb \lb e^{-i \pi b Q}e^{-2\pi b \theta_0 \sigma_3} \rb^{j-1} \boldsymbol{B}(t).
\end{split} \eeq

It remains to compute the limit of $\tilde{C}_{\infty 0}(\epsilon \Lambda)$. This matrix can be explicitly written as
$$\tilde{C}_{\infty 0}(\epsilon \Lambda) = \begin{pmatrix}  L_{\infty 0}^{(++)}(\theta_*,\theta_0,\epsilon \Lambda) & L_{\infty 0}^{(++)}(\theta_*,-\theta_0,\epsilon \Lambda) \\ L_{\infty 0}^{(-+)}(\theta_*,\theta_0,\epsilon \Lambda) & L_{\infty 0}^{(-+)}(\theta_*,-\theta_0,\epsilon \Lambda) \end{pmatrix},$$
where
\begin{equation*} \begin{split}
& L_{\infty 0}^{(++)}(\theta_*,\theta_0,\epsilon \Lambda)= -i e^{\pi  b (\theta_*-\theta_0)} (i b \epsilon \Lambda )^{-\frac{1}{2}+i b (\theta_0-\theta_*)} \frac{\Gamma (2 i b \theta_0)}{\Gamma \left(\frac{1}{2}+ib (\theta_0+\theta_*)\right)} \frac{\Gamma (1+ib \theta_*-i b \epsilon \Lambda)}{\Gamma \left(\frac{1}{2}+ib \theta_0-i b \epsilon \Lambda\right)}, \\
& L_{\infty 0}^{(-+)}(\theta_*,\theta_0,\epsilon \Lambda)=-i e^{-\pi b(\theta_*+\theta_0)} (i b \epsilon \Lambda)^{-\frac{1}{2}+i b (\theta_0+\theta_*)} \frac{\Gamma (2 i b \theta_0)}{\Gamma \left(\frac12+ib (\theta_0-\theta_*)\right)} \frac{\Gamma (1-i b \theta_*+ib \epsilon  \Lambda)}{\Gamma \left(\frac12+ib \theta_0+ib \epsilon  \Lambda \right)}.\end{split} \end{equation*}
Using the asymptotic formula
\beq \label{asympgamma}
\Gamma(z+a) \sim \sqrt{2\pi}~z^{z+a-\frac12}~e^{-z}, \qquad z \to \infty, ~ z + a \in \mathbb{C} \setminus \mathbb{R}_{\leq 0}, ~  |a| < |z|, 
\eeq
we obtain
\begin{align}\label{Li0}
& \lim\limits_{\Lambda \to + \infty}L_{\infty 0}^{(++)}(\theta_*,\theta_0,\epsilon \Lambda)=-i e^{\pi  b (\theta_*-\theta_0)} (i b \epsilon \Lambda)^{-\frac{1}{2}+i b (\theta_0-\theta_*)}(-i b \epsilon \Lambda)^{\frac{1}{2}-i b (\theta_0-\theta_*)}\frac{\Gamma (2 i b \theta_0)}{\Gamma \left(\frac{1}{2}+ib(\theta_0+\theta_*)\right)}, \\
& \lim\limits_{\Lambda \to + \infty}L_{\infty 0}^{(-+)}(\theta_*,\theta_0,\epsilon \Lambda)=-ie^{-\pi b(\theta_*+\theta_0)} \frac{\Gamma(2ib\theta_0)}{\Gamma(\frac12+ib(\theta_0-\theta_*))}.
\end{align}
Since $(-ib \Lambda)^x = e^{-i \pi x} (ib \Lambda)^x$, we infer that
\beq 
(i b \epsilon \Lambda)^{-\frac{1}{2}+i b (\theta_0-\theta_*)}(-i b \epsilon \Lambda)^{\frac{1}{2}-i b (\theta_0-\theta_*)} = \begin{cases} -ie^{\pi b(\theta_*-\theta_0)}, & \epsilon=+1, \\ ie^{-\pi b(\theta_*-\theta_0)}, & \epsilon=-1. \end{cases}
\eeq
Gathering the previous computations, it is straightforward to obtain the two limits
\beq \label{c32}
\lim_{\Lambda \to + \infty} \tilde{C}_{\infty 0}(\epsilon \Lambda) = \begin{cases} 
C_{3}, & \epsilon=+1, \\ C_{2}, & \epsilon=-1, 
\end{cases}
\eeq
where the connection matrices $C_n$ of the confluent BPZ equation are defined in \eqref{knm}. Moreover, the periodicity relation \eqref{periodcn0} implies that
\begin{align*}
&\lb e^{i\pi b^2} e^{2\pi b \theta_* \sigma_3}\rb^{j-1} C_3 \lb e^{-i \pi b Q}e^{-2\pi b \theta_0 \sigma_3} \rb^{j-1} = C_{2j+1},
	\\
& \lb e^{i\pi b^2} e^{2\pi b \theta_* \sigma_3}\rb^{j-1} C_2 \lb e^{-i \pi b Q}e^{-2\pi b \theta_0 \sigma_3} \rb^{j-1} = C_{2j},
\end{align*}
for any integer $j \geq 1$. Hence we have shown that
\beq \label{limitci0}
\lb e^{i\pi b^2} e^{2\pi b \theta_* \sigma_3}\rb^{j-1} \lim\limits_{\Lambda \to + \infty} \lb \tilde{C}_{\infty 0}(\epsilon \Lambda)\rb \lb e^{-i \pi b Q}e^{-2\pi b \theta_0 \sigma_3} \rb^{j-1} = \begin{cases} C_{2j+1}, & \epsilon=+1, \\  C_{2j}, & \epsilon=-1.\end{cases} 
\eeq
Substituting \eqref{limitci0} into \eqref{firstlimit}, we obtain
\beq
\lb e^{i\pi b^2} e^{2\pi b \theta_* \sigma_3}\rb^{j-1} \lim\limits_{\substack{\Lambda \to +\infty}}\tilde{\boldsymbol{ F}}^{\infty}\lb \frac{t e^{-2i\pi(j-1)}}{ib\epsilon\Lambda},\epsilon \Lambda\rb =\begin{cases}  C_{2j+1} \boldsymbol{B}(t), & \epsilon=+1,~t \in \Omega_{2j+1}^-, \\ C_{2j} \boldsymbol{B}(t), & \epsilon=-1,~t \in \Omega_{2j}^-. \end{cases}
\eeq
Since $\boldsymbol{D}_n(t)=C_n \boldsymbol{B}(t)$ for $t\in \Omega_n$ by \eqref{conni}, the confluent limit \eqref{confluencea} follows.

The proof of the second limit \eqref{confluenceb} involves a similar computation. Indeed, we have
\begin{equation} \label{ccc} \begin{split}
\lb e^{i\pi b^2} e^{2\pi b \theta_* \sigma_3}\rb^{j-1} & \tilde{\boldsymbol{F}}^{1}\lb \frac{t e^{-2i\pi(j-1)}}{ib\epsilon\Lambda},\epsilon \Lambda \rb = \\
& =  \lb e^{i\pi b^2} e^{2\pi b \theta_* \sigma_3}\rb^{j-1} \tilde{C}_{1 0}(\epsilon \Lambda) \lb e^{-i \pi b Q}e^{-2\pi b \theta_0 \sigma_3} \rb^{j-1} \boldsymbol{\tilde{F}}^{0}\left(\frac{t}{ib\epsilon\Lambda},\epsilon\Lambda\right). \end{split} \end{equation}
The renormalized connection matrix takes the form
$$\tilde{C}_{1 0}(\epsilon \Lambda) = \begin{pmatrix}  L_{1 0}^{(++)}(\theta_*,\theta_0,\epsilon \Lambda) & L_{1 0}^{(++)}(\theta_*,-\theta_0,\epsilon \Lambda) \\ L_{10}^{(-+)}(\theta_*,\theta_0,\epsilon \Lambda) & L_{1 0}^{(-+)}(\theta_*,-\theta_0,\epsilon \Lambda) \end{pmatrix},$$
where
$$L_{1 0}^{(++)}(\theta_*,\theta_0,\epsilon \Lambda)=(i b \Lambda  \epsilon )^{-\frac{1}{2}+i b (\theta_0-\theta_*)}\frac{\Gamma (2 i b \theta_0)}{\Gamma \left(\frac12+ib(\theta_0+\theta_*)\right)}\frac{\Gamma (1+ib \theta_*+ib \epsilon  \Lambda)}{\Gamma \left(\frac{1}{2}+ib \theta_0+ib \epsilon  \Lambda\right)},$$
$$L_{1 0}^{(-+)}(\theta_*,\theta_0,\epsilon \Lambda)=(i b \epsilon \Lambda)^{-\frac{1}{2}+i b (\theta_0+\theta_*)} \frac{\Gamma(2 i b \theta_0)}{\Gamma \left(\frac{1}{2}+ib(\theta_0-\theta_*)\right)} \frac{\Gamma (1-i b \theta_*-i b \epsilon  \Lambda)}{\Gamma \left(\frac{1}{2}+ib \theta_0-i b \epsilon  \Lambda\right)}.$$
Using the asymptotics \eqref{asympgamma}, a direct computation shows that 
$$\lim_{\Lambda \to +\infty} \tilde{C}_{1 0}(\epsilon \Lambda) = \begin{cases} C_2, & \epsilon=+1, \\  C_1, & \epsilon=-1. \end{cases}$$
Finally, observing that
\begin{align*}
&\lb e^{i\pi b^2} e^{2\pi b \theta_* \sigma_3}\rb^{j-1} C_{2} \lb e^{-i \pi b Q}e^{-2\pi b \theta_0 \sigma_3} \rb^{j-1} = C_{2j}, \\
& \lb e^{i\pi b^2} e^{2\pi b \theta_* \sigma_3}\rb^{j-1} C_{1} \lb e^{-i \pi b Q}e^{-2\pi b \theta_0 \sigma_3} \rb^{j-1} = C_{2j-1},
\end{align*}
for $j=1,2,\dots$, equation (\ref{confluenceb}) follows. 
\end{proof}

Proposition \ref{prop2p4} can be used to determine the Stokes matrices $S_n$ of the confluent BPZ equation defined by \eqref{stokesrelation}. In fact, consider the connection formula between the renormalized solution bases $\tilde{\boldsymbol{F}}^\infty$ and $\tilde{\boldsymbol{F}}^1$:
 
\begin{equation}\label{tildeFinftytildeF1}
\boldsymbol{\tilde{F}}^{\infty}\left(\frac{t e^{-2i\pi(j-1)}}{ib\epsilon\Lambda},\epsilon \Lambda \right) = \tilde{C}_{\infty 1}(\epsilon \Lambda) \boldsymbol{\tilde{F}}^{1}\left(\frac{t e^{-2i\pi(j-1)}}{ib\epsilon\Lambda},\epsilon \Lambda \right), \qquad \epsilon=\pm1.
\end{equation}
Introducing the prefactors appearing in \eqref{confluencea} and \eqref{confluenceb}, this relation can be rewritten as
\begin{equation} \label{connectioninf1}
\begin{split}
 \lb e^{i\pi b^2} e^{2\pi b \theta_* \sigma_3}\rb^{j-1} \boldsymbol{\tilde{F}}^{\infty}\left(\frac{t e^{-2i\pi(j-1)}}{ib\epsilon\Lambda},\epsilon \Lambda \right) = & \lb e^{i\pi b^2} e^{2\pi b \theta_* \sigma_3}\rb^{j-1} \tilde{C}_{\infty 1}(\epsilon \Lambda) \lb e^{i\pi b^2} e^{2\pi b \theta_* \sigma_3}\rb^{-j+1}
	 \\
& \times \lb e^{i\pi b^2} e^{2\pi b \theta_* \sigma_3}\rb^{j-1}\boldsymbol{\tilde{F}}^{1}\left(\frac{t e^{-2i\pi(j-1)}}{ib\epsilon\Lambda},\epsilon \Lambda \right),
\end{split}
\end{equation}
where
\begin{align*}
\tilde{C}_{\infty 1}(\epsilon \Lambda) = \left(
\begin{array}{cc}
 \frac{\Gamma (1+ib (\theta_*-\epsilon\Lambda )) \Gamma (-i b (\theta_*+\epsilon\Lambda ))}{\Gamma \left(\frac12+ib (\theta_0-\epsilon\Lambda )\right) \Gamma \left(\frac{1}{2}-i b (\theta_0+\epsilon\Lambda )\right)} & \frac{e^{\pi  b (\theta_*+\epsilon\Lambda )} (i b \epsilon\Lambda )^{-2 i b \theta_*} \Gamma (1+ib(\theta_*-\epsilon\Lambda )) \Gamma (i b (\theta_*+\epsilon\Lambda ))}{\Gamma \left(\frac{1}{2}-i b (\theta_0-\theta_*)\right) \Gamma \left(\frac12+ib (\theta_0+\theta_*)\right)} \\
 \frac{e^{-\pi b (\theta_*+\epsilon\Lambda )} (i b \epsilon\Lambda )^{2 i b \theta_*} \Gamma (-i b (\theta_*+\epsilon\Lambda )) \Gamma (1+ib (\epsilon\Lambda -\theta_*))}{\Gamma \left(\frac12+ib (\theta_0-\theta_*)\right) \Gamma \left(\frac{1}{2}-i b (\theta_0+\theta_*)\right)} & \frac{\Gamma (i b (\theta_*+\epsilon\Lambda )) \Gamma (1+ib (\epsilon\Lambda -\theta_*))}{\Gamma \left(\frac{1}{2}-i b (\theta_0-\epsilon\Lambda )\right) \Gamma \left(\frac12+ib (\theta_0+\epsilon\Lambda )\right)} \\
\end{array}
\right).
\end{align*}

Proposition \ref{prop2p4} implies that the confluent limit of the relation \eqref{connectioninf1} must lead to the formulas $\boldsymbol{D}_{2j+1}(t)=S_{2j} \boldsymbol{D}_{2j}(t)$ and $\boldsymbol{D}_{2j}(t)=S_{2j-1} \boldsymbol{D}_{2j-1}(t)$ for $\epsilon=+1$ and $\epsilon=-1$, respectively. 
Let us verify this explicitly. 
As $\Lambda \to +\infty$, a direct computation utilizing \eqref{asympgamma} yields
\begin{align}\label{tildeCinfty1}
\tilde{C}_{\infty 1}(\epsilon \Lambda) \sim \left(
\begin{array}{cc}
 1 & \frac{-2i\pi \epsilon ~ e^{\pi  b \theta_* (\epsilon +1)} e^{\pi  b \Lambda  (\epsilon -1)}}{\Gamma \left(\frac{1}{2}-i b (\theta_0-\theta_*)\right) \Gamma \left(\frac12+ib(\theta_0+\theta_*)\right)} \\
 \frac{2 i \pi \epsilon ~ e^{-\pi  b \theta_* (\epsilon +1)} e^{-\pi  b \Lambda  (\epsilon +1)}}{\Gamma \left(\frac12+ib (\theta_0-\theta_*)\right) \Gamma \left(\frac{1}{2}-i b (\theta_0+\theta_*)\right)} & 1 \\
\end{array}
\right),\qquad \epsilon=\pm1.
\end{align}
In each of the two cases $\epsilon=\pm1$, one of the off-diagonal entries on the right-hand side of (\ref{tildeCinfty1}) has exponential decay, and it is straightforward to obtain
\beq
\lim_{\Lambda \to \infty} \tilde{C}_{\infty 1}(\epsilon \Lambda) = \begin{cases} S_2, & \epsilon=+1, \\ S_1, & \epsilon=-1, \end{cases}
\eeq 
where $S_1$ and $S_2$ are given by \eqref{stokes}. Moreover, thanks to the periodicity relation \eqref{otherstokes}, 
\beq \begin{split}
& \lb e^{i\pi b^2} e^{2\pi b \theta_* \sigma_3}\rb^{j-1}S_2 \lb e^{i\pi b^2} e^{2\pi b \theta_* \sigma_3}\rb^{-j+1} = S_{2j}, \\
&  \lb e^{i\pi b^2} e^{2\pi b \theta_* \sigma_3}\rb^{j-1} S_1 \lb e^{i\pi b^2} e^{2\pi b \theta_* \sigma_3}\rb^{-j+1} = S_{2j-1},
\end{split}
\eeq
for any integer $j\geq1$. Recalling Propositions \ref{prop2p4} and noting that
\beq
\Omega_{2j+1}^- = \Omega_{2j}^+ = \Omega_{2j+1} \cap \Omega_{2j}, \qquad \Omega_{2j}^- = \Omega_{2j-1}^+ = \Omega_{2j} \cap \Omega_{2j-1},
\eeq 
it follows that the two relations 
\begin{align*}
& \boldsymbol{D}_{2j+1}(t)=S_{2j} \boldsymbol{D}_{2j}(t), \qquad t \in \Omega_{2j+1} \cap \Omega_{2j},
	\\
& \boldsymbol{D}_{2j}(t)=S_{2j-1} \boldsymbol{D}_{2j-1}(t), \qquad t \in \Omega_{2j} \cap \Omega_{2j-1},
\end{align*} 
are recovered by taking the limit $\Lambda \to +\infty$ of \eqref{tildeFinftytildeF1} for $\epsilon=+1$ and $\epsilon = -1$, respectively.

\section{Four-point Virasoro conformal blocks} \label{section3}
The remainder of this article will be devoted to generalizing the picture developed in Section \ref{section2} to the case of generic four-point Virasoro conformal blocks. We start by recalling their main properties.

The (regular) $4$-point conformal block is often represented in the literature by a trivalent graph encoding the expectation value of a composition of two primary vertex operators as follows:
 \ben
 \begin{tikzpicture}[baseline,yshift=-0.3cm,scale=0.8]
 \draw [thick] (-1,0) -- (2,0);
 \draw [thick] (0,0) -- (0,1);
 \draw [thick] (1,0) -- (1,1); 
 \draw (-0.7,0) node[above] {\scriptsize $\theta_\infty$};
 \draw (0.5,0) node[above] {\scriptsize $\sigma_s$};
 \draw (1.7,0) node[above] {\scriptsize $\theta_0$};
 \draw (0,1) node[left] {\scriptsize $\theta_1$};
 \draw (1,1) node[right] {\scriptsize $\theta_t$};
 \end{tikzpicture}
  \lb z\rb \equiv \mathcal F\lb\substack{\theta_{1}\;\quad 
  \theta_{t}\\ \sigma_s \\ \theta_{\infty}\quad \theta_0};z\rb: =\left\langle \Delta(\theta_\infty)\right|V_{\Delta(\theta_\infty),\Delta(\sigma_s)}^{\Delta(\theta_1)}\lb 1\rb V_{\Delta(\sigma_s),\Delta(\theta_0)}^{\Delta(\theta_t)}\lb z \rb\left|\Delta(\theta_0)\right\rangle.
 \ebn 
 It depends on the Virasoro central charge $c$, five conformal dimensions $\Delta\lb x\rb=\frac{c-1}{24}+x^2$ attached to the edges labeling highest weight modules, and the anharmonic ratio $t$ of four points on $\Cb\Pb^1$. The vertices of the graph represent chiral vertex operators \cite{T04}.
The series representation for conformal blocks was made explicit by the discovery of the AGT relation \cite{AGT} between two-dimensional conformal field theories and four-dimensional supersymmetric gauge theories. Denoting by $\mathbb Y$ the set of Young diagrams, the 4-point conformal block is expressed as \cite{AFLT,Nekrasov}
 \beq\label{CBseries}
 \mathcal F\lb\substack{\theta_{1}\;\quad 
      \theta_{t}\\ \sigma_s \\ \theta_{\infty}\quad \theta_0};z\rb=
      z^{\Delta(\sigma_s)-\Delta(\theta_0)-\Delta(\theta_t)}\lb 1-z\rb^{2\lb\theta_t-\frac{iQ}{2}\rb
      \lb\theta_1-\frac{iQ}{2}\rb}\sum_{\lambda,\mu\in\mathbb Y} \mathcal F_{\lambda,\mu}\lb\substack{\theta_{1}\;\quad 
            \theta_{t}\\ \sigma_s \\ \theta_{\infty}\quad \theta_0}\rb z^{|\lambda|+|\mu|},
 \eeq
 where $|\lambda|$ denotes the number of boxes in the diagram $\lambda\in\mathbb Y$.
 In order to write the coefficients $\mathcal F_{\lambda,\mu}$ explicitly, let $a_{\lambda}\lb \square\rb$ and $l_{\lambda}\lb \square\rb$ denote the arm-length and leg-length of the box $\square$ in $\lambda$. 
Moreover, for  $\theta\in\mathbb C$ and $\lambda,\mu\in\mathbb Y$, introduce the Nekrasov functions $ Z_{\lambda,\mu}\lb \theta\rb$ by
 \begin{align}\nonumber
 Z_{\lambda,\mu}\lb \theta\rb = &\; \prod_{\square\in\lambda}
 \Bigl(i b^{-1}\lb a_{\lambda}\lb\square\rb+{\tfrac12}\rb-ib \lb l_{\mu}\lb \square\rb +\tfrac12\rb+\theta\Bigr)
 	\\ \label{nekrofun}
&\times \prod_{\square\in\mu} \Bigl( i b^{-1}\lb a_{\mu}\lb\square\rb+{\tfrac12}\rb-ib \lb l_{\lambda}\lb \square\rb +\tfrac12\rb-\theta\Bigr).
\end{align}
The expansion coefficients $\mathcal F_{\lambda,\mu}$ can then be expressed as  
$$\mathcal F_{\lambda,\mu}\lb\substack{\theta_{1}\;\quad 
             \theta_{t}\\ \sigma_s \\ \theta_{\infty}\quad \theta_0}\rb=\frac{\prod_{\epsilon=\pm}
 Z_{\emptyset,\lambda}\lb\epsilon\theta_0-\theta_t-\sigma_s\rb
 Z_{\emptyset,\mu}\lb\epsilon\theta_0-\theta_t+\sigma_s\rb         Z_{\lambda,\emptyset}\lb\epsilon\theta_{\infty}+\theta_1
 +\sigma_s\rb  Z_{\mu,\emptyset}\lb\epsilon\theta_{\infty}+\theta_1-\sigma_s\rb
             }{Z_{\lambda,\lambda}\lb -\tfrac{iQ}{2}\rb  
             Z_{\mu,\mu}\lb -\tfrac{iQ}{2}\rb Z_{\lambda,\mu}\lb -\tfrac{iQ}{2}+2\sigma_s\rb Z_{\mu,\lambda}\lb -\tfrac{iQ}{2}-2\sigma_s\rb}.
$$
The series in \eqref{CBseries} is believed to be convergent inside the unit disk $|z|<1$. Another hypothesis is that the only singularities of the conformal blocks as a function of $z$ are branch points at $0,1,\infty$ \cite{Zam87,HJP}. Under this assumption, conformal blocks are naturally defined for $z \in \mathbb{C} \setminus((-\infty,0 ] \cup [1,\infty))$. Moreover, the conformal blocks are believed to be analytic in the external dimensions $\theta_p$, $p=0,t,1,\infty$, and meromorphic in the internal momentum $\sigma$, with the only possible poles located at $\pm \sigma_s^{(m,n)}=-\frac{i}{2}(m b+n b^{-1})$ for $m,n \in \mathbb{Z}_{>0}$, where $b$ appears in the parametrisation (\ref{cQdef}) of the Liouville central charge.
The analytic continuation of the conformal blocks around the singular point $z=0$ can be found from the expansion \eqref{CBseries}:
\beq \label{continuationcb}
\mathcal F\lb\substack{\theta_{1}\;\quad 
      \theta_{t}\\ \sigma_s \\ \theta_{\infty}\quad \theta_0};ze^{\pm 2i\pi}\rb=e^{\pm 2i\pi(\Delta(\sigma_s)-\Delta(\theta_0)-\Delta(\theta_t))}\mathcal F\lb\substack{\theta_{1}\;\quad 
      \theta_{t}\\ \sigma_s \\ \theta_{\infty}\quad \theta_0};z\rb.
\eeq

\subsection{Crossing transformations}
The linear span of conformal blocks forms an infinite-dimensional representation of $\Gamma(\Sigma_{0,4})=\text{PSL}_2(\mathbb{Z})$, the mapping class group of the four-puncture Riemann sphere. It is generated by the braiding $B$ and fusion $F$ moves, such that $F^2=(BF)^3=1$. The three ways of splitting four points on $\mathbb{C}\mathbb{P}^1$ into two pairs define the $s$-, $t$-, and $u$-channel bases for the space of conformal blocks. The cross-ratio argument of conformal blocks in these channels are chosen from $\{z,\frac{z}{z-1}\}$,$\{1-z,\frac{z-1}{z}\}$, and $\{\frac{1}{z},\frac{1}{1-z}\}$, respectively. 
The braiding move B acts on the $s$-channel conformal blocks as follows:
$$
 \mathcal F\lb\substack{\theta_{1}\;\quad 
      \theta_{t}\\ \sigma_s \\ \theta_{\infty}\quad \theta_0};z\rb=e^{\pm i\pi \lb \Delta(\sigma_s)-\Delta(\theta_0)-\Delta(\theta_t)\rb} \lb 1-z\rb^{-2\Delta(\theta_t)} \mathcal F\lb\substack{\theta_{\infty}\quad 
                \theta_{t}\\ \sigma_s \\ \theta_{1}\;\quad \theta_0};\frac{z}{z-1}\rb,\quad \im z \gtrless 0.
$$
On the other hand, the fusion move $F$ is represented by the integral transform
 \beq \label{s-t}
 \qquad\mathcal F\lb\substack{\theta_{1}\;\quad 
     \theta_{t}\\ \sigma_s \\ \theta_{\infty}\quad \theta_0};z\rb=\int_{\mathbb R_+}d\sigma_t ~ 
  F\left[\substack{\theta_1\;\;\;\theta_t\vspace{0.1cm}\\ \theta_{\infty}\;\;\theta_0};\substack{\sigma_s\vspace{0.15cm} \\  \sigma_t}\right]    
     \mathcal F\lb\substack{\theta_{0}\;\quad 
          \theta_{t}\\ \sigma_t \\ \theta_{\infty}\quad \theta_1};1-z\rb.
 \eeq
The kernel of this transformation, called the \textit{Virasoro fusion kernel}, has been related to Racah-Wigner coefficients for a continuous series of representations of $\mathcal{U}_q(sl_2)$, and can be expressed as \cite{pt1,pt2}
\begin{equation}\label{fusion01}
\begin{split}
F\left[\substack{\theta_1\;\;\;\theta_t\vspace{0.1cm}\\ \theta_{\infty}\;\;\theta_0};\substack{\sigma_s \vspace{0.15cm} \\  \sigma_t}\right] = &  \prod_{\epsilon,\epsilon'=\pm1} \frac{g_b \lb \epsilon \theta_1+\theta_{t}+\epsilon' \sigma_t\rb g_b \lb \epsilon \theta_0-\theta_\infty+\epsilon' \sigma_t \rb}{g_b \lb \epsilon \theta_0 + \theta_t + \epsilon' \sigma_s \rb g_b \lb \epsilon \theta_1-\theta_\infty+\epsilon' \sigma_s \rb} \prod_{\epsilon=\pm1} \frac{g_b \lb \frac{iQ}2+2\epsilon \sigma_s\rb}{g_b \lb -\frac{iQ}2+2\epsilon \sigma_t \rb}
	\\
& \times \int_{\mathsf{F}} dx~\prod_{\epsilon=\pm1} \frac{s_b \lb x+ \epsilon \theta_1 \rb s_b \lb x+\epsilon\theta_0+\theta_\infty+\theta_t \rb}{s_b \lb x+\frac{iQ}{2}+\theta_\infty+\epsilon \sigma_s \rb s_b \lb x+\frac{i Q}{2}+\theta_t+\epsilon \sigma_t \rb},
\end{split}
\end{equation}
where the special functions $g_b(x)$ and $s_b(x)$ are defined in the appendix. For $0<b\leq1$, the integrand in (\ref{fusion01}) has eight vertical semi-infinite sequences of poles, four of them increasing and the other four decreasing; the contour of integration $\mathsf{F}$ runs from $-\infty$ to $+\infty$, separating the upper and lower sequences of poles. When the conformal dimensions are real and positive (which is Liouville's spectrum), the contour of integration lies in the strip $\text{Im}~x\in ]-\tfrac{iQ}2,0[$. More generally, the fusion kernel \eqref{fusion01} can be extended to a meromorphic function of all of its parameters provided that $c \in \mathbb{C} \setminus \mathbb{R}_{\leq1}$, which corresponds to $b\notin i \mathbb{R}$. 

Finally, further crossing transformations can be obtained by composing braiding and fusion moves. Let us fix $0<\arg{z}<\pi$. The first two crossing transformations of interest are
\begin{subequations}\label{modulartransfos}
\begin{align}\label{modulartransfosa}
&  \mathcal F\lb\substack{\theta_{1}\;\quad 
      \theta_{t}\\ \sigma_u \\ \theta_0 \quad \theta_\infty};\frac1z\rb=\int_0^{+\infty} d\sigma_s ~ e^{i\pi \lb \Delta(\theta_0)+\Delta(\theta_\infty) -\Delta(\sigma_s)-\Delta(\sigma_u)\rb} ~ F\left[\substack{\theta_0\;\;\;\theta_{t} \vspace{0.1cm} \\ \theta_1\;\;\theta_\infty}
;\substack{\sigma_u \vspace{0.15cm} \\ \sigma_s} \right]  z^{2 \Delta(\theta_t)} \mathcal F\lb\substack{\theta_{1}\;\quad 
\theta_{t}\\ \sigma_s \\ \theta_{\infty}\quad \theta_0};z\rb,	
	\\ \label{modulartransfosb}
& \mathcal F\lb\substack{\theta_\infty\;\quad 
     \theta_{t}\\ \sigma_t \\ \theta_0\quad \theta_1}; 1-\frac1z \rb=\int_0^{+\infty}d\sigma_s ~e^{i\pi\lb \Delta(\sigma_t)-\Delta(\theta_1)-\Delta(\theta_t) \rb} ~
~ F\left[\substack{\theta_0\;\;\;\theta_t\vspace{0.1cm}\\ \theta_\infty\;\;\theta_1};\substack{\sigma_t\vspace{0.15cm} \\  \sigma_s}\right]  ~ z^{2 \Delta(\theta_t)} \ \mathcal F\lb\substack{\theta_1\;\quad 
          \theta_t\\ \sigma_s \\ \theta_{\infty}\quad \theta_0};z\rb. 
\end{align}
\end{subequations}
Such transformations can be seen as infinite-dimensional analogs of the connection formulas for the BPZ equation given by the connection matrices \eqref{connectioni0} and \eqref{connection10}, respectively. In Section \ref{section2}, a suitable confluent limit of these matrices allowed us to recover the solutions of the confluent BPZ equation normalized at $t=\infty$ in any Stokes sector. We will adopt a similar approach to construct the confluent conformal blocks of the second kind in any Stokes sector. The last crossing transformation that we will use is
\beq \label{modulartransfoi1} \mathcal F\lb\substack{\theta_{1}\;\quad 
      \theta_{t}\\ \sigma_u \\ \theta_0 \quad \theta_\infty};\frac1z\rb= \int_0^{+\infty} d\sigma_t ~ F\left[\substack{\theta_1\;\;\;\theta_t\vspace{0.1cm}\\ \theta_0\;\;\theta_\infty};\substack{\sigma_u\vspace{0.15cm} \\  \sigma_t}\right] \mathcal F\lb\substack{\theta_\infty\;\quad 
     \theta_{t}\\ \sigma_t \\ \theta_0\quad \theta_1}; 1-\frac1z\rb. \eeq
This is the analog of the connection formula for the BPZ equation given by the connection matrix \eqref{connectioni1}. In Section \ref{section2}, we recovered the Stokes matrices of the confluent BPZ equation by taking appropriate confluent limits of this connection matrix. We will use a similar approach to find the Stokes transformations acting on the confluent conformal blocks of the second kind.

\section{Confluent conformal blocks of the first kind} \label{section4}
The confluent conformal blocks of the first kind emerge as limits of the $s$-channel Virasoro conformal blocks. This confluent limit is analogous to the confluence of the hypergeometric function ${}_2F_1$ to ${}_1F_1$. Utilizing the change of variables (\ref{confluence}), the confluent conformal blocks of the first kind $\mathcal{B}(t)$ can be expressed as \cite{LNR}
\begin{equation}\label{courtedistance}
\mathcal{B}(t) \equiv \mathcal B\lb \theta_*;\sigma_s;\substack{\theta_t \\ \theta_0};t\rb= \lim_{\Lambda \to \infty}\mathcal{N}_0(\Lambda) ~  \mathcal F\lb\substack{\frac{\Lambda+\theta_*}{2}\quad 
              \theta_{t}\\ \quad\sigma_s \\ \frac{\Lambda-\theta_*}{2}\quad \theta_0};\frac{t}{ib\Lambda}\rb,
\end{equation}
where
\beq \label{qs}
\mathcal{N}_0(\Lambda) = \lb ib\Lambda \rb^{-\Delta(\theta_0)-\Delta(\theta_t)+\Delta(\sigma_s)}.
 \eeq
 Equation (\ref{courtedistance}) is an infinite-dimensional generalization of equation (\ref{limit0}).
The power series expansion of $ \mathcal B(t)$ can be found by substituting the series representation (\ref{CBseries}) for $\mathcal{F}$ into (\ref{courtedistance}) and computing the limit termwise. The result reads \cite{LNR}
 \beq\label{confCB1series}
 \begin{split}
 \mathcal B\lb \theta_*;\sigma_s;\substack{\theta_t \\ \theta_0};t\rb= t^{\Delta(\sigma_s)-\Delta(\theta_0)-\Delta(\theta_t)}e^{-\frac{t}{ib}\lb\theta_t-\frac{iQ}{2}\rb}\sum_{\lambda,\mu\in\mathbb Y}
 \mathcal B_{\lambda,\mu}\lb \theta_*;\sigma_s;\substack{\theta_t \\ \theta_0}\rb
\lb\frac{t}{ib}\rb^{|\lambda|+|\mu|},
\end{split}
 \eeq
where the coefficients $\mathcal B_{\lambda,\mu}$ are given by
 \ben
 \mathcal B_{\lambda,\mu}\lb \theta_*;\sigma_s;\substack{\theta_t \\ \theta_0}\rb=\frac{
  Z_{\lambda,\emptyset}\lb\theta_*+\sigma_s\rb  Z_{\mu,\emptyset}\lb\theta_*-\sigma_s\rb
 \prod_{\epsilon=\pm}
  Z_{\emptyset,\lambda}\lb\epsilon\theta_0-\theta_t-\sigma_s\rb
  Z_{\emptyset,\mu}\lb\epsilon\theta_0-\theta_t+\sigma_s\rb        
              }{Z_{\lambda,\lambda}\lb -\tfrac{iQ}{2}\rb  
              Z_{\mu,\mu}\lb -\tfrac{iQ}{2}\rb Z_{\lambda,\mu}\lb -\tfrac{iQ}{2}+2\sigma_s\rb Z_{\mu,\lambda}\lb -\tfrac{iQ}{2}-2\sigma_s\rb},
  \ebn
  with $Z_{\lambda,\mu}\lb \theta\rb$ defined by (\ref{nekrofun}). It follows from the series representation \eqref{confCB1series} that $\mathcal{B}$ diagonalizes the monodromy operator at $t = 0$:
\beq \label{continuationcb1}
 \mathcal B\lb \theta_*;\sigma_s;\substack{\theta_t \\ \theta_0};te^{-2ik\pi}\rb = e^{-2i\pi k \lb \Delta(\sigma_s)-\Delta(\theta_0)-\Delta(\theta_t) \rb}\mathcal B\lb \theta_*;\sigma_s;\substack{\theta_t \\ \theta_0};t\rb, \qquad k \in \mathbb{Z}.
\eeq

\section{Main results}\label{section5}
Before stating our two main results, Theorem \ref{mainth1} and Theorem \ref{mainth2}, we need to make some assumptions and define the confluent conformal blocks of the second kind. 

\subsection{Assumptions}
In the remainder of this article, we make the following two assumptions. 

\begin{assumption}[Analyticity of $\mathcal{B}(t)$] \label{assumption1}
We assume that the sum over $\lambda, \mu \in \mathbb{Y}$ in \eqref{confCB1series} converges and defines an entire function of $t$, which is furthermore analytic in the parameters $\theta_0, \theta_t,\theta_*$, and meromorphic in $\sigma_s$ except for possible poles located at $\sigma = \pm \sigma_s^{(m,n)}$ for $m,n \in \mathbb{Z}_{>0}$, where $\sigma_s^{(m,n)} = - \frac{i}{2}(m b+n b^{-1})$. 
\end{assumption}
 
\begin{assumption}[Restrictions on the parameters] \label{paramassume}
We assume that
\beq
0 < b< 1, \quad \theta_0 \geq 0, \quad\theta_* \in \mathbb{R}.
\eeq
\end{assumption}

Assumption \ref{assumption1} is believed to be true \cite{LNR}.
Assumption \ref{paramassume} is made primarily for simplicity; we expect all our results to admit analytic continuations to more general values of the parameters, such as $b \in \mathbb{C} \setminus i \mathbb{R}$ and $(\theta_0,\theta_*) \in \mathbb{C}^2$.

\subsection{Confluent conformal blocks of the second kind}
Recall that Proposition \ref{prop2p4} provides a construction of the solution bases $\boldsymbol{D}_n(t)$ of the confluent BPZ equation in any Stokes sector $t \in \Omega_n$ from suitable confluent limits of the renormalized solutions bases $\tilde{\boldsymbol{F}}^\infty(t)$ and $\tilde{\boldsymbol{F}}^1(t)$ of the hypergeometric BPZ equation. We will define the confluent conformal blocks of the second kind by generalizing Proposition \ref{prop2p4} to the nondegenerate case. First, we define the renormalized four-point conformal blocks $\mathcal{\tilde{F}}^0$, $\mathcal{\tilde{F}}^1$, and $\mathcal{\tilde{F}}^\infty$ by 
\beq \label{renormu} \left\{\begin{split}
& \mathcal{\tilde{F}}^0\lb z,\Lambda,\sigma_s\rb = \mathcal{N}_0(\Lambda) ~ \mathcal F\lb\substack{\frac{\Lambda+\theta_*}{2}\;\quad 
      \theta_{t}\\ \quad \sigma_s \\ \frac{\Lambda-\theta_*}{2}\quad \theta_0};z\rb,
          \\
& \mathcal{\tilde{F}}^1\lb z,\Lambda,\nu\rb = z^{-2\Delta(\theta_t)} \mathcal{N}_\infty\lb \Lambda,\nu \rb \mathcal{F}\lb\substack{\frac{\Lambda-\theta_*}{2}\;\;\qquad 
               \theta_{t}\;\;\\  \frac{\Lambda}{2}-\nu \\ 
               \;\;\;\theta_0\qquad \frac{\Lambda+\theta_*}{2}};1-\frac{1}{z} \rb,
	\\
&\mathcal{\tilde{F}}^\infty \lb z,\Lambda,\nu\rb = z^{-2\Delta(\theta_t)} \mathcal{N}_\infty\lb\Lambda,\nu\rb ~ \mathcal{F}\lb\substack{\frac{\Lambda+\theta_*}{2}\;\;\qquad 
               \theta_{t}\;\;\\  \frac{\Lambda}{2}+\nu \\ 
               \;\;\;\theta_0\qquad \frac{\Lambda-\theta_*}{2}};\frac{1}{z} \rb, 
\end{split}\right. 
\eeq
where the normalization factor $\mathcal{N}_0$ is given by \eqref{qs} and the normalization factor $\mathcal{N}_\infty$ is defined by
\beq \label{qu}
\mathcal{N}_\infty(\Lambda,\nu)=e^{i\pi \lb \Delta(\theta_t)+\frac{\theta_*^2}{4}-\nu^2+\Lambda\lb \frac{\theta_*}{2}+\nu\rb \rb} (ib\Lambda)^{\frac{\theta_*^2}{2}-2\nu^2}.
\eeq
We can now define the confluent conformal blocks of the second kind.
\begin{definition}[Confluent conformal blocks of the second kind] The confluent conformal blocks of the second kind $\mathcal D_{n}\lb\substack{\theta_t\\ \theta_*};\nu;\theta_0;t\rb$ are defined by the following confluent limits for any integer $j\geq1$:
\begin{subequations}\label{calDndef}
\begin{align}
& \label{confluentlimitu} e^{2i\pi(j-1)\lb \frac{\theta_*^2}{2}-2\nu^2\rb} \lim\limits_{\substack{\Lambda \to +\infty}} \mathcal{\tilde{F}}^\infty \lb \frac{te^{-2i\pi (j-1)}}{ib\epsilon \Lambda}, \epsilon \Lambda,\nu\rb = \begin{cases} \mathcal D_{2j+1}\lb\substack{\theta_t\\ \theta_*};\nu;\theta_0;t\rb, & t \in \Omega_{2j+1}^-,~ \epsilon=+1,\\ \mathcal D_{2j}\lb\substack{\theta_t\\ \theta_*};\nu;\theta_0;t\rb, & t \in \Omega_{2j}^-,~ \epsilon=-1,\end{cases} \\
&\label{confluentlimitt} e^{2i\pi(j-1)\lb \frac{\theta_*^2}{2}-2\nu^2\rb} \lim\limits_{\substack{\Lambda \to +\infty}} \mathcal{\tilde{F}}^1\lb \frac{te^{-2i\pi (j-1)}}{ib\epsilon\Lambda}, \epsilon \Lambda,\nu\rb = \begin{cases} \mathcal D_{2j}\lb\substack{\theta_t\\ \theta_*};\nu;\theta_0;t\rb, & t \in \Omega_{2j}^+,~\epsilon=+1,  \\ \mathcal D_{2j-1}\lb\substack{\theta_t\\ \theta_*};\nu;\theta_0;t\rb, & t \in \Omega_{2j-1}^+,~\epsilon=-1,\end{cases}
\end{align}
\end{subequations}
where the ``half'' Stokes sectors $\Omega_n^\pm$ are defined in \eqref{halfstokes} and the renormalized conformal blocks $\mathcal{\tilde{F}}^\infty$ and $\mathcal{\tilde{F}}^1$ are defined in \eqref{renormu}. \end{definition}

\subsection{First main result} 
Our first main result provides an explicit integral representation for the confluent conformal blocks of the second kind $\mathcal{D}_n$ in terms of $\mathcal{B}$.

Let $g_b(z)$ and $s_b(z)$ be the special functions defined in the appendix.
We define the kernel $\mathcal{C}_n$ for any integer $n \geq 1$ by
\begin{equation}\label{gnm}
\begin{split}
\mathcal{C}_n\left[\substack{\theta_t\vspace{0.08cm} \\ \theta_{*}\;\;\;\theta_0};\substack{\nu\vspace{0.15cm} \\  \sigma_s}\right]  = P^{(n)}\left[\substack{\theta_t\vspace{0.08cm} \\ \theta_{*}\;\;\;\theta_0};\substack{\nu\vspace{0.15cm} \\  \sigma_s}\right] \displaystyle \int_{\mathsf{C}} dx ~ I^{(n)}\left[\substack{x\;\;\;\theta_t\vspace{0.08cm} \\ \theta_{*}\;\;\;\theta_0};\substack{\nu \vspace{0.15cm} \\  \sigma_s}\right]\end{split}, 
\end{equation}
where the prefactor $P^{(n)}$ is defined by\footnote{Here and in what follows complex powers are defined on the universal cover of $\mathbb{C}\setminus \{0\}$, i.e., $$(e^{2i\pi(\lfloor \frac{n}2 \rfloor-\frac12)}b)^\alpha = b^\alpha e^{2i\pi \alpha (\lfloor \frac{n}2 \rfloor-\frac12)}.$$}
\beq \label{P} \begin{split}
P^{(n)}\left[\substack{\theta_t\vspace{0.08cm} \\ \theta_{*}\;\;\;\theta_0};\substack{\nu\vspace{0.15cm} \\  \sigma_s}\right] = & \left(e^{2i\pi(\lfloor \frac{n}2 \rfloor-\frac12)}b\right)^{\Delta(\theta_0)+\Delta(\theta_t)-\Delta(\sigma_s)+\frac{\theta_*^2}2-2\nu^2}  
	\\
& \times  \displaystyle \prod_{\epsilon=\pm}\frac{g_b\lb \epsilon \sigma_s-\theta_* \rb g_b \lb \epsilon \sigma_s-\theta_0-\theta_t\rb g_b \lb \epsilon \sigma_s+\theta_0-\theta_t \rb}{g_b \lb -\frac{iQ}{2}+2\epsilon \sigma_s \rb g_b\left(\nu-\frac{\theta_*}{2}+\epsilon \theta_0\right)g_b\left(-\theta_t+\epsilon(\nu+\frac{\theta_*}{2})\right)},
\end{split} \eeq
and the integrand $I^{(n)}$ is given by
\beq \label{I}
I^{(n)}\left[\substack{x\;\;\;\theta_t\vspace{0.08cm} \\ \theta_{*}\;\;\;\theta_0};\substack{\nu\vspace{0.15cm} \\  \sigma_s}\right] = e^{(-1)^{n+1}i\pi x\lb \frac{iQ}{2}+\frac{\theta_*}{2}+\theta_t+\nu \rb}\frac{s_b \lb x+\frac{\theta_*}{2}+\nu-\theta_t \rb}{s_b \lb x+ \frac{i Q}{2} \rb } \prod_{\epsilon=\pm} \frac{s_b \lb x+\epsilon \theta_0 +\nu-\frac{\theta_*}{2}\rb }{s_b \lb x+ \frac{iQ}{2}+\nu-\frac{\theta_*}{2}-\theta_t + \epsilon \sigma_s \rb}.
\eeq
The integration contour $\mathsf{C}$ in (\ref{gnm}) is defined as follows. In view of (\ref{polesb}), the numerator in (\ref{I}) has three decreasing semi-infinite sequences of poles, while the denominator has three increasing semi-infinite sequences of zeros. The contour $\mathsf{C}$ in (\ref{gnm}) is any curve from $-\infty$ and $+\infty$ which separates the increasing from the decreasing sequences. For example, if $(\theta_t,\sigma_s, \nu) \in \mathbb{R}^3$, Assumption \eqref{paramassume} implies that the decreasing sequences start at points on the horizontal line $\im x = -\tfrac{iQ}2$, whereas the increasing sequences start at points on the real axis $\im x = 0$. Thus, in this case the contour $\mathsf{C}$ can be any horizontal line in the strip $\im x \in ]-\frac{iQ}{2},0[$. 
More generally, the function \eqref{gnm} extends to a meromorphic function of $(\theta_0,\theta_t,\theta_*,\nu,\sigma_s)$ provided that $b\notin i \mathbb{R}$. 

The following theorem is our first main result. It describes how the confluent conformal blocks of the second kind ${\mathcal D}_n$ can be constructed from $\mathcal{B}$.  

\begin{theorem}[Construction of the confluent conformal blocks of the second kind]\label{mainth1} 
Let $\Omega_n$ denote the Stokes sectors defined in (\ref{stokessector}).
The confluent conformal blocks of the second kind ${\mathcal D}_{n}$ admit the following integral representations:
\begin{equation}\label{dk}\begin{split}
& \mathcal{D}_{n}\lb\substack{\theta_t\\ \theta_*};\nu;\theta_0;t\rb
= \int_0^{+\infty} d\sigma_s~\mathcal{C}_{n}\left[\substack{\theta_t\vspace{0.08cm} \\ \theta_{*}\;\;\;\theta_0};\substack{\nu\vspace{0.15cm} \\  \sigma_s}\right]  \mathcal B\lb \theta_*;\sigma_s;\substack{\theta_t \\ \theta_0};t\rb, \qquad t \in \Omega_{n}, ~ n = 1,2,\dots.
\end{split}\end{equation}
\end{theorem}

The proof of Theorem \ref{mainth1} consists of computing suitable confluent limits of the crossing transformations \eqref{modulartransfos} and is presented in Section \ref{proof1subsec}. The representation (\ref{dk}) is an infinite-dimensional analog of the connection formulas \eqref{conni} relating the bases of solutions $\boldsymbol{D}_n(t)$ and $\boldsymbol{B}(t)$ of the confluent BPZ equation.

\subsection{Second main result}
We define the Stokes kernel $\mathcal{S}_n$ for any integer $n \geq 1$ by
\beq \label{stokessn} \begin{split}
&\mathcal{S}_{n}\left[\substack{\theta_t\vspace{0.08cm} \\ \theta_{*}\;\;\;\theta_0};\substack{\nu_{n+1}\vspace{0.15cm} \\  \nu_n}\right]=\mathcal{P}_{n}\left[\substack{\theta_t\vspace{0.08cm} \\ \theta_{*}\;\;\;\theta_0};\substack{\nu_{n+1}\vspace{0.15cm} \\ \nu_n}\right] \int_{\mathsf{S}}dx~\mathcal{I}_{n}\left[\substack{x\;\;\;\; \theta_t\vspace{0.08cm} \\ \theta_{*}\;\;\;\theta_0};\substack{\nu_{n+1} \vspace{0.15cm} \\  \nu_n}\right],
\end{split}\eeq
where the prefactor $\mathcal{P}_{n}$ is defined by
\begin{align} \nonumber
\mathcal P_{n}\left[\substack{\theta_t\vspace{0.08cm} \\ \theta_{*}\;\;\;\theta_0};\substack{\nu_{n+1}\vspace{0.15cm} \\  \nu_n}\right]= &\; e^{i\pi (-1)^{n+1} \lb \nu_n \theta_0+\nu_{n+1} \theta_t+\frac{iQ}2(\theta_*+\nu_{n+1}-\nu_n)+\frac{\theta_*}2(\theta_0-\theta_t)\rb}\lb e^{i\pi(n-1)} b\rb^{2\nu_n^2-2\nu_{n+1}^2}  
	 \\\nonumber
& \times e^{i\pi \lb \frac{Q^2}4-\frac{iQ}2(\theta_0-\theta_t)-\frac{\theta_*}2(\nu_n-\nu_{n+1}-\frac{\theta_*}2)-\theta_0 \theta_t-\frac{\nu_{n+1}^2+\nu_n^2}2 \rb}  
	\\ \label{matp} 
&\times \prod_{\epsilon=\pm1} \frac{g_b \lb -\theta_0+\epsilon(\nu_n-\frac{\theta_*}{2}) \rb g_b \lb -\theta_t+\epsilon(\nu_n+\frac{\theta_*}{2}) \rb}{g_b \lb -\theta_0+\epsilon(\nu_{n+1}-\frac{\theta_*}{2}) \rb g_b \lb -\theta_t+\epsilon(\nu_{n+1}+\frac{\theta_*}{2})\rb},
\end{align}
the integrand $\mathcal{I}_{n}$ is given by
\beq \begin{split} \label{mati}
\mathcal I_{n}\left[\substack{x\;\;\;\; \theta_t\vspace{0.08cm} \\ \theta_{*}\;\;\;\theta_0};\substack{\nu_{n+1}\vspace{0.15cm} \\  \nu_n}\right]=&\; e^{i\pi x\lb (-1)^{n+1}(\nu_n-\nu_{n+1})-i Q\rb} 
	 \\ 
& \times \frac{s_b \lb x+(-1)^n\theta_* \rb s_b \lb x+\theta_0-\theta_t \rb}{s_b \lb x+ \frac{iQ}{2}+(-1)^{n+1}(\nu_n-\frac{\theta_*}{2})-\theta_t\rb s_b \lb x+\frac{iQ}{2}+\theta_0+(-1)^{n}(\frac{\theta_*}{2}+\nu_{n+1}) \rb},\end{split}
\eeq
and the integration contour $\mathsf{S}$ is such that it separates the two increasing sequences of poles of the integrand from the two decreasing ones. 
If $(\nu_n,\nu_{n+1},\theta_t) \in \mathbb{R}^3$, then Assumption \ref{paramassume} implies that the integrand $\mathcal{I}_{n}$ has two increasing sequences of poles starting from points on the line $\im x=0$, and two decreasing sequences of poles starting from points on the line $\im x = -\frac{iQ}2$. Therefore, in this case the contour of integration $\mathsf{S}$ can be chosen to be any curve in the strip $\im x \in ]-\frac{iQ}{2},0 [$ going from $-\infty$ to $+\infty$. More generally, the Stokes kernel $\mathcal{S}_n$ can be extended to a meromorphic function of $\theta_*,\theta_t,\theta_0,\nu_{n+1},\nu_n$ provided that $b \notin i\mathbb{R}$.

The next theorem, which is our second main result, describes how the confluent conformal blocks in different Stokes sectors are related.

\begin{theorem}[Stokes transformations]\label{mainth2}
For any integer $n \geq 1$, the confluent conformal blocks of the second kind in the two overlapping Stokes sectors $\Omega_{n}$ and $\Omega_{n+1}$ are related by
\beq \label{stokestransform} \begin{split}
{\mathcal D}_{n+1}\lb\substack{\theta_t\\ \theta_*};\nu_{n+1};\theta_0;t\rb=\displaystyle \int_{-\infty}^{+\infty} d\nu_n ~ \mathcal{S}_{n}\left[\substack{\theta_t\vspace{0.08cm} \\ \theta_{*}\;\;\;\theta_0};\substack{\nu_{n+1}\vspace{0.15cm} \\  \nu_n}\right] {\mathcal D}_{n}\lb\substack{\theta_t\\ \theta_*};\nu_n;\theta_0;t\rb, \qquad t \in \Omega_{n} \cap \Omega_{n+1}.
\end{split} \eeq
\end{theorem}

The proof of Theorem \ref{mainth2} consists of taking suitable confluent limits of the crossing transformation \eqref{modulartransfoi1} and is presented in Section \ref{derivationstokes}.
The Stokes transformations (\ref{stokestransform}) are infinite-dimensional analogs of the equations \eqref{stokesrelation} that relate the solutions $\boldsymbol{D}_n(t)$ of the confluent BPZ equation in different Stokes sectors. 

\subsection{Remarks}
We conclude this section with some remarks on the definition (\ref{calDndef}) of $\mathcal{D}_n$.
We observed in Section \ref{section2} that the solution basis $\boldsymbol{D}_n(t)$ given in \eqref{solutionsinf} of the confluent BPZ equation asymptotes to the series $\boldsymbol{D}_{\text{asymp}}(t)$  given in \eqref{asymptotic} as $t$ approaches $\infty$ in the Stokes sector $\Omega_n$. Similarly, the confluent conformal blocks of the second kind ${\mathcal D}_{n}(t)$ are expected to admit a particular asymptotic expansion $\mathcal D_{\text{asymp}}(t)$ in $\Omega_n$. We believe that the expansion $\mathcal D_{\text{asymp}}$ coincides with the one given in \cite[Eq. (1.7)]{LNR}\footnote{Equation \eqref{seried} is equivalent to \cite[Eq. 1.7]{LNR} under the rescaling $t \to \frac{t}{ib}$.}:
\beq \label{seried}
\mathcal D_{\text{asymp}} \lb\substack{\theta_t\\ \theta_*};\nu;\theta_0;t\rb =t^{\frac{\theta_*^2}{2}-2\nu^2} e^{\lb \frac{\theta_*}{2}+\nu\rb \frac{t}{ib}} \lb 1+\sum_{k=1}^\infty \hat{\mathcal{D}}_k\lb\substack{\theta_t\\ \theta_*};\nu;\theta_0\rb \lb \frac{t}{ib}\rb^{-n} \rb,
\eeq
where the first two coefficients are given by
 \begin{align}\nonumber
 {\hat{\mathcal D}}_1\lb\substack{\theta_t\\ \theta_*};\nu;\theta_0\rb=&\; 4\nu^3-\lb 2\Delta(\theta_0)+2\Delta(\theta_t)+
 \theta_*^2\rb\nu +\lb \Delta(\theta_t)-\Delta(\theta_0)\rb \theta_*,
 	\\ \label{firstterms}
 {\hat{\mathcal D}}_2\lb\substack{\theta_t\\ \theta_*};\nu;\theta_0\rb
 =&\; \tfrac12{\hat{\mathcal D}}_1^2\lb\substack{\theta_t\\ \theta_*};\nu;\theta_0\rb +
 3\nu\, {\hat{\mathcal D}}_1\lb\substack{\theta_t\\ \theta_*};\nu;\theta_0\rb
 -2\nu^4+\lb\Delta(\theta_t)-\Delta(\theta_0)\rb\theta_*\nu
 	\\ \nonumber
 & + \tfrac18\lb 4\Delta(\theta_0)-\theta_*^2\rb \lb 4\Delta(\theta_t)-\theta_*^2\rb+\tfrac{c-1}{12}\lb \theta_*^2-4\nu^2\rb.
 \end{align}
The leading asymptotics of the series in \eqref{seried} was found in \cite{GT} by computing a confluent limit of the first terms in the series expansion of the $u$-channel four-point Virasoro conformal blocks. This recipe was extended to higher orders in \cite{LNR}. The left-hand side of \eqref{confluentlimitu} for $j=1$ and $\epsilon=+1$ is similar to the right-hand side of \cite[Eq. (1.7)]{LNR}. However, here we take the confluent limit of the crossing transformation \eqref{modulartransfosa} rather than of the series expansion of the $u$-channel conformal blocks. 

The framework of irregular vertex operators developed in \cite{Nagoya} provides a different but equivalent approach to the construction of $\mathcal D_{\text{asymp}}(t)$. The series in \eqref{seried} is expected to diverge everywhere in the complex plane of $t$ and no closed formula is known for its coefficients.

Two observations suggest that the confluent conformal blocks of the second kind $\mathcal{D}_n(t)$ defined in (\ref{calDndef}) asymptote to $\mathcal D_{\text{asymp}}(t)$ as $t \to \infty$ in the Stokes sector $\Omega_n$. 
First, using \eqref{continuationcb1}, we observe that
$$\mathcal{C}_{n}\left[\substack{\theta_t\vspace{0.08cm} \\ \theta_{*}\;\;\;\theta_0};\substack{\nu\vspace{0.15cm} \\  \sigma_s}\right]e^{2i\pi \lb \frac{\theta_*^2}{2}-2\nu^2 \rb} e^{-2i\pi \lb \Delta(\sigma_s)-\Delta(\theta_0)-\Delta(\theta_t) \rb}=\mathcal{C}_{n+2}\left[\substack{\theta_t\vspace{0.08cm} \\ \theta_{*}\;\;\;\theta_0};\substack{\nu\vspace{0.15cm} \\  \sigma_s}\right].$$
It follows from this relation that the ${\mathcal D}_n(t)$ satisfy the periodicity relation
\beq \label{period}
{\mathcal D}_n\lb\substack{\theta_t\\ \theta_*};\nu;\theta_0;te^{-2i\pi}\rb=e^{-2i\pi ( \frac{\theta_*^2}{2}-2\nu^2 )}{\mathcal D}_{n+2}\lb\substack{\theta_t\\ \theta_*};\nu;\theta_0;t\rb
\eeq
and the phase in \eqref{period} corresponds to the formal monodromy of the asymptotic expansion \eqref{seried}. 

Second, comparison of the first few terms of $\mathcal{D}_{\text{asymp}}$ given by \eqref{firstterms} with the ones of $\boldsymbol{D}_{\text{asymp}}(t)$ given by \eqref{asymptotic} suggests that the following BPZ limit holds:
\beq \label{bpzasymp}
\lim\limits_{\theta_t \to \frac{ib}2+\frac{iQ}2} \begin{pmatrix} \lim\limits_{\nu \to -\frac{\theta_*}2-\frac{ib}2} \mathcal D_{\text{asymp}} \lb\substack{\theta_t\\ \theta_*};\nu;\theta_0;t\rb \\  \lim\limits_{\nu \to -\frac{\theta_*}2+\frac{ib}2}  \mathcal D_{\text{asymp}} \lb\substack{\theta_t\\ \theta_*};\nu;\theta_0;t\rb \end{pmatrix} = \boldsymbol{D}_{\text{asymp}}(t).
\eeq
Moreover, we will show in Proposition \ref{continuationCB} that $\mathcal{D}_n(t)$ tends to the solution $\boldsymbol{D}_n(t)$ of the confluent BPZ equation in the BPZ limit. Also, we know from Section \ref{section2} that $\boldsymbol{D}_n(t)$ asymptotes to $\boldsymbol{D}_{\text{asymp}}(t)$ as $t \to \infty$ in $\Omega_n$. We can summarize these observations as follows:
\begin{center}
\begin{tikzcd}[row sep=0.5cm, column sep = 3.5cm]
\mathcal{D}_n(t) \ar[d,"\text{BPZ limit}"] & \mathcal D_{\text{asymp}}(t) \ar[d,"\text{BPZ limit}"] \\ 
\boldsymbol{D}_n(t) \ar[r,"\text{asymptotic expansion}"] & \boldsymbol{D}_{\text{asymp}}(t)\end{tikzcd}
\end{center}
This diagram suggests that the confluent conformal blocks of the second kind $\mathcal{D}_n(t)$ asymptote to $\mathcal{D}_{\text{asymp}}$ as $t$ approaches $\infty$ in the Stokes sector $\Omega_n$.

\section{Proofs} \label{section6}
We will establish Theorem \ref{mainth1} and Theorem \ref{mainth2} by computing suitable confluent limits of the crossing transformations \eqref{modulartransfos} and \eqref{modulartransfoi1}, respectively.

\subsection{Proof of Theorem \ref{mainth1}}\label{proof1subsec}
The proof of Theorem \ref{mainth1} is achieved by computing confluent limits of the crossing transformations \eqref{modulartransfosa} and \eqref{modulartransfosb}. 
Let us first consider \eqref{modulartransfosa}. Introducing appropriate normalization factors and recalling the definitions (\ref{renormu}) of $\mathcal{\tilde{F}}^\infty$ and $\mathcal{\tilde{F}}^0$, we can write \eqref{modulartransfosa} as
\begin{equation*} \begin{split}
\mathcal{\tilde{F}}^\infty\lb \frac{te^{-2i\pi (j-1)}}{ib\epsilon \Lambda},\epsilon\Lambda,\nu\rb = \int_0^{+\infty} & d\sigma_s ~\frac{\mathcal{N}_\infty(\epsilon \Lambda,\nu)}{\mathcal{N}_0(\epsilon \Lambda)} ~ e^{i\pi \lb \Delta(\theta_0)+\Delta(\frac{\epsilon\Lambda-\theta_*}{2}) -\Delta(\sigma_s)-\Delta(\frac{\epsilon \Lambda}{2}+\nu)\rb} 
	 \\
& \times F\left[\begin{matrix}\substack{\theta_0,\quad\theta_{t}}\\ \substack{\frac{\epsilon \Lambda+\theta_*}{2} ,\frac{\epsilon\Lambda-\theta_*}{2} }
\end{matrix};\begin{matrix}\substack{\frac{\epsilon\Lambda}{2}+\nu} \\ \substack{\sigma_s}\end{matrix} \right] \mathcal{\tilde{F}}^0\lb \frac{te^{-2i\pi (j-1)}}{ib\epsilon \Lambda},\epsilon\Lambda,\sigma_s\rb.
\end{split} \end{equation*}
Since $i \epsilon = e^{\frac{i\pi \epsilon}2}$ for $\epsilon=\pm1$, the factors in the first line of the integrand take the explicit form
\beq \label{div1} \begin{split}
\frac{\mathcal{N}_\infty(\epsilon \Lambda,\nu)}{\mathcal{N}_0(\epsilon \Lambda)} &~ e^{i\pi \lb \Delta(\theta_0)+\Delta(\frac{\epsilon\Lambda-\theta_*}{2}) -\Delta(\sigma_s)-\Delta(\frac{\epsilon \Lambda}{2}+\nu)\rb} 
	\\ 
= &\; e^{2i\pi\lb \frac12+\frac{\epsilon}4\rb \lb \Delta(\theta_0)+\Delta(\theta_t)-\Delta(\sigma_s)+\frac{\theta_*^2}2-2\nu^2 \rb}(b\Lambda)^{\Delta(\theta_0)+\Delta(\theta_t)-\Delta(\sigma_s)+\frac{\theta_*^2}2-2\nu^2}.\end{split} \eeq
Moreover, thanks to \eqref{continuationcb}, the renormalized conformal blocks $\mathcal{\tilde{F}}^0$ satisfies
\beq
\mathcal{\tilde{F}}^0\lb \frac{te^{-2i\pi (j-1)}}{ib\epsilon \Lambda},\epsilon\Lambda,\sigma_s\rb = e^{-2i\pi(j-1)\lb \Delta(\sigma_s)- \Delta(\theta_0)-\Delta(\theta_t)\rb}\mathcal{\tilde{F}}^0\lb \frac{t}{ib\epsilon \Lambda},\epsilon\Lambda,\sigma_s\rb
\eeq
On the other hand, using the confluent limit \eqref{courtedistance} of the $s$-channel conformal blocks, we obtain
\begin{align} \nonumber
e^{2i\pi(j-1)\big( \frac{\theta_*^2}{2}-2\nu^2\big)}\lim\limits_{\substack{\Lambda \to +\infty}}\mathcal{\tilde{F}}^\infty \lb \frac{te^{-2i\pi (j-1)}}{ib\epsilon \Lambda},\epsilon \Lambda,\nu\rb  = & \int_0^{+\infty} d \sigma_s~\lb e^{2i\pi \lb j-\frac12+\frac{\epsilon}4 \rb} b\Lambda \rb^{\Delta(\theta_0)+\Delta(\theta_t)-\Delta(\sigma_s)+\frac{\theta_*^2}2-2\nu^2} 
	\\ \label{inter} 
&\times \lim\limits_{\substack{\Lambda \to +\infty}} F\left[\begin{matrix}\substack{\theta_0,\quad\theta_{t}}\\ \substack{\frac{\epsilon \Lambda+\theta_*}{2} ,\frac{\epsilon\Lambda-\theta_*}{2} }
\end{matrix};\begin{matrix}\substack{\frac{\epsilon\Lambda}{2}+\nu} \\ \substack{\sigma_s}\end{matrix} \right]~\mathcal B\lb \theta_*;\sigma_s;\substack{\theta_t \\ \theta_0};t\rb.
\end{align}
The limit of the Virasoro fusion kernel $F$ remains to be computed. The conformal blocks are symmetric under any sign changes of the parameters, so the Virasoro fusion kernel also has this symmetry. Thus, replacing $\theta_t$ by $-\theta_t$ in \eqref{fusion01} and shifting the contour by $x \to x+\nu-\frac{\theta_*}{2},$ we have
\begin{equation*}
\begin{split}
F\left[\begin{matrix}\substack{\theta_0,~~~~~-\theta_{t}}\\ \substack{\frac{\epsilon \Lambda+\theta_*}{2} ,\frac{\epsilon\Lambda-\theta_*}{2} }
\end{matrix};\begin{matrix}\substack{\frac{\epsilon\Lambda}{2}+\nu} \\ \substack{\sigma_s} \end{matrix} \right]= & \prod_{k=\pm1} \tfrac{g_b \lb \frac{iQ}{2}+k(\epsilon\Lambda+2\nu) \rb g_b \lb -\epsilon \Lambda+k \sigma_s \rb}{g_b \lb -\epsilon \Lambda-\frac{\theta_*}{2}-\nu+k \theta_0 \rb g_b \lb -\theta_t+k(\epsilon \Lambda+\nu-\frac{\theta_*}{2})\rb} \\
&  \times \displaystyle \prod_{k=\pm1}\tfrac{g_b\lb k \sigma_s-\theta_* \rb g_b \lb k \sigma_s-\theta_0-\theta_t\rb g_b \lb k \sigma_s+\theta_0-\theta_t \rb}{g_b \lb -\frac{iQ}{2}+2 k \sigma_s \rb g_b\left(\nu-\frac{\theta_*}{2}+k \theta_0\right)g_b\left(-\theta_t+k(\nu+\frac{\theta_*}{2})\right)} \\
& \times \displaystyle \int_{\mathsf{F}} dx ~ \tfrac{s_b \lb x-\theta_t+\nu-\frac{\theta_*}{2} + \epsilon \Lambda \rb}{s_b \lb x+\frac{iQ}{2}+2\nu+\epsilon \Lambda \rb} \tfrac{s_b \lb x+\frac{\theta_*}{2}+\nu-\theta_t \rb}{s_b \lb x+ \frac{i Q}{2} \rb } \prod_{k=\pm1} \tfrac{s_b \lb x+k \theta_0 +\nu-\frac{\theta_*}{2}\rb }{s_b \lb x+ \frac{iQ}{2}+\nu-\frac{\theta_*}{2}-\theta_t + k \sigma_s \rb}.
\end{split}
\end{equation*}
Using \eqref{asympgammab}, it is straightforward to compute the asymptotics of the first line as $\Lambda \to +\infty$:
\beq \label{div2} \begin{split}
 \prod_{k=\pm1} \tfrac{g_b \lb \frac{iQ}{2}+k(\epsilon\Lambda+2\nu) \rb g_b \lb -\epsilon \Lambda+k \sigma_s \rb}{g_b \lb -\epsilon \Lambda-\frac{\theta_*}{2}-\nu+k \theta_0 \rb g_b \lb -\theta_t+k(\epsilon \Lambda+\nu-\frac{\theta_*}{2})\rb} 
\sim &\; e^{\epsilon \frac{i\pi}{2}\lb \Delta(\sigma_s)-\Delta(\theta_0)-\frac{\theta_*^2}{4}-\nu^2 +\theta_* \theta_t-\nu(\theta_*+2\theta_t+2iQ) \rb}
\\
& \times  e^{-i \pi \Lambda(\frac{i Q}{2}+\frac{\theta_*}{2}+\theta_t+\nu)}\Lambda^{-\lb\Delta(\theta_0)+\Delta(\theta_t)-\Delta(\sigma_s)+\frac{\theta_*^2}2-2\nu^2\rb}.
\end{split}
\eeq
Moreover, the asymptotic formula \eqref{asympsb} for $s_b$ yields, as $\Lambda \to +\infty$,
\beq \label{div3} \begin{split}
\frac{s_b \lb x-\theta_t+\nu-\frac{\theta_*}{2} +\epsilon \Lambda \rb}{s_b \lb x+\frac{iQ}{2}+2\nu + \epsilon \Lambda \rb} \sim &\; e^{\epsilon \frac{i\pi}{2}\lb -\Delta(\theta_t)-\frac{\theta_*^2}{4}+3\nu^2-\theta_* \theta_t+\nu(\theta_*+2\theta_t+2iQ)\rb} \\
&\times e^{\epsilon i\pi x\lb \frac{iQ}{2}+\frac{\theta_*}{2}+\theta_t+\nu \rb} e^{i \pi \Lambda(\frac{i Q}{2}+\frac{\theta_*}{2}+\theta_t+\nu)}.
\end{split} \eeq
Multiplication of the preceding two equations produces a factor 
$$(e^{\epsilon \frac{i\pi}2} \Lambda)^{-\lb\Delta(\theta_0)+\Delta(\theta_t)-\Delta(\sigma_s)+\frac{\theta_*^2}2-2\nu^2\rb}e^{\epsilon i\pi x\lb \frac{iQ}{2}+\frac{\theta_*}{2}+\theta_t+\nu \rb}.$$
The first part of this factor cancels part of the integrand in \eqref{inter}, and the second part yields the phase in the integrand of the confluent fusion kernel \eqref{gnm}. Finally, the two families of confluent fusion kernels $\mathcal{C}_{2j+1}$ and $\mathcal{C}_{2j}$ are obtained after gathering the phases and taking $\epsilon = +1$ and $\epsilon = -1$, respectively. To summarize, we have shown that 
\begin{align*}
 e^{2i\pi(j-1)\lb \frac{\theta_*^2}{2}-2\nu^2\rb} & \lim\limits_{\substack{\Lambda \to +\infty}} \mathcal{\tilde{F}}^\infty \lb \frac{te^{-2i\pi (j-1)}}{ib\epsilon \Lambda}, \epsilon \Lambda,\nu\rb  
	\\
& =\begin{cases}  \int_0^{+\infty} d\sigma_s~\mathcal{C}_{2j+1}\left[\substack{\theta_t\vspace{0.08cm} \\ \theta_{*}\;\;\;\theta_0};\substack{\nu\vspace{0.15cm} \\  \sigma_s}\right]  \mathcal B\lb \theta_*;\sigma_s;\substack{\theta_t \\ \theta_0};t\rb, &  t \in \Omega_{2j+1}^-,~\epsilon=+1, \\
\int_0^{+\infty} d\sigma_s~\mathcal{C}_{2j}\left[\substack{\theta_t\vspace{0.08cm} \\ \theta_{*}\;\;\;\theta_0};\substack{\nu\vspace{0.15cm} \\  \sigma_s}\right]  \mathcal B\lb \theta_*;\sigma_s;\substack{\theta_t \\ \theta_0};t\rb, & t \in \Omega_{2j}^-,~\epsilon=-1.  \end{cases}
\end{align*}
Recalling the definition (\ref{calDndef}) of $\mathcal{D}_n$, this concludes the proof of Theorem \ref{mainth1} in the case when $t \in \cup_{n=1}^\infty \Omega_n^-$. 

The proof when $t \in \cup_{n=1}^\infty \Omega_n^+$ is rather similar and consists of computing the confluent limit of the crossing transformation \eqref{modulartransfosb}. Introducing the relevant normalization factors, \eqref{modulartransfosb} becomes
\begin{equation*} \begin{split}
\mathcal{\tilde{F}}^1\lb\frac{te^{-2i\pi (j-1)}}{ib\epsilon\Lambda},\epsilon\Lambda,\nu\rb = \int_0^{+\infty} & d\sigma_s ~ \frac{\mathcal{N}_\infty(\epsilon \Lambda,\nu)}{\mathcal{N}_0(\epsilon \Lambda)} ~ e^{i\pi \lb \Delta\lb\frac{\epsilon \Lambda}2-\nu\rb-\Delta\lb \frac{\epsilon \Lambda+\theta_*}2\rb -\Delta(\theta_t) \rb} \\
& \times F\left[\begin{matrix}\substack{\theta_0,~~~~~\theta_{t}}\\ \substack{\frac{\epsilon\Lambda-\theta_*}{2} ,\frac{\epsilon\Lambda+\theta_*}{2} }
\end{matrix};\begin{matrix}\substack{\frac{\epsilon\Lambda}{2}-\nu} \\ \substack{\sigma_s}\end{matrix} \right] \mathcal{\tilde{F}}^0\lb\frac{te^{-2i\pi (j-1)}}{ib\epsilon\Lambda},\epsilon\Lambda,\sigma_s\rb.
\end{split} \end{equation*}
The factors in the first line of the integrand take the form
\beq
\frac{\mathcal{N}_\infty(\epsilon \Lambda,\nu)}{\mathcal{N}_0(\epsilon \Lambda)} ~ e^{i\pi \lb \Delta\lb\frac{\epsilon \Lambda}2-\nu\rb-\Delta\lb \frac{\epsilon \Lambda+\theta_*}2\rb -\Delta(\theta_t) \rb} = \lb e^{\frac{i\pi \epsilon}{2}} b \Lambda\rb^{\Delta(\theta_0)+\Delta(\theta_t)-\Delta(\sigma_s)+\frac{\theta_*^2}2-2\nu^2}.
\eeq 
Using the analytic continuation \eqref{continuationcb} and the confluent limit \eqref{courtedistance}, we obtain
\begin{align}\nonumber
e^{2i\pi(j-1)\lb \frac{\theta_*^2}{2}-2\nu^2\rb}\lim\limits_{\substack{\Lambda \to +\infty}}\mathcal{\tilde{F}}^\infty\lb\frac{te^{-2i\pi (j-1)}}{ib\epsilon\Lambda},\epsilon\Lambda\rb  
= & \int_0^{+\infty} d \sigma_s~\lb e^{2i\pi \lb j-1+\frac{\epsilon}4 \rb} b\Lambda \rb^{\Delta(\theta_0)+\Delta(\theta_t)-\Delta(\sigma_s)+\frac{\theta_*^2}2-2\nu^2}
	\\ \label{inter2} 
&\times \lim\limits_{\substack{\Lambda \to +\infty}} F\left[\begin{matrix}\substack{\theta_0,~~~~~\theta_{t}}\\ \substack{\frac{\epsilon\Lambda-\theta_*}{2} ,\frac{\epsilon\Lambda+\theta_*}{2} }
\end{matrix};\begin{matrix}\substack{\frac{\epsilon\Lambda}{2}-\nu} \\ \substack{\sigma_s}\end{matrix} \right] ~\mathcal B\lb \theta_*;\sigma_s;\substack{\theta_t \\ \theta_0};t\rb.
\end{align}
Performing a contour shift $x \to x+\nu-\frac{\theta_*}2$ and using the even symmetry of the bottom left parameter of the Virasoro fusion kernel, we find from \eqref{fusion01} that
\begin{equation*}
\begin{split}
F\left[\begin{matrix}\substack{\theta_0,~~~~~-\theta_{t}}\\ \substack{-\frac{\epsilon\Lambda-\theta_*}{2} ,\frac{\epsilon\Lambda+\theta_*}{2} }
\end{matrix};\begin{matrix}\substack{\frac{\epsilon\Lambda}{2}-\nu} \\ \substack{\sigma_s} \end{matrix} \right]= & \prod_{k=\pm1} \tfrac{g_b \lb \frac{iQ}{2}+k(\epsilon\Lambda-2\nu) \rb g_b \lb \epsilon\Lambda+k \sigma_s \rb}{g_b \lb \epsilon\Lambda-\frac{\theta_*}{2}-\nu+k \theta_0 \rb g_b \lb -\theta_t+k(-\epsilon \Lambda+\nu-\frac{\theta_*}{2})\rb} \\
& \times \displaystyle \prod_{k=\pm1}\tfrac{g_b\lb k \sigma_s-\theta_* \rb g_b \lb k \sigma_s-\theta_0-\theta_t\rb g_b \lb k \sigma_s+\theta_0-\theta_t \rb}{g_b \lb -\frac{iQ}{2}+2 k \sigma_s \rb g_b\left(\nu-\frac{\theta_*}{2}+k \theta_0\right)g_b\left(-\theta_t+k(\nu+\frac{\theta_*}{2})\right)} \\
& \times \displaystyle \int_{\mathsf{F}} dx \tfrac{s_b \lb x-\theta_t+\nu-\frac{\theta_*}2-\epsilon \Lambda \rb}{s_b \lb x+\frac{iQ}2+2\nu-\epsilon \Lambda \rb} \tfrac{s_b \lb x+\frac{\theta_*}2+\nu-\theta_t \rb}{s_b \lb x+\frac{iQ}2 \rb} \prod_{k=\pm1} \tfrac{s_b \lb x+k\theta_0+\nu-\frac{\theta_*}2 \rb}{s_b \lb x+\frac{iQ}2+\nu-\frac{\theta_*}2-\theta_t+k\sigma_s \rb}.
\end{split}
\end{equation*}
Using \eqref{asympgammab}, the first line has the following asymptotics as $\Lambda \to + \infty$:
\beq \begin{split}
\prod_{k=\pm1} \tfrac{g_b \lb \frac{iQ}{2}+k(\epsilon\Lambda-2\nu) \rb g_b \lb \epsilon\Lambda+k\sigma_s \rb}{g_b \lb \epsilon\Lambda-\frac{\theta_*}{2}-\nu+k \theta_0 \rb g_b \lb -\theta_t+k(-\epsilon \Lambda+\nu-\frac{\theta_*}{2})\rb} \sim &\; e^{\epsilon\frac{ i\pi}2 \lb \Delta(\theta_0)-\Delta(\sigma_s)+\frac{\theta_*^2}4+\nu^2+2iQ\nu+2\theta_t \nu-\theta_* \theta_t +\theta_* \nu \rb}
	\\
&\times e^{-i\pi \Lambda \lb \theta_t+\nu+\frac{\theta_*}{2}+\frac{iQ}{2} \rb}  \Lambda^{-\lb\Delta(\theta_0)+\Delta(\theta_t)-\Delta(\sigma_s)+\frac{\theta_*^2}2-2\nu^2 \rb}
\end{split}\eeq 
and using \eqref{asympsb} we find, as $\Lambda \to +\infty$,
\begin{equation}\begin{split}
\frac{s_b \lb x-\theta_t+\nu-\frac{\theta_*}2-\epsilon \Lambda\rb}{s_b \lb x+\frac{iQ}2+2\nu-\epsilon \Lambda\rb} \sim &\; e^{\epsilon\frac{i\pi}{2} \lb \Delta(\theta_t)+\frac{\theta_*^2}4-3\nu^2-2iQ\nu-2\theta_t \nu+\theta_* \theta_t-\theta_* \nu \rb} e^{-\epsilon i \pi x \lb \frac{iQ}{2}+\frac{\theta_*}{2}+\theta_t+\nu \rb} \\
&\times e^{i\pi \Lambda \lb \theta_t+\nu+\frac{\theta_*}{2}+\frac{iQ}{2} \rb}.
\end{split}\end{equation}
The multiplication of the preceding two equations produces a factor
$$\lb e^{-\epsilon\frac{i \pi}2} \Lambda \rb^{-\lb\Delta(\theta_0)+\Delta(\theta_t)-\Delta(\sigma_s)+\frac{\theta_*^2}2-2\nu^2\rb}e^{-i\pi x \epsilon\lb \frac{iQ}2+\frac{\theta_*}2+\theta_t+\nu \rb}.$$
Substitution into \eqref{inter2} leads to the family of confluent fusion kernels $\mathcal{C}_{2j}$ and $\mathcal{C}_{2j-1}$ for $\epsilon=+1$ and $\epsilon=-1$, respectively. This proves \eqref{dk} also for $t \in \cup_{n=1}^\infty \Omega_n^+$ and concludes the proof of Theorem \ref{mainth1}.

\subsection{Proof of Theorem \ref{mainth2}}\label{derivationstokes}
Theorem \ref{mainth2} will be established by computing an appropriate confluent limit of the crossing transformation \eqref{modulartransfoi1} relating the $u$- and $t$-channel conformal blocks. The cases of odd and even $n$ will be considered separately.

\subsubsection{Derivation of the Stokes transformations for $n=2j-1$.} \label{oddcasesubsec}
The integrand in \eqref{modulartransfoi1} is an even function of $\sigma_t$ and the Virasoro fusion kernel $F$ is an even function of $\theta_t$. Thus, performing the change of variables 
\beq \label{change1}
\sigma_t = -\frac{\Lambda}{2}-\nu_{2j-1},   
\eeq
and using the symmetry $\theta_t \to -\theta_t$ of the Virasoro fusion kernel, we can write \eqref{modulartransfoi1} as
\beq \label{modulartransfoi1b} 
\mathcal F\lb\substack{\theta_1\;\quad 
      \theta_{t}\\ \sigma_u \\ \theta_0 \quad \theta_\infty};\frac1z\rb
      = \int_{-\frac{\Lambda}{2}}^{+\infty} d\nu_{2j-1} ~ F\left[\substack{\theta_1\;\;\; -\theta_t\vspace{0.1cm}\\ \theta_0\;\;\theta_\infty};\substack{\sigma_u\vspace{0.15cm} \\  -\frac{\Lambda}{2} - \nu_{2j-1}}\right] \mathcal F\lb\substack{\theta_\infty \quad\quad\quad 
     \theta_{t}\\ -\frac{\Lambda}{2} - \nu_{2j-1} \\ \theta_0\quad\quad\quad \theta_1}; 1-\frac1z\rb. 
 \eeq
Introducing appropriate normalization factors, letting 
$$\theta_1=\frac{-\Lambda+\theta_*}{2}, \quad \theta_\infty=\frac{-\Lambda-\theta_*}{2}, \quad z = -\frac{te^{-2i\pi (j-1)}}{ib\Lambda}, \quad \sigma_u = -\frac{\Lambda}{2} + \nu_{2j}, $$
recalling the definitions (\ref{renormu}) of $\mathcal{\tilde{F}}^\infty$ and $\mathcal{\tilde{F}}^0$, and taking the limit $\Lambda \to +\infty$, we can write \eqref{modulartransfoi1b} as
\begin{equation*}
\begin{split}
&\lim\limits_{\substack{\Lambda \to +\infty}}e^{2i\pi(j-1)\lb \frac{\theta_*^2}{2}-2\nu_{2j}^2\rb} \mathcal{\tilde{F}}^\infty \lb -\frac{te^{-2i\pi (j-1)}}{ib\Lambda},-\Lambda,\nu_{2j}\rb =\lim\limits_{\substack{\Lambda \to +\infty}} \displaystyle \int_{-\Lambda/2}^{+\infty} d \nu_{2j-1}~e^{4i\pi(j-1)\lb \nu_{2j-1}^2-\nu_{2j}^2 \rb} 
	 \\
& \hspace{1.5cm}  \times \frac{\mathcal{N}_\infty(-\Lambda,\nu_{2j})}{\mathcal{N}_\infty(-\Lambda,\nu_{2j-1})} F\left[\substack{\frac{-\Lambda+\theta_*}{2}\;\;\;-\theta_t\vspace{0.1cm}\\ \theta_0\;\;\frac{-\Lambda-\theta_*}{2}};\substack{-\frac{\Lambda}{2}+\nu_{2j}\vspace{0.15cm} \\ -\tfrac{\Lambda}{2}-\nu_{2j-1}}\right] e^{2i\pi(j-1)\lb \frac{\theta_*^2}{2}-2\nu_{2j-1}^2\rb}  \mathcal{\tilde{F}}^1\lb -\frac{te^{-2i\pi (j-1)}}{ib\Lambda},-\Lambda,\nu_{2j-1}\rb.
\end{split}
\end{equation*}
Using the limits \eqref{confluentlimitu} and \eqref{confluentlimitt} for $\epsilon=-1$, we obtain
\begin{equation*}\begin{split}
{\mathcal D}_{2j}\lb\substack{\theta_t\\ \theta_*};\nu_{2j};\theta_0;t\rb=\displaystyle \int_{-\infty}^{+\infty}&d\nu_{2j-1}~e^{4i\pi(j-\frac34)\lb \nu_{2j-1}^2-\nu_{2j}^2 \rb}e^{-i\pi \Lambda(\nu_{2j}-\nu_{2j-1})} \lb -i b\Lambda \rb^{2\nu_{2j-1}^2-2\nu_{2j}^2} \\
& \times \lim\limits_{\substack{\Lambda \to +\infty}}  F\left[\substack{\frac{-\Lambda+\theta_*}{2}\;\;\;-\theta_t\vspace{0.1cm}\\ \theta_0\;\;\frac{-\Lambda-\theta_*}{2}};\substack{-\frac{\Lambda}{2}+\nu_{2j} \vspace{0.15cm} \\ -\tfrac{\Lambda}{2}-\nu_{2j-1}}\right] {\mathcal D}_{2j-1} \lb\substack{\theta_t\\ \theta_*};\nu_{2j-1};\theta_0;t\rb.
\end{split}\end{equation*}
Deforming the contour of integration $\mathsf{F}$ in the expression (\ref{fusion01}) for the Virasoro fusion kernel by shifting $x \to x-\frac{\Lambda+\theta_*}2$, we obtain \begin{equation*}
\begin{split}
F\left[\substack{\frac{-\Lambda+\theta_*}{2}\;\;\;-\theta_t\vspace{0.1cm}\\ \theta_0\;\;\frac{-\Lambda-\theta_*}{2}};\substack{-\frac{\Lambda}{2}+\nu_{2j}\vspace{0.15cm} \\ -\tfrac{\Lambda}{2}-\nu_{2j-1}}\right]
= &\; X_1(\Lambda) \displaystyle \prod_{\epsilon=\pm1} \tfrac{g_b \lb -\theta_0+\epsilon(\nu_{2j-1}-\frac{\theta_*}{2}) \rb g_b \lb -\theta_t+\epsilon(\nu_{2j-1}+\frac{\theta_*}{2}) \rb}{g_b \lb -\theta_0+\epsilon(\nu_{2j}-\frac{\theta_*}{2}) \rb g_b \lb -\theta_t+\epsilon(\nu_{2j}+\frac{\theta_*}{2}) \rb} 
	\\
& \times \displaystyle \int_{\mathsf{F}} dx X_2(x, \Lambda) \tfrac{s_b \lb x-\theta_*\rb s_b \lb x+\theta_0-\theta_t \rb}{s_b \lb x+\frac{iQ}{2}+\nu_{2j-1}-\frac{\theta_*}{2}-\theta_t \rb s_b \lb x+\frac{iQ}{2}+\theta_0-\nu_{2j}-\frac{\theta_*}{2}\rb},
\end{split}
\end{equation*}
where
\begin{align*}
& X_1(\Lambda) = \prod_{\epsilon=\pm}\tfrac{g_b\lb \frac{iQ}{2}+\epsilon (-\Lambda+2\nu_{2j})  \rb g_b \lb -\theta_0+\epsilon(\Lambda+\frac{\theta_*}{2}+\nu_{2j-1}) \rb g_b \lb -\theta_t+\epsilon(\Lambda+\nu_{2j-1}-\frac{\theta_*}{2})\rb}{g_b\lb -\frac{iQ}{2}+\epsilon(\Lambda+2\nu_{2j-1}) \rb g_b \lb -\theta_0+\epsilon(-\Lambda+\nu_{2j}+\frac{\theta_*}{2}) \rb g_b\lb -\theta_t+\epsilon(-\Lambda-\frac{\theta_*}{2}+\nu_{2j}) \rb )},
	\\
& X_2(x,\Lambda) = \tfrac{s_b \lb x-\Lambda \rb s_b \lb x+\theta_0-\theta_*-\theta_t-\Lambda \rb}{s_b \lb x+\frac{iQ}{2}-\nu_{2j-1}-\frac{\theta_*}{2}-\theta_t-\Lambda \rb s_b \lb x+\frac{iQ}{2}+\theta_0-\frac{\theta_*}{2}+\nu_{2j}-\Lambda \rb}.
\end{align*}	
We have the following asymptotics as $\Lambda \to +\infty$:
\begin{align*} 
 X_1(\Lambda) \sim &\; e^{\pi Q \Lambda}e^{-i\pi \lb \theta_*(\theta_t-\theta_0)-(\nu_{2j-1}+\nu_{2j})(\theta_t+\theta_0)+iQ(\nu_{2j-1}-\nu_{2j})\rb} \Lambda^{-2\nu_{2j-1}^2+2\nu_{2j}^2},
	\\
X_2(x,\Lambda) \sim &\; e^{i\pi \Lambda(iQ-\nu_{2j-1}+\nu_{2j})}e^{i\pi x(-\nu_{2j}+\nu_{2j-1}-iQ)} 
 e^{i\pi\big[ \frac{Q^2}{4}+iQ\nu_{2j-1}-\theta_*(\theta_0+\nu_{2j-1}-\theta_t)-\frac{\nu_{2j}^2}{2}-(\theta_0+\frac{\theta_*}{2})(-\frac{\theta_*}{2}+\theta_t)\big]}
 	\\
& \times e^{i\pi\big[\nu_{2j}(iQ-\theta_0+\frac{\theta_*}{2})+\frac{iQ}{2}(-3\nu_{2j}-\theta_0+\theta_*+\theta_t-\nu_{2j-1})-\nu_{2j-1}(-\frac{\theta_*}{2}+\theta_t+\frac{\nu_{2j-1}}{2}) \big]}. 
\end{align*}
The Stokes kernels $\mathcal{S}_{2j-1}$, $j=1,2,\dots$, are obtained after gathering the phases. This proves Theorem \ref{mainth2} for odd values of $n$. 
\\

\subsubsection{Derivation of the Stokes transformations for $n=2j$.} 
We perform the following change of variables in the integrand of \eqref{modulartransfoi1}:
$$\sigma_t=\frac{\Lambda}2-\nu_{2j}.$$
The confluent limit of the crossing transformation \eqref{modulartransfoi1} can then be written as
\begin{equation*}
\begin{split}
&\lim\limits_{\substack{\Lambda \to +\infty}}e^{2i\pi(j-1)\lb \frac{\theta_*^2}{2}-2\nu_{2j+1}^2\rb} \mathcal{\tilde{F}}^\infty \lb \frac{te^{-2i\pi (j-1)}}{ib\Lambda},\Lambda,\nu_{2j+1}\rb 
= \lim\limits_{\substack{\Lambda \to +\infty}} \int_{-\infty}^{+\frac{\Lambda}{2}} d\nu_{2j}~ e^{4i\pi(j-1)\lb \nu_{2j}^2-\nu_{2j+1}^2 \rb} 
	\\
&\hspace{1.5cm} \times \frac{\mathcal{N}_\infty(\Lambda,\nu_{2j+1})}{\mathcal{N}_\infty(\Lambda,\nu_{2j})} F\left[\substack{\frac{\Lambda+\theta_*}{2}\;\;\;-\theta_t\vspace{0.1cm}\\ \theta_0\;\;\frac{\Lambda-\theta_*}{2}};\substack{\frac{\Lambda}{2}+\nu_{2j+1} \vspace{0.15cm} \\ \tfrac{\Lambda}{2}-\nu_{2j}}\right]e^{2i\pi(j-1)\lb \frac{\theta_*^2}{2}-2\nu_{2j}^2\rb}  \mathcal{\tilde{F}}^1\lb \frac{te^{-2i\pi (j-1)}}{ib\Lambda},\Lambda,\nu_{2j}\rb.
\end{split}
\end{equation*}
Using the limits \eqref{confluentlimitu} and \eqref{confluentlimitt} for $\epsilon=+1$, we obtain
\begin{equation*}\begin{split}
\mathcal D_{2j+1}\lb\substack{\theta_t\\ \theta_*};\nu_{2j+1};\theta_0;t\rb = &  \int_{-\infty}^{+\infty} d\nu_{2j}~e^{4i\pi(j-\frac34)\lb \nu_{2j}^2-\nu_{2j+1}^2 \rb}e^{i\pi \Lambda(\nu_{2j+1}-\nu_{2j})} \\
& \times \lb i b\Lambda \rb^{2\nu_{2j}^2-2\nu_{2j+1}^2} \lim\limits_{\substack{\Lambda \to +\infty}} F\left[\substack{\frac{\Lambda+\theta_*}{2}\;\;\;-\theta_t\vspace{0.1cm}\\ \theta_0\;\;\frac{\Lambda-\theta_*}{2}};\substack{\frac{\Lambda}{2}+\nu_{2j+1}\vspace{0.15cm} \\ \tfrac{\Lambda}{2}-\nu_{2j}}\right]
 {\mathcal D}_{2j} \lb\substack{\theta_t\\ \theta_*};\nu_{2j};\theta_0;t\rb.
\end{split} \end{equation*}
The evaluation of the confluent limit of the Virasoro fusion kernel $F$ is similar to the evaluation presented in subsection \ref{oddcasesubsec}. In the end, one arrives at the Stokes kernel $\mathcal{S}_{2j}$, which proves Theorem \ref{mainth2} also for even values of $n$.

\section{The BPZ limit} \label{section7}
In this section, we verify explicitly that the connection formula \eqref{dk} of Theorem \ref{mainth1} reduces to the connection formula \eqref{conni} of the confluent BPZ equation in the BPZ limit. We also verify that the Stokes transformation \eqref{stokestransform} of Theorem \ref{mainth2} reduces to the Stokes formula \eqref{stokesrelation}. 

We will make the following assumption on the BPZ limit of the confluent conformal blocks of the first kind $\mathcal{B}(t)$. 

\begin{assumption}[BPZ limit of $\mathcal{B}(t)$]
We assume that the following BPZ limit holds:
\beq \label{limitB}
\lim\limits_{\theta_t \to \frac{i Q}{2}+\frac{ib}{2}} \begin{pmatrix} \lim\limits_{\sigma \to \theta_0-\theta_t+\frac{iQ}2} \mathcal B\lb \theta_*;\sigma_s;\substack{\theta_t \\ \theta_0};t\rb \\ \lim\limits_{\sigma \to \theta_0+\theta_t-\frac{iQ}2}\mathcal B\lb \theta_*;\sigma_s;\substack{\theta_t \\ \theta_0};t\rb \end{pmatrix} = \boldsymbol{B}(t),
\eeq  
where $\boldsymbol{B}(t)$ is the degenerate confluent block of the first kind defined in \eqref{whittaker1}. \end{assumption}

The limits in (\ref{limitB}) can be verified numerically to high order by expanding both sides in power series of $t$, but we are not aware of an analytic proof.

\subsection{BPZ limit of $\mathcal{C}_n$} 
We first compute the BPZ limit of the right-hand side of equation \eqref{dk}.

\begin{proposition}[BPZ limit of $\mathcal{C}_n$]\label{BPZlimitCnprop}
Define $\nu_\pm=-\frac{\theta_*}{2}\pm \frac{ib}{2}$. The following limit holds:
\beq \label{continuationodd}
\lim_{\theta_t \to \frac{iQ}{2}+\frac{ib}{2}} \begin{pmatrix} \lim\limits_{\substack{\nu \to \nu_-}} 
\int_0^{+\infty} d\sigma_s~\mathcal{C}_{n}\left[\substack{\theta_t\vspace{0.08cm} \\ \theta_{*}\;\;\;\theta_0};\substack{\nu\vspace{0.15cm} \\  \sigma_s}\right]  \mathcal B\lb \theta_*;\sigma_s;\substack{\theta_t \\ \theta_0};t\rb \\
\lim\limits_{\substack{\nu \to \nu_+}}
\int_0^{+\infty} d\sigma_s~\mathcal{C}_{n}\left[\substack{\theta_t\vspace{0.08cm} \\ \theta_{*}\;\;\;\theta_0};\substack{\nu\vspace{0.15cm} \\  \sigma_s}\right]  \mathcal B\lb \theta_*;\sigma_s;\substack{\theta_t \\ \theta_0};t\rb
\end{pmatrix}= C_n \boldsymbol{B}(t), \qquad n = 1, 2, \dots,
\eeq
where $\boldsymbol{B}(t)$ is the solution basis for the confluent BPZ equation defined in \eqref{whittaker1}.
\end{proposition}
\begin{proof}
The function $P^{(n)} $ defined in \eqref{P} can be split into two parts as follows:
$$P^{(n)}\left[\substack{\theta_t\vspace{0.08cm} \\ \theta_{*}\;\;\;\theta_0};\substack{\nu\vspace{0.15cm} \\  \sigma_s}\right] = P_1(\nu,\theta_t,\theta_*)P_2^{(n)}\left[\substack{\theta_t\vspace{0.08cm} \\ \theta_{*}\;\;\;\theta_0};\substack{\nu\vspace{0.15cm} \\  \sigma_s}\right],$$
where
\beq \label{P1}
P_1(\nu,\theta_t,\theta_*)=\frac{1}{g_b(-\theta_t + \nu+\frac{\theta_*}{2})g_b(-\theta_t-\nu-\frac{\theta_*}{2})},
\eeq
and
\beq \label{P2}
P_2^{(n)}\left[\substack{\theta_t\vspace{0.08cm} \\ \theta_{*}\;\;\;\theta_0};\substack{\nu\vspace{0.15cm} \\  \sigma_s}\right]=\left(e^{2i\pi(\lfloor \frac{n}2 \rfloor-\frac12)}b\right)^{\Delta(\theta_0)+\Delta(\theta_t)-\Delta(\sigma_s)+\frac{\theta_*^2}2-2\nu^2} \displaystyle \prod_{\epsilon=\pm}\tfrac{g_b\lb \epsilon \sigma_s-\theta_* \rb g_b \lb \epsilon \sigma_s-\theta_0-\theta_t\rb g_b \lb \epsilon \sigma_s+\theta_0-\theta_t \rb}{g_b \lb -\frac{iQ}{2}+2\epsilon \sigma_s \rb g_b\left(\nu-\frac{\theta_*}{2}+\epsilon \theta_0\right)}.
\eeq
Using this notation, recalling the definition (\ref{gnm}) of $\mathcal{C}_n$, and adopting the short-hand notation $\mathcal{B}(t) \equiv \mathcal B\lb \theta_*;\sigma_s;\substack{\theta_t \\ \theta_0};t\rb$, we can  write the integral on the left-hand side of (\ref{continuationodd}) as follows:
\begin{align}\label{continuation}
\int_0^{+\infty} d\sigma_s~\mathcal{C}_{n}\left[\substack{\theta_t\vspace{0.08cm} \\ \theta_{*}\;\;\;\theta_0};\substack{\nu\vspace{0.15cm} \\  \sigma_s}\right]  \mathcal{B}(t)
= \displaystyle P_1(\nu,\theta_t,\theta_*) \int_{\mathbb{R}_+} d\sigma_s \displaystyle \int_{\mathsf{C}} dx ~
P_2^{(n)}\left[\substack{\theta_t\vspace{0.08cm} \\ \theta_{*}\;\;\;\theta_0};\substack{\nu\vspace{0.15cm} \\  \sigma_s}\right]  I^{(n)}\left[\substack{x\;\;\;\theta_t\vspace{0.08cm} \\ \theta_{*}\;\;\;\theta_0};\substack{\nu \vspace{0.15cm} \\  \sigma_s}\right] \mathcal{B}(t).
\end{align}

Before presenting a detailed evaluation of the BPZ limit of (\ref{continuation}), we briefly describe the main idea of the argument. Recall that $0<b<1$. The special functions $s_b(z)$ and $g_b(z)$ defined in \eqref{defsb} and \eqref{gb} possess semi-infinite sequences of poles, whose locations are given by \eqref{polesb} and (\ref{polegb}), respectively.
The prefactor $P_1(\nu,\theta_t,\theta_*)$ in (\ref{continuation}) satisfies
\beq \label{pref} 
P_1(\nu_\pm,\theta_t,\theta_*)= \frac{1}{g_b\lb -\frac{ib}{2} -\theta_t \rb g_b \lb \frac{ib}{2}-\theta_t  \rb},
\eeq
thus, by (\ref{polegb}), $P_1(\nu_\pm,\theta_t,\theta_*)$ has a double zero at $\theta_t=\frac{iQ}{2}+\frac{ib}{2}$. 
In the BPZ limit, the contours of integration $\mathsf{C}$ and $\mathbb{R}_+$ in (\ref{continuation}) are pinched between pairs of moving poles; a similar mechanism was described in \cite{HJP,H,PSS}. For example, the pinching of the integration contour $\mathbb{R}_+$ is due to the factor $g_b\lb \sigma_s-\theta_0-\theta_t \rb g_b \lb -\sigma_s+\theta_0-\theta_t \rb$ in \eqref{P2}. Indeed, this factor has two poles which cross $\mathbb{R}_+$ in the limit $\theta_t \to \frac{iQ}{2}+\frac{ib}{2}$: the pole of $g_b \lb\sigma_s-\theta_0-\theta_t\rb$ located at $\sigma_s=\theta_0+\theta_t-\frac{iQ}{2}$ crosses $\mathbb{R}_+$ and collides with the pole $\sigma_s=\theta_0-\theta_t+\frac{iQ}{2} + ib$ of $g_b\lb -\sigma_s+\theta_0-\theta_t \rb$, and the pole of $g_b\lb -\sigma_s+\theta_0-\theta_t \rb$ located at $\sigma_s=\theta_0-\theta_t+\frac{iQ}{2}$ crosses $\mathbb{R}_+$ and collides with the pole $\sigma_s=\theta_0+\theta_t-\frac{iQ}{2} - ib$ of $g_b \lb\sigma_s-\theta_0-\theta_t\rb$, see Figure \ref{fig:fig2b}. Therefore, before taking the limit, we deform the contour $\mathbb{R}_+$ to a new contour $\mathbb{R}_+'$ as shown in Figure \ref{fig:fig2b} and pick up two residue contributions from the two poles; these contributions are easily computed with the help of \eqref{resgb}. A similar mechanism occurs for the integral over $x$. In the end, after performing the deformations of both $\mathsf{C}$ and $\mathbb{R}_+$, we are able to express the double integral on the right-hand side of (\ref{continuation}) as a sum of three types of terms: terms which are regular at $\theta_t=\frac{iQ}{2}+\frac{ib}{2}$, terms which have a simple pole at $\theta_t=\frac{iQ}{2}+\frac{ib}{2}$, and terms which have a double pole at $\theta_t=\frac{iQ}{2}+\frac{ib}{2}$.
Since the prefactor $P_1$ has a double zero at $\theta_t=\frac{iQ}{2}+\frac{ib}{2}$, only those terms that have a double pole will yield a nonzero contribution to \eqref{continuation} in the limit $\theta_t \to \frac{iQ}{2}+\frac{ib}{2}$. Computing this contribution explicitly and using that $\mathcal B(t) \to \boldsymbol{B}(t)$ in the BPZ limit by Assumption \ref{assumption1}, the proposition will follow. 
 
 \tikzset{->-/.style={decoration={
  markings,
  mark=at position .5 with {\arrow{>}}},postaction={decorate}}}
 \begin{figure}[h!]
\centering
\subcaptionbox{Deformation of the contour $\mathsf{C}'$ \label{fig:fig2a}}{\begin{tikzpicture}[x=1cm,y=0.4cm]
  \draw[-stealth] (-3,0)--(2.5,0) node[right]{\text{Re}~x}; 
  \draw[-stealth] (0,-8)--(0,5) node[above]{\text{Im}~x}; 
  \draw[dashed,black,->-] (-3,-3)--(3,-3) node[above]{}; 
\draw[dashed,blue](-3,-3) .. controls (-3,-3) .. (-2.1,-4.5)
		            (-1,-4.5).. controls (0,-3)..(3,-3);	          
\draw[dashed,blue,->] (-2.1,-4.5) to[bend right] (-1,-4.5);	
\node at (0.2,0.5) {$0$};
\node at (-2.5,-4) {$\substack{\theta_t -\frac{iQ}2-\frac{ib}2}$};
  \node at(3,-2.4){${\color{black}{\mathsf{F}'}}$};
  \node at(3,-3.7){${\color{blue}{\mathsf{F}''}}$};
\draw[->-](-1.5,-4)--(0,0);
 \draw[dashed,red,decoration={markings, mark=at position 0.125 with {\arrow{<}}},
        postaction={decorate}] (-1.5,-4) circle (0.2 cm);

 \fill (-1.5,-4)  circle[radius=1.5pt];   \fill (-1.5,-6)  circle[radius=1.5pt]; \fill (-1.5,-8)  circle[radius=1.5pt];
  
 \fill (-.5,-4)  circle[radius=1.5pt];  \fill (-.5,-6)  circle[radius=1.5pt]; \fill (-.5,-8)  circle[radius=1.5pt];
    
 \fill (1.5,-4)  circle[radius=1.5pt];   \fill (1.5,-6)  circle[radius=1.5pt];  \fill (1.5,-8)  circle[radius=1.5pt];

\fill (-1,-1)  circle[radius=1.5pt];    \fill (-1,1)  circle[radius=1.5pt];  \fill (-1,3)  circle[radius=1.5pt];
\fill (0,0)  circle[radius=1.5pt];  \fill (0,2)  circle[radius=1.5pt];  \fill (0,4)  circle[radius=1.5pt];
\fill (1,-1)  circle[radius=1.5pt];   \fill (1,1)  circle[radius=1.5pt];  \fill (1,3)  circle[radius=1.5pt];
\end{tikzpicture}}\hspace{1cm}\subcaptionbox{Deformation of the contour $\mathbb{R}_+$ \label{fig:fig2b}}{\begin{tikzpicture}[x=1cm,y=0.4cm]
  \draw[-stealth] (0,0)--(6,0) node[right]{\text{Re}~$\sigma_s$}; 
  \draw[-stealth] (0,-8)--(0,5) node[above]{\text{Im}~$\sigma_s$}; 
\draw[dashed,blue,->](0,0) .. controls (0,0) .. (.5,-2)
			(.5,-2) .. controls (1,-3.7) .. (1.5,-2)
			(1.5,-2) .. controls (2,0) .. (2.5,2)
			(2.5,2) .. controls (3,3.7) .. (3.5,2)
			(3.5,2) .. controls (4,0) .. (5,0);
 \draw[dashed,red,decoration={markings, mark=at position 0.125 with {\arrow{>}}},
        postaction={decorate}] (3,1) circle (0.2 cm);
\draw[dashed,red,decoration={markings, mark=at position 0.125 with {\arrow{<}}},
        postaction={decorate}] (1,-1) circle (0.2 cm);
  \fill (3,1)  circle[radius=1.5pt];   \fill (3,7)  circle[radius=1.5pt];  
 \fill (1,-1)  circle[radius=1.5pt];   \fill (1,-7)  circle[radius=1.5pt]; 
\node at (4.2,7) {$\substack{\theta_0-\theta_t+\frac{iQ}{2}+ib}$};
\node at (4,1) {$\substack{\theta_0-\theta_t+\frac{iQ}{2}}$};
\node at (0.05,-1) {$\substack{\theta_0+\theta_t-\frac{iQ}{2}}$};
\node at (-0.2,-7) {$\substack{\theta_0+\theta_t-\frac{iQ}{2}-ib}$};
\node at (1.9,1) {\color{blue}{{$\mathbb{R}'_+$}}};
\draw[->] (1,-1)--(1.9,2.8);
\draw[->] (3,1)--(2.1,-2.8);
\draw[->] (3,7)--(2.1,3.3);
\draw[->] (1,-7)--(1.9,-3.3);
  \fill (2,3)  circle[radius=1.5pt];   \fill (2,-3)  circle[radius=1.5pt]; 
  \node at (1.5,3.8) {$\substack{\theta_0+\tfrac{ib}{2}}$};
  \node at (2.5,-3.6) {$\substack{\theta_0-\tfrac{ib}{2}}$};
\end{tikzpicture}}
\caption{Schematic illustration of the deformations of the contours $\mathsf{C}'$ and $\mathbb{R}_+$.\label{fig:fig2}}
\end{figure}
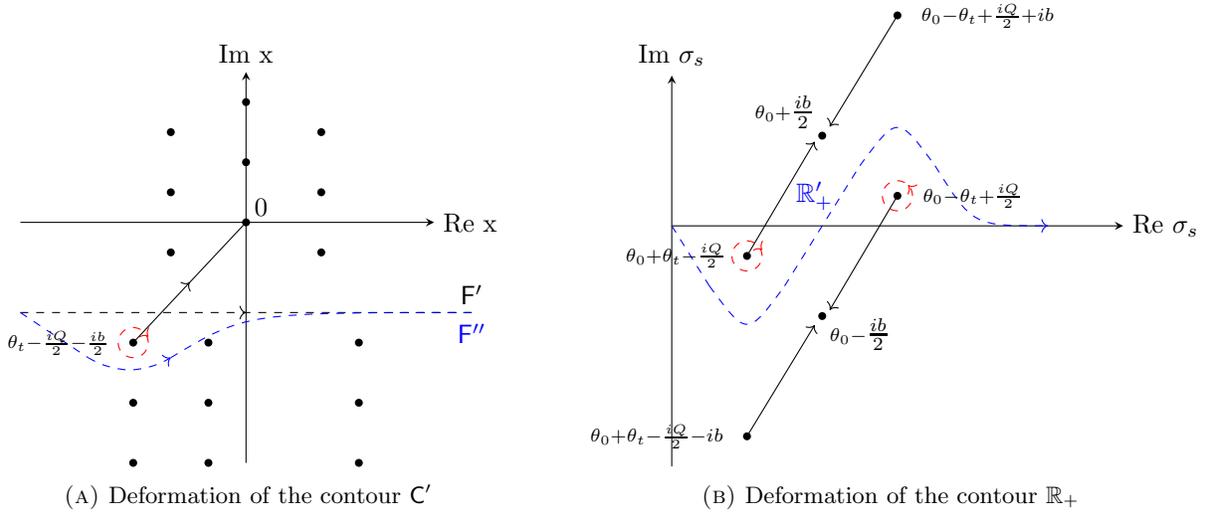

We consider the two components of equation (\ref{continuationodd}) involving the limits $\nu \to \nu_+$ and $\nu \to \nu_-$ separately. 

\textbf{The case $\nu \to \nu_+$.} 
We first deform the contour $\mathsf{C}$ in \eqref{continuation} into a suitable contour $\mathsf{C'}$ such that no pole in $x$ crosses $\mathsf{C}'$ in the limit $\nu\to\nu_+$. For $\nu = \nu_+$, we have 
$$I^{(n)}\left[\substack{x\;\;\;\theta_t\vspace{0.08cm} \\ \theta_{*}\;\;\;\theta_0};\substack{\nu_+\vspace{0.15cm} \\  \sigma_s}\right] = e^{\frac{\pi x}{2} (-1)^n \lb b+Q-2i \theta_t \rb}  \tfrac{s_b \lb x+\frac{ib}{2}-\theta_t \rb}{s_b \lb x+\frac{iQ}{2} \rb} \times \displaystyle \prod_{j=\pm1} \tfrac{s_b \lb x+\frac{ib}{2}+j \theta_0-\theta_* \rb}{s_b \lb x+\frac{ib}{2}+\frac{iQ}{2}-\theta_t-\theta_* +j\sigma_s \rb}.$$
In the limit $\theta_t \to \frac{iQ}{2}+\frac{ib}{2}$, the pole of $s_b \lb x+\frac{ib}{2}-\theta_t \rb$ located at $x=-\frac{iQ}{2}-\frac{ib}{2}+\theta_t$ crosses the contour of integration $\mathsf{C}'$ and collides with the pole of $s_b \big( x+\frac{iQ}{2} \big)^{-1}$ located at $x=0$. Therefore, we choose to deform the contour $\mathsf{C}'$ as shown in Figure \ref{fig:fig2a}, so that the integral over $\mathsf C'$ becomes an integral over $\mathsf C''$ plus a residue at $x=\theta_t-\frac{iQ}2-\frac{ib}2$. Moreover, as described above, we deform the integration contour $\mathbb{R}_+$ as in Figure \ref{fig:fig2b}.

Using the contour deformations described in Figure \ref{fig:fig2}, we find that the BPZ limit of \eqref{continuation} is given by
\begin{align} \label{X}
& \lim_{\theta_t \to \frac{i Q}{2}+\frac{ib}{2}} \lim_{\nu \to \nu_+} 
\int_0^{+\infty} d\sigma_s~\mathcal{C}_{n}\left[\substack{\theta_t\vspace{0.08cm} \\ \theta_{*}\;\;\;\theta_0};\substack{\nu\vspace{0.15cm} \\  \sigma_s}\right]  \mathcal{B}(t)
= T_1 + T_2 + T_3 + T_4,
\end{align}
where $T_1, T_2, T_3, T_4$ are given by
\begin{align*}
& T_1 = \lim_{\theta_t \to \frac{iQ}{2}+\frac{ib}{2}} P_1(\nu_+,\theta_t,\theta_*) \displaystyle \int_{\mathbb{R}'_+} d\sigma_s \displaystyle \int_{\mathsf{C}''}dx P_2^{(n)}\left[\substack{\theta_t\vspace{0.08cm} \\ \theta_{*}\;\;\;\theta_0};\substack{\nu_+\vspace{0.15cm} \\  \sigma_s}\right] I^{(n)}\left[\substack{x\;\;\;\theta_t\vspace{0.08cm} \\ \theta_{*}\;\;\;\theta_0};\substack{\nu_+\vspace{0.15cm} \\ \sigma_s}\right] \mathcal{B}(t),
	\\
& T_2 = -2i\pi\lim_{\theta_t \to \frac{iQ}{2}+\frac{ib}{2}} P_1(\nu_+,\theta_t,\theta_*) \displaystyle \sum_{\epsilon=\pm1}  \displaystyle \int_{\mathsf{C}''}dx \epsilon\underset{\sigma_s=\theta_0+\epsilon(\theta_t-\frac{iQ}{2})}{\text{Res}}\left\{P_2^{(n)}\left[\substack{\theta_t\vspace{0.08cm} \\ \theta_{*}\;\;\;\theta_0};\substack{\nu_+\vspace{0.15cm} \\  \sigma_s}\right]I^{(n)}\left[\substack{x\;\;\;\theta_t\vspace{0.08cm} \\ \theta_{*}\;\;\;\theta_0};\substack{\nu_+\vspace{0.15cm} \\ \sigma_s}\right]\mathcal{B}(t) \right\},
	\\ \nonumber
& T_3 =  - 2i \pi \lim_{\theta_t \to \frac{iQ}{2}+\frac{ib}{2}} P_1(\nu_+,\theta_t,\theta_*)  \displaystyle \int_{\mathbb{R}'_+} d\sigma_s  P_2^{(n)}\left[\substack{\theta_t\vspace{0.08cm} \\ \theta_{*}\;\;\;\theta_0};\substack{\nu_+\vspace{0.15cm} \\  \sigma_s}\right] \underset{x=\theta_t-\frac{iQ}{2}-\frac{ib}{2}}{\text{Res}}   I^{(n)}\left[\substack{x\;\;\;\theta_t\vspace{0.08cm} \\ \theta_{*}\;\;\;\theta_0};\substack{\nu_+\vspace{0.15cm} \\  \sigma_s}\right] \mathcal{B}(t),
	\\ \nonumber
& T_4 = -4\pi^2 \lim_{\theta_t \to \frac{iQ}{2}+\frac{ib}{2}}P_1(\nu_+,\theta_t,\theta_*)\displaystyle \sum_{\epsilon=\pm1}\epsilon\underset{\sigma_s=\theta_0+\epsilon(\theta_t-\frac{iQ}{2})}{\text{Res}} \left\{ P_2^{(n)}\left[\substack{\theta_t\vspace{0.08cm} \\ \theta_{*}\;\;\;\theta_0};\substack{\nu_+\vspace{0.15cm} \\ \sigma_s}\right] \underset{x=\theta_t-\frac{iQ}{2}-\frac{ib}{2}}{\text{Res}} I^{(n)}\left[\substack{x\;\;\;\theta_t\vspace{0.08cm} \\ \theta_{*}\;\;\;\theta_0};\substack{\nu_+\vspace{0.15cm} \\ \sigma_s}\right] \mathcal{B}(t) \right\}.
\end{align*}
The contours $\mathsf{C}''$ and $\mathbb{R}'_+$ are not pinched in the limit $\theta_t \to \frac{iQ}{2}+\frac{ib}{2}$. Therefore the double integral in $T_1$ is regular at $\theta_t=\frac{iQ}{2}+\frac{ib}{2}$; in view of the double zero of $P_1(\nu_+,\theta_t,\theta_*)$, it follows that $T_1 = 0$. 
On the other hand, we have
\begin{align}\label{sigmasresidue}
-2i\pi \underset{\sigma_s=\theta_0+\epsilon(\theta_t-\tfrac{iQ}{2})}{\text{Res}}P_2^{(n)}\left[\substack{\theta_t\vspace{0.08cm} \\ \theta_{*}\;\;\;\theta_0};\substack{\sigma_s\vspace{0.15cm} \\  \nu_+}\right] = g_b \lb \frac{iQ}{2}-2\theta_t \rb C_\epsilon, \qquad \epsilon=\pm1,
\end{align}
where $C_\epsilon$ is regular in the limit $\theta_t \to \frac{iQ}{2}+\frac{ib}{2}$. Consequently, the pole of the right-hand side of (\ref{sigmasresidue}) at $\theta_t = \frac{iQ}{2}+\frac{ib}{2}$ is only simple, so again, thanks to the double zero of $P_1(\nu_+,\theta_t,\theta_*)$, we have $T_2 = 0$. 
The term $T_3$ also vanishes because of a similar argument. Indeed, a calculation yields
$$-2i \pi\underset{x=\theta_t-\frac{iQ}{2}-\frac{ib}{2}}{\text{Res}} ~ I^{(n)}\left[\substack{x\;\;\;\theta_t\vspace{0.08cm} \\ \theta_{*}\;\;\;\theta_0};\substack{\nu_+\vspace{0.15cm} \\ \sigma_s}\right]=e^{(-1)^{n+1}\frac{i\pi}4 \lb (b+Q)^2+4\theta_t^2\rb} \tfrac{1}{s_b \lb-\frac{ib}{2}+\theta_t \rb} \displaystyle \prod_{j=\pm} \tfrac{s_b \lb -\frac{i Q}{2}+j \theta_0-\theta_*+\theta_t \rb}{s_b \lb -\theta_*+ j \sigma_s \rb},$$
and the pole of this expression at $\theta_t = \frac{iQ}{2}+\frac{ib}{2}$ is only simple. 

It only remains to compute $T_4$. Let us write
$$T_4 = 
 \lim_{\theta_t \to \frac{iQ}{2}+\frac{ib}{2}}P_1(\nu_+,\theta_t,\theta_*)\displaystyle \sum_{\epsilon=\pm1}
 X_\epsilon \mathcal{B},$$
where $X_\epsilon$ is defined for $\epsilon = \pm 1$ by
$$X_\epsilon = -4\pi^2\epsilon\underset{\sigma_s=\theta_0 +\epsilon (\theta_t-\frac{iQ}{2})}{\text{Res}} \left\{P_2^{(n)}\left[\substack{\theta_t\vspace{0.08cm} \\ \theta_{*}\;\;\;\theta_0};\substack{\nu_+\vspace{0.15cm} \\ \sigma_s}\right] \underset{x=\theta_t-\frac{iQ}{2}-\frac{ib}{2}}{\text{Res}} I^{(n)}\left[\substack{x\;\;\;\theta_t\vspace{0.08cm} \\ \theta_{*}\;\;\;\theta_0};\substack{\nu_+\vspace{0.15cm} \\ \sigma_s}\right]\right\}.$$
A tedious but straightforward computation shows that
\begin{align} \nonumber
X_\epsilon = &\; g_b \lb \frac{ib}{2}-\theta_t\rb g_b \lb \frac{iQ}{2}-2\theta_t \rb\tfrac{g_b \big( \tfrac{iQ}{2} \big) e^{\frac{i\pi}{4}(-1)^{n+1}\lb (b+Q)^2+4\theta_t^2 \rb}}{g_b \lb -\frac{ib}{2}+\theta_t \rb} \lb be^{2i\pi(\lfloor \frac{n}2 \rfloor-\frac12)}\rb^{\frac{b^2}{2}+ib\theta_*+(\frac{Q}{2} +\epsilon i\theta_0)(Q+2i\theta_t)}  
	\\ \label{xpm} 
& \times \tfrac{g_b \lb -\frac{iQ}{2} +\epsilon 2\theta_0 \rb g_b \lb -\frac{iQ}{2}+\epsilon \theta_0+\theta_*+\theta_t\rb}{g_b \lb -\frac{3iQ}{2}+2\epsilon\theta_0+2\theta_t \rb g_b \lb \frac{iQ}{2}+\epsilon\theta_0+\theta_*-\theta_t \rb} \prod_{j=\pm1} \tfrac{g_b \lb -\tfrac{iQ}{2}+j\theta_0-\theta_*+\theta_t\rb}{g_b \lb \frac{ib}{2}+j\theta_0-\theta_* \rb}, \qquad \epsilon=\pm1.
\end{align}
Due to the first two factors on the right-hand side of \eqref{xpm}, $X_\epsilon$ has a double pole at $\theta_t = \frac{iQ}{2}+\frac{ib}{2}$. This pole is canceled by the double zero of $P_1(\nu_+,\theta_t,\theta_*)$ and a computation yields, for $\epsilon=\pm1$,
\begin{align*}
\mathcal{R}_\epsilon := \lim_{\theta_t \to \frac{iQ}{2}+\frac{ib}{2}} P_1(\nu_+,\theta_t, \theta_*) X_\epsilon =  b^{-\frac12-ib(\epsilon \theta_0-\theta_*)} e^{2\pi b \lb \lfloor \frac{n}2 \rfloor -\frac12 \rb \lb \frac{ib}2-\frac{iQ}2+\epsilon \theta_0-\theta_* \rb} \tfrac{g_b \lb -\frac{iQ}{2}+ 2\epsilon\theta_0 \rb}{g_b \lb ib-\frac{iQ}{2}+ 2\epsilon\theta_0 \rb}\tfrac{g_b \lb \frac{ib}{2}+\theta_*+\epsilon \theta_0 \rb}{g_b \lb -\frac{ib}{2}+\theta_*+\epsilon \theta_0 \rb}.
\end{align*}
Using the identity \eqref{propgb} satisfied by the function $g_b(z)$, we find
\beq
\mathcal{R}_\epsilon=e^{-2i\pi \lb \lfloor \frac{n}2 \rfloor-\frac12 \rb \lb \frac12+ib(\epsilon \theta_0-\theta_*) \rb} \frac{\Gamma \lb -2ib\epsilon \theta_0 \rb}{\Gamma \lb \frac12-ib(\epsilon \theta_0+\theta_*) \rb}, \qquad \epsilon = \pm1.
\eeq
We conclude that
\beq
\mathcal R_+ = K_{-+}^{(n)}(\theta_*,-\theta_0), \quad \mathcal R_- = K_{-+}^{(n)}(\theta_*,\theta_0),
\eeq
where $K_{-+}^{(n)}(\theta_*,\theta_0)$ is defined in \eqref{kelement}. Using the BPZ limit $\mathcal{B}(t) \to \boldsymbol{B}(t) = (B_+(t), B_-(t))$ given in \eqref{limitB}, it is concluded that
\beq \begin{split}
\lim_{\theta_t \to \frac{i Q}{2}+\frac{ib}{2}} \lim_{\nu \to \nu_+} & \int_0^{+\infty} d\sigma_s~\mathcal{C}_n\left[\substack{\theta_t\vspace{0.08cm} \\ \theta_{*}\;\;\;\theta_0};\substack{\nu\vspace{0.15cm} \\  \sigma_s}\right] \mathcal B\lb \theta_*;\sigma_s;\substack{\theta_t \\ \theta_0};t\rb 
= T_4 = \mathcal{R}_-B_+(t) + \mathcal{R}_+B_-(t)
	\\
&=K_{-+}^{(n)}(\theta_*,\theta_0)B_+(t) +K_{-+}^{(n)}(\theta_*,-\theta_0) B_-(t).\end{split} 
\eeq
In view of the expression (\ref{knm}) for $C_n$, the right-hand side equals the second component of $C_n \boldsymbol{B}(t)$. This completes the proof of the second component of \eqref{continuationodd}. \\

\textbf{The case $\nu \to \nu_-$.} A similar mechanism occurs in this case. 
For $\nu = \nu_-$, we have
$$I^{(n)} \left[\substack{x\;\;\;\theta_t\vspace{0.08cm} \\ \theta_{*}\;\;\;\theta_0};\substack{\nu_-\vspace{0.15cm} \\  \sigma_s}\right] = e^{\frac{\pi x}{2}(-1)^n \lb -b+Q-2i \theta_t \rb}  \tfrac{s_b \lb x-\frac{ib}{2}-\theta_t \rb}{s_b \lb x+\frac{iQ}{2} \rb} \times \displaystyle \prod_{j=\pm1} \tfrac{s_b \lb x- \frac{ib}{2} +j \theta_0-\theta_* \rb}{s_b \lb x- \frac{ib}{2}+\frac{iQ}{2}-\theta_t-\theta_* +j \sigma_s \rb}.$$
In the limit $\theta_t \to \frac{iQ}{2}+\frac{ib}{2}$, the poles of $s_b \lb x-\frac{ib}{2}-\theta_t \rb$ located at $x_1:=\theta_t+\frac{ib}{2}-\frac{iQ}{2}$ and $x_2:=\theta_t+\frac{ib}{2}-\frac{iQ}{2}-ib$ cross the contour of integration and collide with the poles of $s_b \lb x+\frac{iQ}{2}\rb^{-1}$ located at $x=+ib$ and $x=0$, respectively. Deforming the contour $\mathsf{C}'$ into $\mathsf{C}''$ in a similar manner as before, and noting that the analogs of $T_1, T_2$, and $T_3$ vanish also in this case, we arrive at
\beq \label{limlimnuminus}
\begin{split}
 \lim_{\theta_t \to \frac{i Q}{2}+\frac{ib}{2}} \lim_{\nu \to \nu_-}
& \int_0^{+\infty} d\sigma_s~\mathcal{C}_{n}\left[\substack{\theta_t\vspace{0.08cm} \\ \theta_{*}\;\;\;\theta_0};\substack{\nu\vspace{0.15cm} \\  \sigma_s}\right]  \mathcal{B}(t)
 =-4\pi^2 \lim_{\theta_t \to \frac{iQ}{2}+\frac{ib}{2}}P_1(\nu_-,\theta_t,\theta_*)
 	\\
&\times \sum_{\epsilon=\pm1} \epsilon \underset{\sigma_s=\theta_0+\epsilon(\theta_t-\frac{iQ}{2})}{\text{Res}} \left\{ P_2^{(n)}\left[\substack{\theta_t\vspace{0.08cm} \\ \theta_{*}\;\;\;\theta_0};\substack{\nu_- \vspace{0.15cm} \\  \sigma_s}\right] \displaystyle \sum_{j=1,2}\underset{x=x_j}{\text{Res}} I^{(n)}\left[\substack{x\;\;\;\theta_t\vspace{0.08cm} \\ \theta_{*}\;\;\;\theta_0};\substack{\nu_-\vspace{0.15cm} \\  \sigma_s}\right] \mathcal{B}(t) \right\}.
\end{split} \eeq
The residues at $x_1$ and $x_2$ can be computed using the property \eqref{differencesb} of the function $s_b(z)$. We obtain
\begin{equation*} \begin{split} & X_1 :=  -2i\pi \underset{x=x_1}{\text{Res}} I^{(n)}\left[\substack{x\;\;\;\theta_t\vspace{0.08cm} \\ \theta_{*}\;\;\;\theta_0};\substack{\nu_-\vspace{0.15cm} \\  \sigma_s}\right]
=e^{(-1)^{n+1}i\pi\lb \theta_t^2+\frac{1}{4b^2} \rb}\frac{\operatorname{sech}{\lb \pi b \theta_t \rb}}{2s_b \lb \theta_t-\frac{ib}{2}\rb}\displaystyle \prod_{j=\pm1} \frac{s_b \lb -\frac{iQ}{2}+\theta_t-\theta_*+j \theta_0\rb}{s_b \lb -\theta_*+j \sigma_s\rb}, \\
& \resizebox{\hsize}{!}{$X_2 :=  -2i\pi \underset{x=x_2}{\text{Res}} I^{(n)}\left[\substack{x\;\;\;\theta_t\vspace{0.08cm} \\ \theta_{*}\;\;\;\theta_0};\substack{\nu_-\vspace{0.15cm} \\  \sigma_s}\right]
=e^{(-1)^{n+1}i\pi(\frac{Q}{2}-\frac{b}{2}-i\theta_t)(\frac{Q}{2}+\frac{b}{2}+i\theta_t)}\frac{\operatorname{sec}{\lb \frac{\pi b}{2} (b+Q) \rb}}{2s_b \lb \theta_t-\frac{ib}{2}\rb}\displaystyle \prod_{j=\pm1} \frac{s_b \lb -ib-\frac{iQ}{2}+\theta_t-\theta_*+j \theta_0\rb}{s_b \lb -ib-\theta_*+j \sigma_s\rb}.$}\end{split}\end{equation*}
Defining $R_{j,\epsilon}$ for $j = 1,2$ and $\epsilon = \pm 1$ by
$$R_{j,\epsilon} = -2i\pi  \lim_{\theta_t \to \frac{iQ}{2}+\frac{ib}{2}}P_1(\nu_-,\theta_t,\theta_*)\epsilon\underset{\sigma_s=\theta_0+\epsilon(\theta_t-\frac{iQ}{2})}{\text{Res}} \left\{P_2^{(n)}\left[\substack{\theta_t\vspace{0.08cm} \\ \theta_{*}\;\;\;\theta_0};\substack{\nu_- \vspace{0.15cm} \\  \sigma_s}\right]X_j\right\},$$
a long computation yields, for $\epsilon=\pm1$,
\begin{equation*} \begin{split}
R_{1,\epsilon} = & -e^{2 i \pi  \left(\left\lfloor \frac{n}{2}\right\rfloor -\frac{1}{2}\right)
   \left(-\frac{1}{2}-i b (\theta_*+\epsilon \theta_0)\right)}e^{(-1)^n i\pi bQ}\csc \left(\pi  b^2\right) \cosh (\pi  b (\theta_*+\epsilon\theta_0))\frac{\Gamma (-2 i b \epsilon  \theta_0)}{\Gamma \left(\frac{1}{2}-i b (\epsilon \theta_0-\theta_*)\right)}, 
   	\\
R_{2,\epsilon} = & -e^{2 i \pi  \left(\left\lfloor \frac{n}{2}\right\rfloor -\frac{1}{2}\right)
   \left(-\frac{1}{2}-i b (\theta_*+\epsilon \theta_0)\right)}  \csc \left(\pi  b^2\right) \cosh (\pi  b (i b+\theta_*+\epsilon\theta_0)) \frac{\Gamma (-2 i b \epsilon  \theta_0)}{\Gamma \left(\frac{1}{2}-i b (\epsilon \theta_0-\theta_*)\right)}.\end{split}
   \end{equation*}
It follows that
$$R_+ :=R_{1,+}+R_{2,+}=K^{(n)}_{++}(\theta_*,-\theta_0), \qquad R_- :=R_{1,-}+R_{2,-}=K^{(n)}_{++}(\theta_*,\theta_0),$$
where $K^{(n)}_{++}(\theta_*,\theta_0)$ is defined in \eqref{kelement}. We conclude that
$$\lim_{\theta_t \to \frac{i Q}{2}+\frac{ib}{2}} \lim_{\nu \to \nu_-} \int_0^{+\infty} d\sigma_s~\mathcal{C}_n\left[\substack{\theta_t\vspace{0.08cm} \\ \theta_{*}\;\;\;\theta_0};\substack{\nu\vspace{0.15cm} \\  \sigma_s}\right] \mathcal{B}(t) = K^{(n)}_{++}(\theta_*,\theta_0) B_+(t) + K^{(n)}_{++}(\theta_*,-\theta_0) B_-(t).$$ 
Recalling the expression (\ref{knm}) for $C_n$, this completes the proof also of the first component of \eqref{continuationodd}.
\end{proof}

\subsection{BPZ limit of $\mathcal{D}_n(t)$} 
Our next proposition shows that the confluent conformal blocks of the second kind $\mathcal{D}_n$ reduce to the degenerate confluent conformal blocks of the second kind $\boldsymbol{D}_{n}$ in the BPZ limit. 

\begin{proposition}[BPZ limit of $\mathcal{D}_n(t)$]\label{continuationCB}
Define $\nu_\pm=-\frac{\theta_*}{2}\pm \frac{ib}{2}$. The following BPZ limit holds:
\beq 
\lim_{\theta_t \to \frac{iQ}{2}+\frac{ib}{2}} \begin{pmatrix} \lim\limits_{\substack{\nu \to \nu_-}} {\mathcal D}_{n}\lb\substack{\theta_t\\ \theta_*};\nu;\theta_0;t\rb  \\
\lim\limits_{\substack{\nu \to \nu_+}}{\mathcal D}_{n}\lb\substack{\theta_t\\ \theta_*};\nu;\theta_0;t\rb\end{pmatrix}=\boldsymbol{D}_{n}(t), \qquad n = 1, 2, \dots,
\eeq
where $ {\mathcal D}_{n}\lb\substack{\theta_t\\ \theta_*};\nu;\theta_0;t\rb$ is defined in \eqref{dk} and $\boldsymbol{D}_{n}(t)$ given in \eqref{solutionsinf} is the solution of the confluent BPZ equation in the Stokes sector $\Omega_n$.
\end{proposition}
\begin{proof}
The proposition follows immediately from Proposition \ref{BPZlimitCnprop} and the identities \eqref{dk} and \eqref{conni}.
\end{proof}

\begin{remark}\upshape
It can also be argued directly that $\mathcal{D}_n(t) \to \boldsymbol{D}_{n}(t)$ in the BPZ limit as follows. First, it is easy to verify that the prefactor $e^{2i\pi(j-1)(\theta_*^2/2 -2\nu^2)}$ in (\ref{calDndef}) reduces to the prefactor $(e^{i\pi b^2} e^{2\pi b \theta_* \sigma_3})^{j-1}$ in (\ref{confluentlimits}) in the BPZ limit. Second, using the convergent series expansions for the four-point conformal blocks, it can be verified numerically to high precision that the following limits hold:
\beq \label{calFBPZlimit}\begin{split}
& \lim_{\theta_t \to \frac{iQ}{2}+\frac{ib}{2}} \begin{pmatrix} \lim\limits_{\substack{\sigma_t \to \theta_1 + \frac{ib}2}} z^{-2\Delta(\theta_t)}  \mathcal F\lb\substack{\theta_\infty\;\quad 
     \theta_{t}\\ \sigma_t \\ \theta_0\quad \theta_1};\substack{1-\frac1z}\rb \\
\lim\limits_{\substack{\sigma_t \to \theta_1-\frac{ib}2}} z^{-2\Delta(\theta_t)} \mathcal F\lb\substack{\theta_\infty\;\quad 
     \theta_{t}\\ \sigma_t \\ \theta_0\quad \theta_1};\substack{1-\frac1z}\rb\end{pmatrix}=\boldsymbol{F}^1(z), \\
& \lim_{\theta_t \to \frac{iQ}{2}+\frac{ib}{2}} \begin{pmatrix} \lim\limits_{\substack{\sigma_u \to \theta_\infty - \frac{ib}2}}  z^{-2\Delta(\theta_t)} \mathcal F\lb\substack{\theta_{1}\;\quad 
      \theta_{t}\\ \sigma_u \\ \theta_0 \quad \theta_\infty};\tfrac1z\rb \\
\lim\limits_{\substack{\sigma_u \to \theta_\infty+\frac{ib}2}} z^{-2\Delta(\theta_t)} \mathcal F\lb\substack{\theta_{1}\;\quad 
      \theta_{t}\\ \sigma_u \\ \theta_0 \quad \theta_\infty};\tfrac1z\rb \end{pmatrix}=\boldsymbol{F}^\infty(z).
\end{split} \eeq
Performing the change of variables 
$$\theta_1=\frac{\Lambda-\theta_*}2, \quad \theta_\infty = \frac{\Lambda-\theta_*}2, \quad \sigma_t = \frac{\Lambda}2-\nu, \quad \sigma_u = \frac{\Lambda}2+\nu,$$ 
the limits (\ref{calFBPZlimit}) imply that
\beq 
\lim_{\theta_t \to \frac{iQ}{2}+\frac{ib}{2}} \begin{pmatrix} \lim\limits_{\substack{\nu \to \nu_-}} \tilde{\mathcal{F}}^p(z,\Lambda,\nu) \\ \lim\limits_{\substack{\nu \to \nu_+}} \tilde{\mathcal{F}}^p(z,\Lambda,\nu) \end{pmatrix} = \tilde{\boldsymbol{F}}^p(z,\Lambda)
\eeq
for $p=1,\infty$; in other words, that $\tilde{\mathcal{F}}^1 \to \tilde{\boldsymbol{F}}^1$ and $\tilde{\mathcal{F}}^\infty \to \tilde{\boldsymbol{F}}^\infty$ in the BPZ limit. Comparing (\ref{confluentlimits}) and (\ref{calDndef}), we see that these limits imply the limit $\mathcal{D}_n(t) \to \boldsymbol{D}_{n}(t)$.
\end{remark}

\subsection{BPZ limit of the Stokes transformations}
Our next proposition provides the BPZ limit of the right-hand side of \eqref{stokestransform} in any odd Stokes sector. 

\begin{proposition}[BPZ limit of $\mathcal{S}_n$ for $n$ odd]\label{BPZodd}
For any integer $j \geq 1$ and any $t \in \Omega_{2j-1}$, it holds that
\beq \label{continuationoddstokes}
\lim_{\theta_t \to \frac{iQ}{2}+\frac{ib}{2}} \begin{pmatrix} \lim\limits_{\substack{\nu_{2j} \to \nu_-}}\displaystyle \int_{-\infty}^{+\infty} d\nu_{2j-1} ~ \mathcal{S}_{2j-1}\left[\substack{\theta_t\vspace{0.08cm} \\ \theta_{*}\;\;\;\theta_0};\substack{\nu_{2j}\vspace{0.15cm} \\  \nu_{2j-1}}\right] {\mathcal D}_{2j-1}\lb\substack{\theta_t\\ \theta_*};\nu_{2j-1};\theta_0;t\rb  \\
\lim\limits_{\substack{\nu_{2j} \to \nu_+}} \displaystyle \int_{-\infty}^{+\infty} d\nu_{2j-1} ~ \mathcal{S}_{2j-1}\left[\substack{\theta_t\vspace{0.08cm} \\ \theta_{*}\;\;\;\theta_0};\substack{\nu_{2j}\vspace{0.15cm} \\  \nu_{2j-1}}\right] {\mathcal D}_{2j-1}\lb\substack{\theta_t\\ \theta_*};\nu_{2j-1};\theta_0;t\rb\end{pmatrix}=S_{2j-1}\boldsymbol{D}_{2j-1}(t),
\eeq
where $S_{2j-1}$ are the odd Stokes matrices of the confluent BPZ equation given by (see \eqref{otherstokes} and \eqref{stokes})
\beq \label{S2jm1}
S_{2j-1} = e^{2\pi b (j-1) \theta_* \sigma_3} S_1 e^{-2\pi b (j-1) \theta_* \sigma_3}=\left(
\begin{array}{cc}
 1 & 0 \\
 -\frac{2 i \pi  e^{-4 \pi b(j-1) \theta_*}}{\Gamma \left(\frac12+ib(\theta_0-\theta_*)\right) \Gamma \left(\frac{1}{2}-i b (\theta_0+\theta_*)\right)} & 1 \\
\end{array}
\right),  \qquad j=1,2,\dots,\eeq
and $\boldsymbol{D}_{2j-1}(t)$ are the degenerate confluent conformal blocks of the second kind in the odd Stokes sectors defined in \eqref{solutionsinf}.
\end{proposition}
\begin{proof}
The proof is of the same complexity as the proof of Proposition \ref{BPZlimitCnprop} and we will only describe the main steps. We consider the Stokes kernel \eqref{stokessn} for the value $n=2j-1$ and split the prefactor \eqref{matp} into two parts:
$$\mathcal P_{2j-1}\left[\substack{\theta_t\vspace{0.08cm} \\ \theta_{*}\;\;\;\theta_0};\substack{\nu_{2j}\vspace{0.15cm} \\  \nu_{2j-1}}\right] := P_1(\nu_{2j},\theta_t,\theta_*) ~ \mathcal P'_{2j-1} \left[\substack{\theta_t\vspace{0.08cm} \\ \theta_{*}\;\;\;\theta_0};\substack{\nu_{2j}\vspace{0.15cm} \\  \nu_{2j-1}}\right],$$
where $P_1(\nu,\theta_t,\theta_*)$ is given in \eqref{P1} and $\mathcal P'_{2j-1}$ is defined by
\begin{equation}\label{pp2jm1}\begin{split}
\mathcal P'_{2j-1}\left[\substack{\theta_t\vspace{0.08cm} \\ \theta_{*}\;\;\;\theta_0};\substack{\nu_{2j}\vspace{0.15cm} \\  \nu_{2j-1}}\right]
=&\; e^{i\pi\lb \nu_{2j-1} \theta_0+\nu_{2j} \theta_t+\frac{iQ}2(\theta_*+\nu_{2j}-\nu_{2j-1})+\frac{\theta_*}2(\theta_0-\theta_t)\rb}\lb e^{2i\pi(j-1)} b\rb^{2\nu_{2j-1}^2-2\nu_{2j}^2}  
	 \\
& \times e^{i\pi \lb \frac{Q^2}4-\frac{iQ}2(\theta_0-\theta_t)-\frac{\theta_*}2(\nu_{2j-1}-\nu_{2j}-\frac{\theta_*}2)-\theta_0 \theta_t-\frac{\nu_{2j}^2+\nu_{2j-1}^2}2 \rb}  
	\\
& \times \prod_{\epsilon=\pm1} \frac{g_b \lb -\theta_0+\epsilon(\nu_{2j-1}-\frac{\theta_*}{2}) \rb g_b \lb -\theta_t+\epsilon(\nu_{2j-1}+\frac{\theta_*}{2}) \rb}{g_b \lb -\theta_0+\epsilon(\nu_{2j}-\frac{\theta_*}{2}) \rb},
\end{split}
\end{equation}
In the BPZ limit, the integration contours $\mathbb{R}$ and $\mathsf{S}$ are pinched by colliding poles. The pinching of the contour $\mathbb{R}$ is caused by the factor $\prod_{\epsilon=\pm1}g_b \lb -\theta_t+\epsilon(\nu_{2j-1}+\frac{\theta_*}{2}) \rb$. Indeed, this factor has two poles which cross $\mathbb{R}$ in the limit $\theta_t \to \frac{iQ}{2}+\frac{ib}{2}$: the pole of $g_b\lb \nu_{2j-1}+\frac{\theta_*}{2}-\theta_t \rb$ located at $\nu_{2j-1}=-\frac{\theta_*}{2}-\frac{iQ}{2}+\theta_t$, as well as the pole of $g_b\lb -\nu_{2j-1}-\frac{\theta_*}{2}-\theta_t \rb$ located at $\nu_{2j-1}=-\frac{\theta_*}2+\frac{iQ}{2}-\theta_t$. Therefore, we deform $\mathbb{R}$ into a new contour $\mathbb{R}'$ so that the integral over $\mathbb{R}$ turns into an integral over $\mathbb{R}'$ plus two contributions originating from the residues at $\nu_{2j-1}=-\frac{\theta_*}2\pm(\frac{iQ}{2}-\theta_t)$. Let us first consider the case $\nu_{2j} \to \nu_+=-\frac{\theta_*}{2}+\frac{ib}{2}$. 

\textbf{The case $\nu_{2j}\to\nu_+$.} 
The contour $\mathsf{S}$ in \eqref{stokessn} can be deformed into a suitable contour $\mathsf{S}'$ such that no pole in $x$ crosses $\mathsf{S}'$ in the limit $\nu_{2j}\to\nu_+$. For $\nu_{2j} = \nu_+$, we have
$$\mathcal I_{2j-1}\left[\substack{x\;\;\;\; \theta_t\vspace{0.08cm} \\ \theta_{*}\;\;\;\theta_0};\substack{\nu_+\vspace{0.15cm} \\  \nu_{2j-1}}\right]=e^{i\pi x \lb \nu_{2j-1}+\frac{\theta_*}{2}-iQ-\frac{ib}{2} \rb} \tfrac{ s_b \lb x+\theta_0-\theta_t \rb}{s_b \lb x+\theta_0-\frac{ib}{2}+\frac{iQ}{2} \rb} \tfrac{s_b \lb x-\theta_* \rb}{s_b \lb x+\nu_{2j-1}-\frac{\theta_*}{2}+\frac{iQ}{2}-\theta_t \rb}.$$
In the limit $\theta_t \to \frac{iQ}{2}+\frac{ib}{2}$, the pole of $s_b \lb x+\theta_0-\theta_t \rb$ at $x=-\frac{iQ}{2}+\theta_t-\theta_0$ crosses the contour of integration $\mathsf{S}'$ and collides with the pole of $s_b \big( x+\theta_0-\frac{ib}{2}+\frac{iQ}{2} \big)^{-1}$ at $x=\frac{ib}{2}-\theta_0$. Hence, similarly to what has been described in Figure \ref{fig:fig2}, we deform the contour $\mathsf{S}'$ downward past the moving pole so that the integral over $\mathsf{S}'$ turns into a sum of an integral over $\mathsf{S}''$ and a residue contribution from $x=-\frac{iQ}{2}+\theta_t-\theta_0$. We find
\begin{align}\nonumber
& \lim_{\theta_t \to \frac{i Q}{2}+\frac{ib}{2}} \lim_{\nu_{2j} \to \nu_+} \int_{-\infty}^{+\infty} d\nu_{2j-1} ~ \mathcal{S}_{2j-1}\left[\substack{\theta_t\vspace{0.08cm} \\ \theta_{*}\;\;\;\theta_0};\substack{\nu_{2j}\vspace{0.15cm} \\  \nu_{2j-1}}\right] \mathcal {D}_{2j-1}\lb\substack{\theta_t\\ \theta_*};\nu_{2j-1};\theta_0;t\rb=-4\pi^2 \lim\limits_{\theta_t \to \frac{iQ}{2}+\frac{ib}{2}}P_1(\nu_+,\theta_t,\theta_*) 
	\\ \label{XX}
&\resizebox{0.91\hsize}{!}{$\times \sum_{\epsilon=\pm1}\epsilon\underset{\nu_{2j-1}=-\frac{\theta_*}2+\epsilon(\theta_t-\frac{iQ}{2})}{\text{Res}} \left\{ \mathcal P'_{2j-1}\left[\substack{\theta_t\vspace{0.08cm} \\ \theta_{*}\;\;\;\theta_0};\substack{\nu_+ \vspace{0.15cm} \\  \nu_{2j-1}}\right] \underset{x=\theta_t-\frac{iQ}{2}-\theta_0}{\text{Res}} \mathcal I_{2j-1}\left[\substack{x\;\;\;\; \theta_t\vspace{0.08cm} \\ \theta_{*}\;\;\;\theta_0};\substack{\nu_+ \vspace{0.15cm} \\  \nu_{2j-1}}\right] \mathcal {D}_{2j-1}\lb\substack{\theta_t\\ \theta_*};\nu_{2j-1};\theta_0;t\rb \right\}.$}
\end{align}
After long but straightforward computations, we obtain
\begin{align*}\begin{split}
\zeta_+ := & -4\pi^2  \underset{\nu_{2j-1}=-\frac{\theta_*}2+\theta_t-\frac{iQ}{2}}{\text{Res}} \left\{ \mathcal P'_{2j-1}\left[\substack{\theta_t\vspace{0.08cm} \\ \theta_{*}\;\;\;\theta_0};\substack{\nu_+ \vspace{0.15cm} \\  \nu_{2j-1}}\right] \underset{x=\theta_t-\frac{iQ}{2}-\theta_0}{\text{Res}} \mathcal I_{2j-1}\left[\substack{x\;\;\;\; \theta_t\vspace{0.08cm} \\ \theta_{*}\;\;\;\theta_0};\substack{\nu_+ \vspace{0.15cm} \\  \nu_{2j-1}}\right]  \right\}
	 \\
 = &\; b^{\lb \theta_t-\frac{ib}{2}-\frac{iQ}{2} \rb \lb ib-iQ-2\theta_*+2\theta_t \rb}e^{i\pi\lb -\theta_t+\frac{ib}2+\frac{iQ}2 \rb \lb \theta_0+(4j-3)\theta_*+(\frac72-4j)\theta_t-ib(2j-\frac74)-iQ(\frac34-2j) \rb} 
	 \\
&  \times g_b \lb \frac{iQ}2-2\theta_t \rb g_b \lb \frac{ib}2-\theta_t \rb \tfrac{g_b \lb \frac{iQ}2 \rb}{g_b \lb \theta_t-\frac{ib}2 \rb} \displaystyle \prod_{\epsilon=\pm1} \tfrac{g_b \lb -\theta_0+\epsilon(\frac{iQ}{2}+\theta_*-\theta_t)\rb}{g_b \lb -\theta_0+\epsilon(\theta_*-\frac{ib}2) \rb}.
\end{split} \end{align*}
As shown by the first pair of $g_b$-functions, $\zeta_+$ has a double pole at $\theta_t=\frac{iQ}2+\frac{ib}2$, so after multiplication by $P_1(\nu_+,\theta_t,\theta_*)$, this yields a finite and nonzero result in the limit $\theta_t \to \frac{iQ}{2}+\frac{ib}2$. In fact, a calculation gives
\beq
\lim_{\theta_t \to \frac{iQ}{2}+\frac{ib}2} \zeta_+ P_1(\nu_+,\theta_t,\theta_*)=1,
\eeq
which agrees with the $(2,2)$ entry of the Stokes matrix $S_{2j-1}$ in \eqref{S2jm1}. The other residue reads
\begin{align*}\begin{split}
\zeta_- := &\; 4\pi^2\underset{\nu_{2j-1}=-\frac{\theta_*}2-\theta_t+\frac{iQ}{2}}{\text{Res}} \left\{ \mathcal P'_{2j-1}\left[\substack{\theta_t\vspace{0.08cm} \\ \theta_{*}\;\;\;\theta_0};\substack{\nu_+ \vspace{0.15cm} \\  \nu_{2j-1}}\right] \underset{x=\theta_t-\frac{iQ}{2}-\theta_0}{\text{Res}} \mathcal I_{2j-1}\left[\substack{x\;\;\;\; \theta_t\vspace{0.08cm} \\ \theta_{*}\;\;\;\theta_0};\substack{\nu_+ \vspace{0.15cm} \\  \nu_{2j-1}}\right] \right\}
	\\
= &\; b^{\lb \theta_t+\frac{ib}{2}-\frac{iQ}{2} \rb \lb -ib-iQ+2\theta_*+2\theta_t \rb} 
	\\ 
&\times e^{\left(\frac{1}{8} \pi  \left(i b^2 (8 j-7)-4 b (\theta_0-3 \theta_*+4 \theta_* j+i Q)+(Q+2 i \theta_t) (-4 (\theta_0+5 \theta_*)-22 \theta_t+16 j (\theta_*+\theta_t)+i (11-8 j) Q)\right)\right)}
	\\
& \times g_b \lb \frac{iQ}2-2\theta_t \rb g_b \lb \frac{ib}2-\theta_t \rb \tfrac{g_b \lb \frac{iQ}2 \rb}{g_b \lb \theta_t-\frac{ib}2 \rb} \frac{g_b \lb -\frac{iQ}2-\theta_0-\theta_*+\theta_t \rb}{g_b \lb \frac{iQ}2+\theta_0+\theta_*-\theta_t \rb} \displaystyle \prod_{\epsilon=\pm1} \tfrac{g_b\lb -\frac{iQ}2+\epsilon \theta_0+\theta_*+\theta_t \rb}{g_b \lb -\theta_0+\epsilon(\theta_*-\frac{ib}2) \rb}. \end{split}
\end{align*}
The same $g_b$-factors appear as before; thus $\zeta_-$ also has a double pole at $\theta_t=\frac{iQ}2+\frac{ib}2$ and we compute
\beq
\lim_{\theta_t \to \frac{iQ}{2}+\frac{ib}2} \zeta_- P_1(\nu_+,\theta_t,\theta_*)=-ib^{2ib\theta_*}e^{-4\pi b\theta_*(j-1)} \frac{g_b \lb \frac{ib}2+\theta_0+\theta_* \rb g_b \lb \frac{ib}2-\theta_0+\theta_* \rb}{g_b \lb -\frac{ib}2+\theta_0+\theta_* \rb g_b \lb -\frac{ib}2-\theta_0+\theta_* \rb}.
\eeq
Using the property \eqref{propgb} of the $g_b$-function, we finally find
\beq
\lim_{\theta_t \to \frac{iQ}{2}+\frac{ib}2} \zeta_- P_1(\nu_+,\theta_t,\theta_*)=-\frac{2i\pi e^{-4\pi b \theta_*(j-1)}}{\Gamma \lb \frac12+ib(\theta_0-\theta_*) \rb \Gamma \lb \frac12-ib(\theta_0+\theta_*) \rb},
\eeq
which agrees with the $(2,1)$ entry of the Stokes matrix $S_{2j-1}$. Since we have shown in Proposition \ref{continuationCB} that the BPZ limit of $\mathcal{D}_{2j-1}(t)$ is equal to $\boldsymbol{D}_{2j-1}(t)$, this completes the proof of the second component of \eqref{continuationoddstokes}. \\

\textbf{The case $\nu_{2j} \to \nu_-$.} For $\nu = \nu_-$, we have
$$\mathcal I_{2j-1}\left[\substack{x\;\;\;\; \theta_t\vspace{0.08cm} \\ \theta_{*}\;\;\;\theta_0};\substack{\nu_-\vspace{0.15cm} \\  \nu_{2j-1}}\right]=e^{i\pi x\lb \nu_{2j-1}+\frac{\theta_*}2+\frac{ib}2-iQ \rb} \tfrac{s_b \lb x+\theta_0-\theta_t\rb}{s_b \lb x+\theta_0+\frac{ib}2+\frac{iQ}2\rb} \tfrac{s_b \lb x-\theta_* \rb}{s_b \lb x+\frac{iQ}2+\nu_{2j-1}-\frac{\theta_*}2-\theta_t \rb}.$$
In the limit $\theta_t \to \frac{iQ}{2}+\frac{ib}2$, the poles of $s_b \lb x+\theta_0-\theta_t\rb$ located at $x_1:=-\theta_0-\frac{iQ}2+\theta_t$ and $x_2:=-\theta_0-\frac{iQ}2+\theta_t-ib$ cross the contour of integration and collide with the poles of $s_b\big( x+\frac{ib}2+\frac{iQ}2+\theta_0\big)$ at $x=+\frac{ib}2-\theta_0$ and $x=-\frac{ib}2-\theta_0$, respectively. We deform the contour $\mathsf{S}$ and pick up residues contributions from $x=x_1$ and $x=x_2$. We find the following analog of (\ref{XX}):
\begin{align*}
 \begin{split}
& \lim\limits_{\substack{\theta_t \to \frac{i Q}{2}+\frac{ib}{2} \\ \nu_{2j} \to \nu_-}} \displaystyle \int_{-\infty}^{+\infty} d\nu_{2j-1} ~ \mathcal{S}_{2j-1}\left[\substack{\theta_t\vspace{0.08cm} \\ \theta_{*}\;\;\;\theta_0};\substack{\nu_{2j}\vspace{0.15cm} \\  \nu_{2j-1}}\right] \mathcal {D}_{2j-1}\lb\substack{\theta_t\\ \theta_*};\nu_{2j-1};\theta_0;t\rb=-4\pi^2 \lim_{\theta_t \to \frac{iQ}{2}+\frac{ib}{2}}\mathcal P_1(\nu_-,\theta_t,\theta_*) 	 \\
&\times \displaystyle \sum_{\epsilon=\pm1}\epsilon\underset{\nu_{2j-1}=-\frac{\theta_*}2+\epsilon(\theta_t-\frac{iQ}{2})}{\text{Res}} \left\{ \mathcal P'_{2j-1}\left[\substack{\theta_t\vspace{0.08cm} \\ \theta_{*}\;\;\;\theta_0};\substack{\nu_- \vspace{0.15cm} \\  \nu_{2j-1}}\right] \displaystyle \sum_{i=1,2} \underset{x=x_i}{\text{Res}}\lb \mathcal I_{2j-1}\left[\substack{x\;\;\;\; \theta_t\vspace{0.08cm} \\ \theta_{*}\;\;\;\theta_0};\substack{\nu_- \vspace{0.15cm} \\  \nu_{2j-1}}\right] \rb\mathcal {D}_{2j-1}\lb\substack{\theta_t\\ \theta_*};\nu_{2j-1};\theta_0;t\rb \right\}.
\end{split}
\end{align*}
Letting $X_j$, $j = 1,2,$ denote the two residue contributions
$$X_j = -2i\pi \underset{x=x_j}{\text{Res}} \mathcal I_{2j-1}\left[\substack{x\;\;\;\; \theta_t\vspace{0.08cm} \\ \theta_{*}\;\;\;\theta_0};\substack{\nu_- \vspace{0.15cm} \\  \nu_{2j-1}}\right], \qquad j = 1,2,$$
we obtain
\begin{equation*} \begin{split}
& X_1 = e^{i\pi \lb \frac{ib}{2}-iQ+\nu_{2j-1}+\frac{\theta_*}2\rb \lb \theta_t-\theta_0-\frac{iQ}2\rb} \tfrac{\operatorname{sech}{\lb \pi b \theta_t \rb}}{2} \tfrac{s_b \lb \theta_t-\frac{iQ}{2}-\theta_0-\theta_* \rb}{s_b \lb \theta_t-\frac{ib}2 \rb s_b \lb \nu_{2j-1}-\theta_0-\frac{\theta_*}2 \rb}, 
	\\
& X_2  = e^{i\pi \lb \frac{ib}{2}-iQ+\nu_{2j-1}+\frac{\theta_*}2\rb \lb \theta_t-\theta_0-ib-\frac{iQ}2\rb} \tfrac{\operatorname{sech}{\lb \pi b (-\frac{ib}2-\frac{iQ}2)\rb}}2 \tfrac{s_b \lb \theta_t-ib-\frac{iQ}{2}-\theta_0-\theta_* \rb}{s_b \lb \theta_t-\frac{ib}2 \rb s_b \lb \nu_{2j-1}-\theta_0-\frac{\theta_*}2-ib \rb}.
\end{split} \end{equation*}
Defining $R_{j,+}$ and $R_{j,-}$ for $j = 1,2$ by
\begin{align*}
R_{j,\pm} = & \mp 2i\pi  \lim_{\theta_t \to \frac{iQ}{2}+\frac{ib}{2}}P_1(\nu_-,\theta_t,\theta_*)\underset{\nu_{2j-1}=-\frac{\theta_*}2+\theta_t \mp \frac{iQ}2}{\text{Res}} \left\{\mathcal P'_{2j-1} \left[\substack{\theta_t\vspace{0.08cm} \\ \theta_{*}\;\;\;\theta_0};\substack{\nu_- \vspace{0.15cm} \\  \nu_{2j-1}}\right]X_j\right\},
\end{align*}
long computations yield
\begin{equation*}\begin{split}
R_{1,+} = & - R_{2,+} = e^{i\pi b^2}e^{\pi b (\theta_0+(4j-3)\theta_*)}\frac{\operatorname{csc}{(\pi b^2)}}{2}\frac{\Gamma \lb \frac12+ib(\theta_0-\theta_*)\rb}{\Gamma \lb \frac12+ib(\theta_0+\theta_*)\rb},
	\\
R_{1,-} = & -i e^{i\pi b^2}e^{\pi b(\theta_0+\theta_*)}\operatorname{csc}{(\pi b^2)}\operatorname{cosh}{\lb \pi b(\theta_0+\theta_*)\rb},
	\\
R_{2,-} = &\; i e^{\pi b(\theta_0+\theta_*)}\operatorname{csc}{(\pi b^2)} \operatorname{cosh}{\lb \pi b(\theta_0+\theta_*+ib)\rb}.
\end{split}\end{equation*}
Since $R_{1,-}+R_{2,-}=1$ and $R_{1,+}+R_{2,+}=0$, we recover the first row of the Stokes matrix $S_{2j-1}$. The proof is complete.
\end{proof}

The next proposition provides the BPZ limit of the right-hand side of \eqref{stokestransform} in any even Stokes sector. The proof involves a different set of colliding poles, but is otherwise similar to the proof of Proposition \ref{BPZodd}; it will therefore be omitted.

\begin{proposition}[BPZ limit of $\mathcal{S}_n$ for $n$ even]\label{BPZeven}
For any integer $j \geq 1$ and any $t \in \Omega_{2j}$, it holds that
\beq \label{continuationstokespair}
\lim_{\theta_t \to \frac{iQ}{2}+\frac{ib}{2}} \begin{pmatrix} \lim\limits_{\substack{\nu_{2j+1} \to \nu_-}}\displaystyle \int_{-\infty}^{+\infty} d\nu_{2j} ~ \mathcal{S}_{2j}\left[\substack{\theta_t\vspace{0.08cm} \\ \theta_{*}\;\;\;\theta_0};\substack{\nu_{2j+1}\vspace{0.15cm} \\  \nu_{2j}}\right] {\mathcal D}_{2j}\lb\substack{\theta_t\\ \theta_*};\nu_{2j};\theta_0;t\rb  \\
\lim\limits_{\substack{\nu_{2j+1} \to \nu_+}} \displaystyle \int_{-\infty}^{+\infty} d\nu_{2j} ~ \mathcal{S}_{2j}\left[\substack{\theta_t\vspace{0.08cm} \\ \theta_{*}\;\;\;\theta_0};\substack{\nu_{2j+1}\vspace{0.15cm} \\  \nu_{2j}}\right] {\mathcal D}_{2j}\lb\substack{\theta_t\\ \theta_*};\nu_{2j};\theta_0;t\rb\end{pmatrix}=S_{2j}\boldsymbol{D}_{2j}(t),
\eeq
where $S_{2j}$ are the even Stokes matrices of the confluent BPZ equation given by (see \eqref{otherstokes} and \eqref{stokes})
\beq \label{s2j}
S_{2j}= e^{2\pi b (j-1) \theta_* \sigma_3}~S_2~e^{-2\pi b (j-1) \theta_* \sigma_3}= \begin{pmatrix} 1 & -\frac{2i\pi e^{2\pi b(2j-1)\theta_*}}{\Gamma\lb \frac12-ib(\theta_0-\theta_*) \rb  \Gamma \lb \frac12+ib(\theta_0+\theta_*)\rb} \\ 0 & 1 \end{pmatrix},~j=1,2,\dots,
\eeq
and $\boldsymbol{D}_{2j}(t)$ are the degenerate confluent conformal blocks of the second kind in the even Stokes sectors defined in \eqref{solutionsinf}.
\end{proposition}

\subsection{Summary}
The results of this section can be summarized as follows.

\begin{corollary}
In the BPZ limit, the connection formula \eqref{dk} of Theorem \ref{mainth1} and the Stokes transformation \eqref{stokestransform} of Theorem \ref{mainth2} reduce to the connection formula \eqref{conni} and the Stokes formula \eqref{stokesrelation} for the confluent BPZ equation, respectively. Schematically, this can be expressed as
\begin{center}
\begin{tikzcd}[row sep=0.5cm, column sep = 2.5cm]
\mathcal{D}_n(t)= \mathcal{C}_n \mathcal{B}(t)  \ar[d,"\text{BPZ limit}"] & \mathcal{D}_{n+1}(t)= \mathcal{S}_n \mathcal{D}_n(t) \ar[d,"\text{BPZ limit}"] \\ 
\boldsymbol{D}_n(t)=C_n \boldsymbol{B}(t) & \boldsymbol{D}_{n+1}(t)=S_n \boldsymbol{D}_n(t)
\end{tikzcd}
\end{center}
\end{corollary}
\begin{proof} 
See Propositions \ref{BPZlimitCnprop}-\ref{BPZeven}.
\end{proof}

\section{Conclusions and perspectives} \label{section8}
In this article we have constructed the confluent conformal blocks of the second kind. We also constructed the Stokes transformations which map such blocks in one Stokes sector to another. Both the confluent conformal blocks and the Stokes transformations were found by taking suitable confluent limits of the crossing transformations of the four-point Virasoro conformal blocks. We explicitly verified that in the BPZ limit the constructed blocks and the associated Stokes transformations reduce to solutions of the confluent BPZ equation and its Stokes matrices, respectively. 

An interesting problem is to combine the holomorphic and anti-holomorphic confluent conformal blocks with an integration measure to construct confluent Liouville correlation functions. Such a construction should be possible by applying the confluence procedure that we have described to the Liouville correlations functions built from the $s$-, $t$-, and $u$-channel conformal blocks. The result is expected to be invariant under the generalized fusion and Stokes transformations that we have constructed in this article. Hence this generalized crossing symmetry would be translated into orthogonality relations satisfied by the confluent fusion and Stokes kernels. Moreover, there exist close connections between quantum Teichmuller theory and Liouville theory \cite{TN,T14}. The collision of holes in Teichmuller theory was studied in \cite{CM,CMR}. It would be interesting to understand the connections between the above two subjects after taking the confluent limit.


Finally, the approach that we have developed can in principle be adopted to construct also confluent conformal blocks with irregular singularities of rank $r>1$.

\appendix
\section{Two special functions}\label{appendix} 
Equation (\ref{fusion01}) expresses the Virasoro fusion kernel in terms of two special functions $s_b(z)$ and $g_b(z)$. These functions are defined by
\begin{equation}\label{defsb}
s_b(z)=\operatorname{exp}{\left[  i \int_0^\infty \frac{dy}{y} \left(\frac{\operatorname{sin}{2yz}}{2\operatorname{sinh}{b^{-1}y}\operatorname{sinh}{by}}-\frac{z}{y}\right)\right]}, \qquad |\im z|<\frac{Q}{2},
\end{equation}
and
\beq \label{gb}
g_b(z)=\operatorname{exp}{\left\{ \int_0^\infty \frac{dt}{t}\left[\frac{e^{2i z t}-1}{4 \operatorname{sinh}{b t} \operatorname{sinh}{b^{-1} t}}+\frac{1}{4}z^2 \lb e^{-2bt}+e^{-\frac{2t}b}\rb-\frac{iz}{2t} \right]\right\}}.
\eeq
They are related to the functions $G$ and  $E$ defined in \cite[Eq. (A.3)]{R1999} and \cite[Eq. (A.43)]{R1999} as follows:
\begin{align}\label{sbgbGE}
s_b(z) = G(b, b^{-1}; z), \qquad g_b(z) = \frac{1}{E(b, b^{-1}; -z)}.
\end{align}
Both $s_b$ and  $g_b$ are obviously invariant under the exchange of $b$ and $b^{-1}$. In what follows, we list some further properties of these functions which follow from (\ref{sbgbGE}) and the results of \cite{R1999}. 

The function $g_b(z)$ satisfies the difference equations
\beq \label{propgb}
\frac{g_b \lb z+\frac{ib}{2}\rb}{g_b\lb z-\frac{ib}{2}\rb}=\frac{b^{-ibz}\sqrt{2\pi}}{\Gamma \lb \frac{1}{2}-ibz \rb}, \qquad
\frac{g_b \lb z+\frac{i}{2b}\rb}{g_b\lb z-\frac{i}{2b}\rb}=\frac{b^{-\frac{iz}{b}}\sqrt{2\pi}}{\Gamma \lb \frac{1}{2}- \frac{iz}{b} \rb}.
\eeq
It has no zeros, but it has simple poles located at 
\beq\label{polegb}
z_{k,l}=-\frac{i Q}{2} -i k b-il b^{-1}, \qquad k,l = 1, 2, \dots.
\eeq
Moreover, 
\beq \label{resgb}
\underset{z=-\frac{iQ}{2}}{\text{Res}}  g_b(z) =\frac{i}{2\pi} g_b\lb \frac{i Q}{2}\rb.
\eeq
As $z \to \infty$ with $\text{Re}(z)\geq0$, $g_b$ satisfies\footnote{Comparing the definition of $g_b(z)$ with \cite[Eq. (A.3)]{CGMP}, we have $g_b(iz - \frac{iQ}2)=\Gamma_b(z)$.} (see \cite[Eq. (A.23)]{CGMP})
\beq \label{asympgammab}
\operatorname{log}{g_b \lb i z-\frac{iQ}{2} \rb} = -\frac{1}{2}z^2 \operatorname{log}{z}+\frac{3}{4}z^2+\frac{Q}{2}z\operatorname{log}{z}-\frac{Q}{2}z-\frac{ Q^2+1}{12}\operatorname{log}{z}+\mathcal{O}(z^0).
\eeq

The function $s_b$ is related to $g_b$ by $s_b(z)=\frac{g_b(z)}{g_b(-z)}$.
It satisfies the difference equations
\begin{equation}\label{differencesb}
\frac{s_b(z+\frac{ib}{2})}{s_b(z-\frac{ib}{2})}=2\operatorname{cosh}{\pi b z}, \qquad \frac{s_b(z+\frac{i}{2b})}{s_b(z-\frac{i}{2b})}=2\operatorname{cosh}{\frac{\pi z}{b}},
\end{equation}
and it extends to a meromorphic function of $z \in \mathbb{C}$ with simple poles and zeros located at
\begin{equation}\label{polesb}
\begin{split}
&z_{k,l}=-\frac{i Q}{2} -i k b-il b^{-1}, \qquad k,l = 1, 2, \dots, \qquad (\text{poles}),\\
&z_{k,l}=\frac{i Q}{2}+i k b +il b^{-1}, \qquad k,l = 1, 2, \dots, \qquad (\text{zeros}).
\end{split}
\end{equation}
The residue at $z=-iQ/2$ is given by
\beq \label{ressb}
\underset{z=-\frac{iQ}{2}}{\text{Res}} s_b(z) =\frac{i}{2\pi},
\eeq
and the following asymptotic formulas hold as $z \to \infty$:
\beq \label{asympsb}
s_b(z) \sim \begin{cases} e^{-\frac{i \pi}{2}\lb z^2+\frac{1}{12}(b^2+b^{-2})\rb}, & |z|\to \infty, ~ |\arg z| < \frac{\pi}{2}, \\ 
e^{\frac{i \pi}{2}\lb z^2+\frac{1}{12}(b^2+b^{-2})\rb}, & |z|\to \infty, ~ |\arg z| > \frac{\pi}{2}. \end{cases}
\eeq

\bigskip
\noindent
{\bf Acknowledgement} {\it J.R. would like to thank Oleg Lisovyy for illuminating discussions at early stages of this project and acknowledges support from the European Research Council, Grant Agreement No. 682537. J.L. acknowledges support from the European Research Council, Grant Agreement No. 682537, the Swedish Research Council, Grant No. 2015-05430, the G\"oran Gustafsson Foundation, and the Ruth and Nils-Erik Stenb\"ack Foundation. }

\end{document}